\documentclass{article}
\usepackage[T1]{fontenc}
\usepackage[margin=1in]{geometry}
\usepackage{amsfonts,amsmath,amsthm,amssymb}
\usepackage{mathtools}  %

\usepackage{cite}

\mathtoolsset{centercolon}
\usepackage{xfrac,nicefrac}
\usepackage{mathdots}
\usepackage{mleftright}  %
\let\left\mleft
\let\right\mright

\usepackage{xspace}
\xspaceaddexceptions{]\}}  %
\usepackage{regexpatch}

\usepackage{bm,bbm,dsfont}  %
\usepackage{caption}
\usepackage[normalem]{ulem}
\usepackage{enumitem}

\usepackage{graphicx}
\usepackage{float}
\usepackage{subcaption}  %
\usepackage{tcolorbox}
\usepackage{tikz}
\usetikzlibrary{decorations.pathreplacing}
\usetikzlibrary{calc}
\usetikzlibrary{positioning}
\usetikzlibrary{arrows.meta}

\usepackage[linesnumbered,boxed,ruled,vlined,algo2e]{algorithm2e}

\SetCommentSty{mycommfont}

\makeatletter
\SetKwComment{tcp}{\hspace{1em}// }{}
\SetKwComment{tcpNoSpace}{// }{}
\makeatother

\usepackage{thmtools,thm-restate}
\usepackage[colorlinks,citecolor=blue,linkcolor=blue,urlcolor=red]{hyperref}

\theoremstyle{plain}

\newtheorem{theorem}{Theorem}[section]  %
\newtheorem{lemma}[theorem]{Lemma}

\newtheorem{observation}[theorem]{Observation}
\newtheorem{proposition}[theorem]{Proposition}
\newtheorem{corollary}[theorem]{Corollary}

\newtheorem{claim}[theorem]{Claim}

\theoremstyle{definition}  %
\newtheorem{definition}[theorem]{Definition}



\newenvironment{proofof}[1]{\begin{proof}[Proof of #1]}{\end{proof}}
\newenvironment{proofsketch}{\begin{proof}[Proof Sketch]}{\end{proof}}

\usepackage[capitalise]{cleveref}
\crefname{algocf}{Algorithm}{Algorithms}
\Crefname{algocf}{Algorithm}{Algorithms}
\crefname{claim}{Claim}{Claims}
\Crefname{claim}{Claim}{Claims}
\crefname{observation}{Observation}{Observations}
\Crefname{observation}{Observation}{Observations}
\crefname{fact}{Fact}{Facts}
\Crefname{fact}{Fact}{Facts}

\usepackage{xurl}

\providecommand{\Description}[2][]{}

\newfloat{Distribution}{htbp}{loa}
\crefname{Distribution}{Distribution}{Distributions}
\Crefname{Distribution}{Distribution}{Distributions}
\newfloat{Protocol}{htbp}{loa}
\crefname{Protocol}{Protocol}{Protocols}
\Crefname{Protocol}{Protocol}{Protocols}

\SetKwProg{Function}{Function}{:}{}
\let\Fn\Function
\SetKwProg{CodeBlock}{}{}{}
\SetKwProg{Repeat}{Repeat}{:}{}

\DeclarePairedDelimiter{\floor}{\lfloor}{\rfloor}

\DeclarePairedDelimiter{\bk}{(}{)}
\DeclarePairedDelimiter{\Bk}{[}{]}
\DeclarePairedDelimiter{\BK}{\{}{\}}

\DeclarePairedDelimiterX\mysetbase[2]{\lbrace}{\rbrace}{#1\,\delimsize\vert\,#2}
\NewDocumentCommand{\myset}{sO{}m m}{%
  \IfBooleanTF{#1}%
  {\mysetbase*{#3}{#4}}%
  {\mysetbase[#2]{#3}{#4}}%
}

\DeclareMathOperator*{\E}{\mathbb{E}}

\let\Pr\PrAux
\DeclareMathOperator{\poly}{poly}

\DeclareMathOperator{\polylog}{polylog}

\DeclareMathOperator*{\ind}{\mathbbm{1}}

\renewcommand{\tilde}{\widetilde}

\newcommand{\defeq}{\coloneqq}
\newcommand{\eps}{\varepsilon}

\newcommand{\N}{\mathbb{N}}

\renewcommand{\l}{\ell}
\renewcommand{\emptyset}{\varnothing}
\renewcommand{\epsilon}{\eps}

\newcommand{\Martin}{\textup{Mart\'in}\xspace}

\newcommand{\defn}[1]{\emph{\boldmath\textbf{#1}}}

\newcommand{\numberthis}{\addtocounter{equation}{1}\tag{\theequation}}

\makeatletter
\xpatchcmd\thmt@restatable{%
  \csname #2\@xa\endcsname\ifx\@nx#1\@nx\else[{#1}]\fi
}{%
  \ifthmt@thisistheone
  \csname #2\@xa\endcsname\ifx\@nx#1\@nx\else[{#1}]\fi
  \else
  \csname #2\@xa\endcsname[{Restated}]
  \fi}{}{}
\makeatother

\makeatletter
\let\oldparagraph\paragraph
\renewcommand{\paragraph}[1]{%
  \oldparagraph{\boldmath #1}%
}
\makeatother

\makeatletter
\let\save@mathaccent\mathaccent
\newcommand*\if@single[3]{%
  \setbox0\hbox{${\mathaccent"0362{#1}}^H$}%
  \setbox2\hbox{${\mathaccent"0362{\kern0pt#1}}^H$}%
  \ifdim\ht0=\ht2 #3\else #2\fi
}
\newcommand*\rel@kern[1]{\kern#1\dimexpr\macc@kerna}
\newcommand*\widebar[1]{\@ifnextchar^{{\wide@bar{#1}{0}}}{\wide@bar{#1}{1}}}
\newcommand*\wide@bar[2]{\if@single{#1}{\wide@bar@{#1}{#2}{1}}{\wide@bar@{#1}{#2}{2}}}
\newcommand*\wide@bar@[3]{%
  \begingroup
  \def\mathaccent##1##2{%
    \let\mathaccent\save@mathaccent
    \if#32 \let\macc@nucleus\first@char \fi
    \setbox\z@\hbox{$\macc@style{\macc@nucleus}_{}$}%
    \setbox\tw@\hbox{$\macc@style{\macc@nucleus}{}_{}$}%
    \dimen@\wd\tw@
    \advance\dimen@-\wd\z@
    \divide\dimen@ 3
    \@tempdima\wd\tw@
    \advance\@tempdima-\scriptspace
    \divide\@tempdima 10
    \advance\dimen@-\@tempdima
    \ifdim\dimen@>\z@ \dimen@0pt\fi
    \rel@kern{0.6}\kern-\dimen@
    \if#31
    \overline{\rel@kern{-0.6}\kern\dimen@\macc@nucleus\rel@kern{0.4}\kern\dimen@}%
    \advance\dimen@0.4\dimexpr\macc@kerna
    \let\final@kern#2%
    \ifdim\dimen@<\z@ \let\final@kern1\fi
    \if\final@kern1 \kern-\dimen@\fi
    \else
    \overline{\rel@kern{-0.6}\kern\dimen@#1}%
    \fi
  }%
  \macc@depth\@ne
  \let\math@bgroup\@empty \let\math@egroup\macc@set@skewchar
  \mathsurround\z@ \frozen@everymath{\mathgroup\macc@group\relax}%
  \macc@set@skewchar\relax
  \let\mathaccentV\macc@nested@a
  \if#31
  \macc@nested@a\relax111{#1}%
  \else
  \def\gobble@till@marker##1\endmarker{}%
  \futurelet\first@char\gobble@till@marker#1\endmarker
  \ifcat\noexpand\first@char A\else
  \def\first@char{}%
  \fi
  \macc@nested@a\relax111{\first@char}%
  \fi
  \endgroup
}
\makeatother

\newcommand{\DefineTerm}[2]{\newcommand{#1}{\ifmmode #2 \else $#2$\xspace \fi}}
\newcommand{\DefineInstruction}[2]{\newcommand{#1}{\ifmmode \textup{\texttt{#2}} \else \textup{\texttt{#2}}\xspace \fi}}
\newcommand{\BoldHeading}[1]{\textbf{\boldmath{}#1}}

\SetKwBlock{Loop}{loop}{}
\SetKwFor{WithProb}{with probability}{\!:}{}
\newcommand{\Break}{\textbf{break}}

\DefineInstruction{\Load}{load}
\DefineInstruction{\Store}{store}
\DefineInstruction{\StoreRandom}{storeRand}
\DefineInstruction{\StoreRand}{storeRand}
\DefineInstruction{\CAS}{CAS}
\DefineInstruction{\FetchAndAdd}{fetch-and-add}
\DefineInstruction{\Nop}{nop}

\DefineTerm{\Write}{\textup{\textbf{\textsf{write}}}}
\DefineTerm{\Read}{\textup{\textbf{\textsf{read}}}}
\DefineTerm{\Set}{\textup{\textbf{\textsf{set}}}}
\DefineTerm{\opCAS}{\textup{\textbf{\textsf{CAS}}}}
\newcommand{\opFont}[1]{ \textup{\textbf{\textsf{#1}}}}

\newcommand{\D}{\mathcal{D}}

\newcommand{\Latency}{\textup{latency}}

\DefineInstruction{\instR}{R}
\DefineInstruction{\instW}{W}
\DefineInstruction{\instS}{S}
\DefineInstruction{\instC}{C}
\DefineInstruction{\instRR}{R'}
\DefineInstruction{\instWW}{W'}
\DefineInstruction{\instSS}{S'}

\DefineInstruction{\History}{hist}
\DefineInstruction{\Randomness}{rands}

\newcommand{\Unif}{\textup{Unif}}

\newcommand{\HistTrue}{\textup{\texttt{hist}}}
\newcommand{\HistExt}{\textup{\texttt{hist}}^{\textup{\textsf{+}}}}

\newcommand{\coin}{\xi}
\newcommand{\CoinSum}{\Xi}

\newcommand{\xNew}{x_{\textup{new}}}
\newcommand{\xOld}{x_{\textup{old}}}
\newcommand{\xNow}{x_{\textup{now}}}
\newcommand{\fNew}{f_{\textup{new}}}
\newcommand{\fOld}{f_{\textup{old}}}
\newcommand{\fNow}{f_{\textup{now}}}
\newcommand{\xExpect}{x_{\textup{exp}}}
\newcommand{\xExp}{x_{\textup{exp}}}

\newcommand{\ASlow}{A^{\textup{slow}}}
\newcommand{\AAbort}{A^{\textup{abort}}}
\newcommand{\AStay}{A^{\textup{stay}}}
\newcommand{\ATransfer}{A^{\textup{transition}}}
\newcommand{\ANew}{A^{\textup{new}}}
\newcommand{\AOld}{A^{\textup{old}}}
\newcommand{\ACont}{A^{\textup{cont}}}

\newcommand{\pSlow}{p^{\textup{slow}}}
\newcommand{\pAbort}{p^{\textup{abort}}}
\newcommand{\pStay}{p^{\textup{stay}}}
\newcommand{\pTransfer}{p^{\textup{transition}}}
\newcommand{\pNew}{p^{\textup{new}}}
\newcommand{\pOld}{p^{\textup{old}}}
\newcommand{\pCont}{p^{\textup{cont}}}

\newcommand{\phiThreshold}{\phi^{\textup{thr}}}
\newcommand{\TThreshold}{T^{\textup{thr}}}
\newcommand{\Falling}{\texttt{falling}}
\newcommand{\Rising}{\texttt{rising}}

\newcommand{\yOld}{y_{\textup{old}}}
\newcommand{\yNow}{y_{\textup{now}}}

\newcommand{\kFall}{k_{\textup{fall}}}
\newcommand{\kRise}{k_{\textup{rise}}}

\newcommand{\SharedCell}{\textup{\textsf{C}}\xspace}
\newcommand{\CellW}{\textup{\textsf{W}}\xspace}
\newcommand{\CellC}{\textup{\textsf{C}}\xspace}
\newcommand{\CellR}{\textup{\textsf{R}}\xspace}

\newcommand{\SLarge}{S_{\text{large}}}
\newcommand{\SSmall}{S_{\text{small}}}

\newcommand{\st}[1]{\texttt{\textup{#1}}}

\title{Fast Concurrent Primitives Despite Contention}

\author{
  Michael A. Bender\thanks{Stony Brook University. Email: \texttt{bender@cs.stonybrook.edu}.}
  \and Guy E. Blelloch\thanks{Carnegie Mellon University. Email: \texttt{blelloch@cs.cmu.edu}.}
  \and \Martin Farach-Colton\thanks{New York University. Email: \texttt{martin.farach-colton@nyu.edu}.}
  \and Yang Hu\thanks{Tsinghua University. Email: \texttt{y-hu22@mails.tsinghua.edu.cn}.}
  \and Rob Johnson\thanks{VMware Research. Email: \texttt{rob@robjohnson.io}.}
  \and Rotem Oshman\thanks{Tel Aviv University and New York University. Email: \texttt{roshman@tau.ac.il}.}
  \and Renfei Zhou\thanks{Carnegie Mellon University. Email: \texttt{renfeiz@andrew.cmu.edu}.}
}
\date{}

\begin{document}

\maketitle

\begin{abstract}
We study the problem of constructing concurrent objects in a setting where $P$ processes run in parallel and interact through a shared memory that is subject to write contention.
Our goal is to transform hardware primitives that are subject to write contention into ones that handle contention gracefully.

We give contention-resolution algorithms for several basic primitives, and analyze them under a relaxed, roughly-synchronous stochastic scheduler, where processes run at roughly the same rate up to a constant factor with high probability.
Specifically, we construct read/write registers and CAS registers that have latency $O(\log P)$ w.h.p.\ under our scheduler model,
using $O(1)$ hardware read/write registers and, in the case of our CAS construction, one hardware CAS register.
Our algorithms guarantee performance even when their operations are invoked by an adaptive adversary that is able to see the entire history of operations so far, including their timing and return values.
This allows them to be used as building blocks inside larger programs; using this compositionality property, we obtain several other constructions (LL/SC, fetch-and-increment, bounded max registers, and counters).

To complement our constructions, we give a trade-off showing that even under a perfectly synchronous schedule and even if each process only executes one operation,
any algorithm that implements any of the primitives that we consider, uses space $M$, and has latency at most $L$ with high probability must have \emph{expected} latency at least $\Omega(\log_{ML} P)$.
\end{abstract}

\section{Introduction}
\label{sec:intro}

\defn{Memory contention} is a first-order performance issue in concurrent algorithms and data structures.
Contention occurs when multiple processes attempt to modify the same memory location simultaneously:
In today's systems,
writes to the same memory location are sequentialized by the hardware, meaning that only one request can be served at a time and the others have to wait.

Several lines of work have tackled the contention problem:
in the context of concurrent algorithms, \cite{DHW97,CNM94,CM10,EHN12,HS03,ACAH16} consider the \emph{stalls model},
which charges processes for attempting to access the same memory address at the same time---they each incur a \emph{stall}, which is counted in the step complexity of the algorithm. This line of work assumes an \emph{adversarial worst-case scheduler}, and as a result, very strong lower bounds are known even for simple primitives such as fetch-and-increment~\cite{DHW97,EHN12}.
In the context of parallel algorithms,
the \emph{Queue-Read Queue-Write machine model (QRQW PRAM)}~\cite{GibbonsMR96,GibbonsMR98a,GibbonsMR98b}
captures contention by queuing accesses to each memory location.
Although one variant~\cite{GibbonsMR98b} assumes an asynchronous model for the purpose
of correctness, these works all analyze the performance of their algorithms under a \emph{synchronous} greedy scheduler, where each process is scheduled as soon as it is ready to take its next step.
Other work~\cite{BDB17} considers a model where instructions are scheduled asynchronously by a worst-case adversarial scheduler, but once an instruction starts, it must proceed synchronously, waiting in a queue until prior updates to the same location complete.
The focus of~\cite{BDB17} is on analyzing exponential back-off, and given the strong adversary, the bounds obtained are weak.

In this work, we present contention resolution algorithms
for several basic synchronization primitives, including read/write registers and compare-and-swap (CAS).
Our goal is to transform hardware primitives that are subject to contention into corresponding ones that handle contention gracefully.
We analyze our algorithms within a roughly-synchronous probabilistic model that subsumes both the synchronous model of~\cite{GibbonsMR96,GibbonsMR98a,GibbonsMR98b} and the semi-synchronous concurrency model (e.g.,~\cite{DLS88,AM94,ADLS94,ELMS05,T07,HK06}),
and we show that our algorithms have polylogarithmic latency in the number of processes with high probability.

While most prior work (e.g.,~\cite{GibbonsMR96,GibbonsMR98a,GibbonsMR98b})
aims to design algorithms that \emph{avoid creating} contention, here our goal is to resolve contention at the level of the synchronization primitives themselves, effectively decoupling contention resolution from algorithm design.
To that end, we prove a composition theorem showing that the primitives we construct can replace their hardware counterparts in a black box manner in a large class of algorithms.
This allows designers to design efficient algorithms \emph{without worrying about contention}, and then plug in our primitives to obtain algorithms that perform well on hardware that does suffer from the effects of contention, provided the schedule is roughly synchronous.

Before presenting our results, we give a brief overview of the scheduler model under which we analyze our algorithms.

\subsection{Overview of Our Scheduler Model}
Many researchers have acknowledged the gap between fully synchronous models and fully asynchronous and adversarial models, and have developed models that try to bridge this gap.
In this paper, we analyze the performance of our algorithms under a relaxed semi-synchronous model,
where processes advance at roughly the same pace up to a constant factor.
We call the model \emph{stochastic Concurrent-Read Queue-Write (CRQW)}.
It is more permissive than several previously suggested models (see Section~\ref{sec:compare} below),
and is general enough to capture many real-world phenomena, such as cache-miss latencies and arbitrary behavior by the rest of the program that uses our data structure.
It does this by giving the scheduler some worst-case power, but also preventing it from delaying processes for a very long time, except with small probability.
We remark that the results we give in this paper were not known previously even for a \emph{synchronous} scheduler, which is weaker than our scheduler.

The model is formally introduced in Section~\ref{sec:model}, but we give a quick overview here.
We remark that this model is used only to analyze the \emph{performance}
of our algorithms, not their \emph{correctness}:
for correctness, we carry out worst-case analysis,
and show that our algorithms are \emph{linearizable} and either \emph{lock-free} or \emph{wait-free} in all executions,
including those that have low probability in our model.

We also note that the stochastic scheduler we introduce is not intended to capture \emph{oversubscription} or \emph{context switches}; modeling
context-switching delays and analyzing algorithms against them is beyond the scope of this paper, and is left for future work. However, we conjecture that our algorithms remain efficient even when some processes can be descheduled for a long block of time, provided that the timing of such blocks is ``somewhat stochastic'' rather than worst-case.

\begin{figure*}[t]
  \centering
  \includegraphics[width=0.6\linewidth]{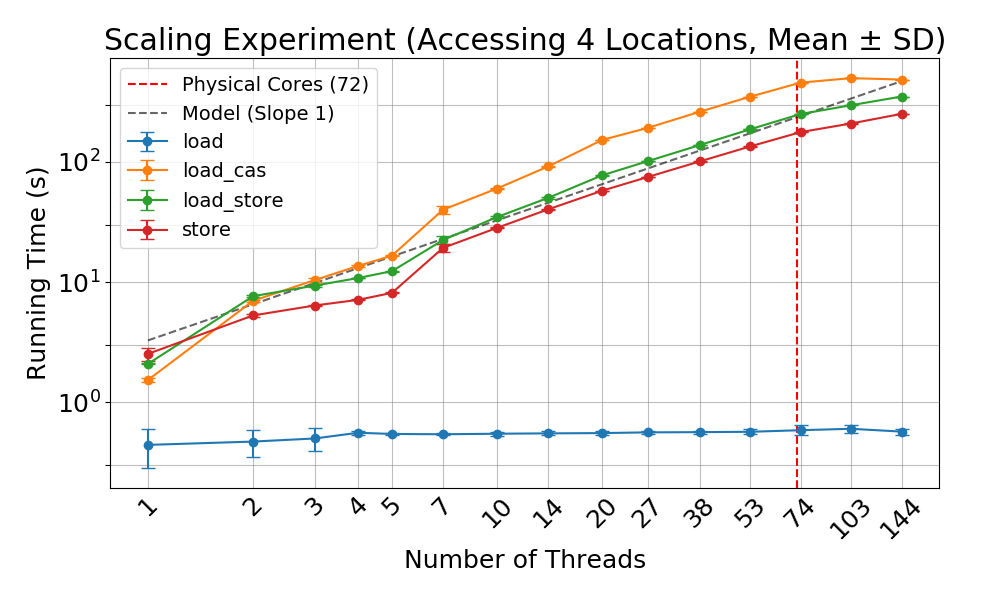}
  \Description{Log-log plot of running time versus thread count for load, load-cas, load-store, and store operations on four memory locations.}
  \caption{\boldmath Running times for $10^8$ operations per thread of various types, randomly applied to four memory locations, as a function of the number of threads on a log-log plot. The experiment shows that load times are flat and hardly affected by the thread count. However, the times for stores, load-cas, and load-stores increase linearly with the number of threads (the dashed line has slope 1). This is what is predicted by CRQW models. More
    details on the setup are given in \cref{sec:experiments}.
  }
  \label{fig:fig-scaling}
\end{figure*}

\paragraph{Modeling contention.}
We adopt the \emph{Concurrent-Read Queue-Write (CRQW)} contention model of~\cite{EHN12,SGBFG15},
where reads are executed immediately (and concurrently with one another), but writes to the same memory location are enqueued, and can only be executed one at a time.
This is a realistic model of modern hardware, as our experimental results in Fig.~\ref{fig:fig-scaling} demonstrate.
In the CRQW model, all write or read-modify-write operations to the same memory location are enqueued in a \emph{hardware queue}.
At each timestep, if the hardware queue is not empty, one operation is dequeued and applied.
We make no assumptions regarding the order in which simultaneous operations are enqueued:
this is controlled by a worst-case adversary.

\paragraph{The scheduler.}
We assume a stochastic scheduler that can exhibit some worst-case behavior.
In each window of $\tau$ time units (for a constant parameter $\tau$),
for each process, with probability $1/2$, the scheduler is \emph{required} to allow the process to take at least one step in the window. However, the scheduler is \emph{allowed} to schedule processes more often than this:
for example, it can deterministically schedule a process $p_1$ at every timestep of the $\tau$ time units, while deterministically scheduling another process $p_2$ to take only one step every $\tau$ time units (i.e., $\tau$ times slower than $p_1$).

In our model, processes that are scheduled take steps in parallel
(in contrast to the stalls model of~\cite{DHW97}, where process steps are interleaved). This is important for studying the \emph{latency} of operations in a system that has true parallelism, e.g., multiple cores.

\paragraph{Modeling the environment.}
Finally, we assume that operations are invoked by an \emph{adaptive worst-case adversary}, which can see the results of all operations invoked so far, including the timing of their responses.
The adversary then chooses what operation to invoke next, and which process should carry out the operation (among the processes that are not currently executing some other operation).
This strong adversary model is necessary for compositionality:
when our primitives are used as building blocks inside other algorithms,
the timing of their responses can determine what subsequent operations are invoked by the outer algorithm.
The fact that our algorithms are fast even under adaptive adversaries allows us to prove our composition theorem.

\subsubsection{Comparison to other models in the literature.}
\label{sec:compare}
Our model generalizes several classes of schedules that have been studied in the concurrent and parallel algorithms literature.
For example, our model generalizes both the well-studied \emph{semi-synchronous} model, where there is an upper bound on the \emph{ratio} between the fastest and the slowest process (see, e.g.,~\cite{DLS88,AM94,ADLS94,ELMS05,T07,HK06} and many others), and a random coin-flipping model, in which each process is scheduled with probability 1/2 in each timestep.
Our scheduler also generalizes the greedy scheduler considered in~\cite{GibbonsMR98b} and in other work on asynchronous PRAM algorithms.
In addition,
our scheduler
is more permissive than
several models of stochastic scheduling that have been used to study the real-life performance of concurrent algorithms (e.g.,~\cite{GM01,ASV15,ARP15,ACS16,ART16}).
These schedulers typically do not have any \emph{non-stochastic} behavior---they are modeled as a purely random process, which does not allow room to model varying instruction and processor speeds, and other non-stochastic effects.
Moreover, some stochastic schedulers, such as~\cite{ACS16}, use the \emph{interleaved} model of concurrency, where the schedule is modeled as a sequence of processor steps,
each chosen at random.
Due to the birthday paradox,
the model of~\cite{ACS16} essentially limits parallelism to $O(\sqrt{P})$,
and thus the model is not appropriate for settings with true parallelism
(see further discussion
in \cref{sec:non_adversarial}).

\subsection{Our Results}

Our main results are contention-resolution algorithms
for
read/write registers and for CAS registers,
which both have latency $O(\log P)$ w.h.p.\
in the stochastic CRQW model with $P$ processes.

For read/write registers, we show:

\begin{theorem}
  \label{thm:register-main-intro}
  For any $\l > 0$,
  there is an implementation of an
  $\l$-bit read/write register that is linearizable and wait-free,
  such that under the stochastic CRQW model with adaptive inputs, each operation has latency at most $O(\log P)$ with high probability in $P$.
  The construction uses $O(1)$ hardware read/write registers with word size $w \ge \l + \eps \log P$ bits, where $\eps > 0$ is a constant.
\end{theorem}

Next, we construct a compare-and-swap (CAS) register:

\begin{theorem}
  \label{thm:cas-register-main-intro}
  For any $\l > 0$, there is an implementation of an
  $\l$-bit CAS register that is linearizable and lock-free, such that under the stochastic CRQW model with adaptive inputs,
  each operation has latency at most $O(\log P)$ with high probability in $P$.
  The construction uses $O(1)$ hardware CAS and read/write registers,
  with word size $w \ge \max\BK{\l + 2 \log \log P, \; 2 \log P}$ bits.
\end{theorem}

We emphasize that the latency guarantees for both constructions hold for arbitrarily long executions, not just executions where the number of operations is polynomially bounded in $P$; this is nontrivial because over very long executions, low-probability ``bad schedules'' and other unlikely events are bound to eventually occur.
Nevertheless, in both of our constructions, any operation has high-probability latency $O(\log P)$, no matter how many operations were invoked before it.
Also, such constructions were not known previously even under a \emph{perfectly synchronous} scheduler (which is weaker than the stochastic CRQW model).

Next, we give a composition theorem that allows us to plug our primitives into existing worst-case concurrent algorithms that were analyzed without accounting for contention, to obtain efficient algorithms in the stochastic CRQW model.
Essentially, this shows that our constructions can allow algorithm designers to ignore contention and work in the ``standard'' worst-case model, and later transform their algorithms into ones that handle contention efficiently in a black-box manner.

\begin{theorem}[Informal]
  Given a concurrent deterministic algorithm that uses read/write registers and/or CAS registers,
  where each operation takes at most $T$ instructions in the worst-case asynchronous shared-memory model (which does not model contention),
  we can plug in our read/write and CAS registers to obtain an algorithm for the stochastic CRQW model
  where each operation has high-probability latency $O(T \log P)$.
\end{theorem}

Using this composition theorem and existing constructions,
we obtain low-latency implementations of several other synchronization primitives.
These include load-linked/store-conditional (LL/SC) \cite{BW20}, fetch-and-increment~\cite{ERW12}, and bounded max registers and counters~\cite{AAC09}.
In all cases, the resulting primitive is linearizable and lock-free, and has polylogarithmic latency in the stochastic CRQW model.
The details of the transformation and applications are given in \cref{sec:applications}.

Finally,
we prove a
trade-off between the \emph{space complexity},
the \emph{high-probability latency} and the \emph{expected latency}
for a wide class of synchronization primitives, including test-and-set, read/write registers, CAS registers, and more.
One might initially expect that we can obtain high-probability latency of $O(\log P)$ (as our algorithms do), while simultaneously having much better \emph{expected} latency, perhaps even $O(1)$:
the scheduler in our model allows processes to take steps every $O(1)$ time units in expectation, and every $O(\log P)$ steps with high probability, so there is no barrier a priori.
However, we show that this is impossible:
for an algorithm with space complexity $M$,
if any operation terminates within $L$ timesteps with high probability,
then the \emph{expected} time to complete an operation is at least $\Omega(\log_{ML} P)$;
in particular,
if $M,L$ are polylogarithmic in $P$,
then the expected time complexity is $\Omega(\log P / \log \log P)$.
This trade-off applies even under the greedy scheduler (from, e.g.,~\cite{GibbonsMR98b}),
and even if each process executes only one operation,
initiated at the very beginning of the execution.

Interestingly, the trade-off does not rule out contention-resolution algorithms
that use $\polylog P$ space
and have expected latency $\polylog \log P$,
but do \emph{not} have low latency with high probability.
Indeed, we conjecture that such algorithms exist.
(For example, the trade-off does not rule out CAS registers that use constant memory and have constant expected latency, but the only high-probability guarantee they provide on latency is $O(P^c)$ for some constant $c > 0$.)

\section{The Stochastic CRQW Model}
\label{sec:model}

In this section, we formally introduce the stochastic CRQW model, which combines the standard CRQW model \cite{GibbonsMR96,GibbonsMR98a,GibbonsMR98b,SGBFG15} with a stochastic scheduler.

We adopt the word RAM model, where the memory consists of $w$-bit \defn{cells}
(also called \defn{machine words}).
There are $P$ processes, and the memory is partitioned into $P$ \defn{local memories},
one for each process, and a \defn{shared memory} that all processes can access.
In addition, each cell $C$ is associated with an \defn{instruction queue} $Q_C$
which stores \emph{waiting} instructions for cell $C$ (explained below).
We refer to the set of processes as $[P] = \left\{ 1,\ldots,P \right\}$,
and it is conventionally assumed that processes know their own process IDs, although this is not needed for algorithms in this paper.

We use $\poly(P)$ to denote a sufficiently large polynomial $P^{\kappa}$, where $\kappa$ is a constant that usually depends on other parameters in the algorithm, and can be set to larger than any fixed constant.
We say an event happens ``with high probability in $P$'' (or ``with high probability'' for short) if its probability is at least $1 - 1 / \poly(P)$.

\paragraph{Executions.}

An \defn{execution} of the system is represented by a sequence of discrete \defn{timesteps}.
As is standard, we use \emph{time} $t$
to refer to the moment immediately before the $t$-th timestep (e.g., we may refer to ``the state of the system at time $t$'').
In each timestep $t$, a subset of the processes is scheduled, and each of these processes \defn{invokes} an instruction.
There are two types of instructions:
\defn{local} instructions, which do not access the shared memory, and \defn{shared} instructions, which access the shared memory.
Shared instructions are further divided into two types:
\defn{state-changing instructions}, which may alter the shared memory,
and \defn{non-state-changing instructions}, which may not.

An instruction invoked in timestep $t$ is \defn{applied} in some timestep $t' \geq t$,
and its effects are visible at time $t' + 1$.
Local instructions and non-state-changing instructions are always applied in the same timestep in which they are invoked;
however, state-changing shared instructions may incur delays due to contention,
and consequently may be applied in a later timestep (see below).
This distinction between state-changing and non-state-changing instructions is what makes our model a variant of the \emph{CRQW} model, where ``reads'' are not queued,
but ``writes'' are.

In this paper, we use three types of shared instructions:
$\Load(C)$, $\Store(C, \xNew)$, and $\CAS(C, \xOld, \xNew)$.
A $\Load(C)$ instruction is non-state-changing:
If invoked in timestep $t$, it is
immediately applied in the same timestep, and returns the contents of cell $C$ at time $t$.
In contrast, $\Store(C, \xNew)$ and $\CAS(C, \xOld, \xNew)$ instructions
are state-changing.
Upon invocation in timestep $t$,
these instructions are \emph{enqueued} in the instruction queue $Q_C$.
If multiple $\Store$ and/or $\CAS$ instructions are invoked on the same cell $C$ in timestep $t$, they are added to $Q_C$ in adversarial order, determined by the scheduler.

After all instruction invocations in a given timestep, for each cell $C$ with a non-empty instruction queue $Q_C$,
the instruction at the head of the queue is dequeued and applied.
For a $\Store(C, \xNew)$ instruction, the value of cell $C$ is changed to $\xNew$, and this value is reflected in the system state at time $t+1$.
For a $\CAS(C, \xExp, \xNew)$ instruction, if the value of cell $C$ is $\xExp$ at time $t$, then the value is changed to $\xNew$ at time $t+1$, and the value returned is ``true''.
Otherwise, the value of cell $C$ is not changed, and the value returned is ``false''.
A process that has invoked an instruction that is currently waiting in some instruction queue is said to be \defn{waiting},
and otherwise the process is \defn{ready} to invoke its next instruction.

\paragraph{Scheduler.}
As we said above, in each timestep, a subset of processes is scheduled to take a step.
Only processes that are \emph{ready} can be scheduled,
and each process that is scheduled invokes its next instruction.
The processes that are scheduled at any given timestep
are chosen by a \defn{scheduler},
which also determines the order in which newly-invoked state-changing instructions applied to the same cell $C$ are added to the instruction queue $Q_C$.

We assume a stochastic adaptive adversarial scheduler, which is given the entire history of the execution,
but cannot see the local randomness of the processes before it affects the shared memory.
It is convenient to model this by adding a new instruction, $\StoreRandom(C, \D)$, where $C$ is a cell in the shared memory and $\D$ is a distribution of values that can be efficiently sampled using $O(1)$ local instructions. When a $\StoreRandom$ instruction is applied, it atomically samples a value $x \sim \D$ and stores it to $C$. In this paper, we only use one type of distribution: a uniform distribution over an interval of integers.
Since the random value is not generated until the $\StoreRandom$ instruction is applied, the scheduler cannot use it to make its scheduling decisions. Note that, although it is convenient to model $\StoreRandom$ as a separate instruction, it is really a constraint on the scheduler's power, and not a new instruction that needs to be implemented in hardware.

The scheduler is required to satisfy the following \emph{random-delay} property:

\begin{definition}[Informal]
  \label{def:random-delay-informal}
  We say that the scheduler satisfies the \defn{random-delay} property if there is a global constant $\tau \ge 1$ such that for every time $t$ and process $p$ that is \emph{ready} at time $t$, the probability that $p$ is scheduled at least once in the time interval $[t, \; t + \tau)$ is at least $1/2$, independently of the history up to time $t$.
\end{definition}

In~\cref{sec:terminology} we formally define this property through a collection of independent coin flips that the scheduler is required to obey,
but for the purpose of our technical overview, the informal definition is convenient.
In our constructions, $\tau$ is treated as a fixed system constant, and the algorithms may depend on an upper bound on it as well as on $P$.

\paragraph{Objects, operations, and the adaptive input model.}
A concurrent object supports a collection of \emph{operations},
which can modify the state of the object and/or return a value.
Operations are \emph{invoked} by processes in some timestep $t$,
and can eventually return in some later timestep;
we say that an operation is \defn{ongoing} if it has been invoked and has not yet returned.
We use boldface (e.g., $\opFont{op}$) to typeset  high-level operations, to distinguish them from hardware instructions (this is especially important for CAS).

In this paper, we assume that operations are invoked by a strong adaptive and global adversary, which we refer to as an \defn{adaptive user} (or \defn{user} for short). The user decides what operations to invoke in a given timestep $t$, and on which processes, based on the entire history of operation invocations and responses, including the times at which operations return (i.e., the user can carry out \emph{timing attacks}).
However, the adversary cannot directly see the internal state of the system, including the shared memory and the internal randomness of the processes; it interacts with the system only by invoking and observing the responses of operations.

Since each process can only run at most one operation at a time, we often conflate an operation with the process that is executing it; e.g., we may say that ``operation $q$ invokes a $\Store$ instruction'' instead of saying that the process executing operation $q$ invoked the instruction.
For the sake of analysis, we also assign a unique identifier $q \in \N$ to each operation in the execution ($q$ is just a label and does not relate to the ordering of operations). The operation does not receive $q$ as input.

We use the standard definitions of linearizability, lock-freedom and wait-freedom, which we omit here for brevity. These are defined with respect to the standard interleaved execution model, where processes take steps one at a time,
and we prove that our algorithms are linearizable and lock-free/wait-free under the interleaved model.

\paragraph{Operation latency.}

The main performance metric that we consider in this paper is \defn{latency},
that is, the time between the invocation of an operation and the response.
Our goal is to construct objects where all operations have latency $O(\log P)$ with high probability.
There is some subtlety involved in formalizing this requirement: because operations are chosen by an adaptive user, the notion of ``each operation'' is not well defined a priori.
We address this issue by quantifying over \emph{times} instead of over operations:

\begin{definition}[High-Probability Latency]
  \label{defn:high-probability-latency}
  Given an implementation of an object,
  we say that the \defn{high-probability latency of each operation} is at most $T$ if
  for any adaptive user and any scheduler satisfying the random delay rule,
  for every fixed time $t \ge 0$, with high probability in $P$, all operations invoked at time $t$ complete within $T$ time.
\end{definition}

This definition is useful because it lends itself to union bounds, allowing us to, for example, argue that several operations invoked one after another have low total latency (w.h.p.).
For the lower bound, we also consider \emph{expected latency} (see \cref{sec:lower-bound}).

\section{Technical Overview}
\label{sec:overview}

In this technical overview, we present our implementations for read/write and CAS registers,
and give a high-level overview of their analysis; we also sketch the main idea behind our lower bound.

\paragraph{Reduced representation.}
Although we write algorithms in pseudocode, for the latency analysis it is convenient to use a reduced representation where multiple local instructions are contracted together with the previous shared instruction. This allows us to reason only about shared instructions, which are where the crux of the analysis lies. Formally, the reduced representation takes the form of a \emph{finite state machine}, where each state represents one shared instruction followed by $O(1)$ local instructions. Transitions between states can be probabilistic, and can also depend on the values of local variables (but not the shared memory). In \cref{sec:state-machine}, we define this representation more formally and prove that it is equivalent to the original pseudocode in terms of latency.

\subsection{Constructing Read/Write Registers}
\label{sec:sketch_reg}

We begin by giving a high-level overview of our read/write register construction and its analysis.
For the sake of simplicity,
we focus on a special case where each $\Write$ operation is invoked with a unique value that is never used again in any other $\Write$ operation.
This means that if we \Read the register twice at times $t_1 < t_2$ while a \Write takes effect in between, the two \Read{}s will return different values, allowing us to detect the change.
In the full algorithm, we do not make this assumption (\cref{alg:shared-register-write}),
and instead
we attach a random \emph{fingerprint} to each value written,
which allows operations to detect changes to the register by other operations with certain probability.

\subsubsection{Overview of the Register Construction}

Our register construction uses a single memory cell, $\SharedCell$.
High-level $\Read$ operations simply $\Load$ the value of $\SharedCell$ and return it.
As for $\Write$ operations, they are implemented through a \emph{back-on} mechanism:
Upon invocation of $\Write(x)$,
a process first $\Load$s $\SharedCell$,
and initializes a write probability $p$ to $1/P^4$.%
\footnote{This initial small value simplifies the analysis, but we believe it is not essential, and can be increased to $1/P$.}
Then the process enters a loop,
where it $\Load$s $\SharedCell$ again to check if its value has changed;
if it has, the process returns immediately,
and if not, with probability $p$, the process invokes a hardware $\Store$
to write the value $x$ to $\SharedCell$ and then returns.
With probability $1-p$, the process instead increases $p$ by a multiplicative factor of $(1 + 1/c)$ for some large constant $c$,
and repeats.

\cref{alg:shared-register-write} gives the full pseudocode, where we use a \StoreRandom instruction instead of a \Store to write the value $x$ together with a random fingerprint $f$. However, for the remainder of this overview, because we assume uniquely-valued \Write{}s, we ignore the fingerprint and conflate \Store with \StoreRandom. See \cref{fig:write} for the state-machine representation of the algorithm: state \st{S} represents the first \Load of the cell, state \st{R} represents the back-on loop where the cell is repeatedly \Load{}ed to detect changes, and state \st{W} represents an invocation of a hardware \Store, after which the operation returns.

\begin{figure*}[t]
  \centering
  \begin{minipage}[c]{0.45\textwidth}
    \makeatletter\@twocolumnfalse\makeatother\begin{algorithm2e}[H]
      \caption{$\Write$ Operation}\label{alg:shared-register-write}
      \DontPrintSemicolon
      \SetKwFunction{fWrite}{\textsf{\textbf{write}}}
      \Fn{\fWrite{$\xNew$}} {
        $p \gets 1 / P^4$ \; \label{line:reg/2}
        $(\xOld, \fOld) \gets \Load(\SharedCell)$ \tcp{\instS} \label{line:reg/3} \label{line:reg/S}
        \Loop{
          $(\xNow, \fNow) \gets \Load(\SharedCell)$ \tcp{\instR} \label{line:reg/5} \label{line:reg/R}
          \If{$\fNow \ne \fOld$}{
            \Return \; \label{line:reg/7}
          }
          \WithProb{$p$}{\label{line:reg/8}
            $\StoreRandom(\SharedCell, \; (\xNew,$ \hspace{5em} \phantom{a}
            $\Unif(\BK{0,1}^{\eps \log P}) ))$ \tcp*[l]{\instW} \label{line:reg/W}
            \Return \; \label{line:reg/abort}\label{line:reg/10}
          }
          $p \gets p \cdot (1 + 1 / c)$ \label{line:register-write/defn-c} \; \label{line:reg/increase-prob}
        }
      }
    \end{algorithm2e}
  \end{minipage}
  \hfill
  \begin{minipage}[c]{0.45\textwidth}
    \centering
    \includegraphics[width=\textwidth]{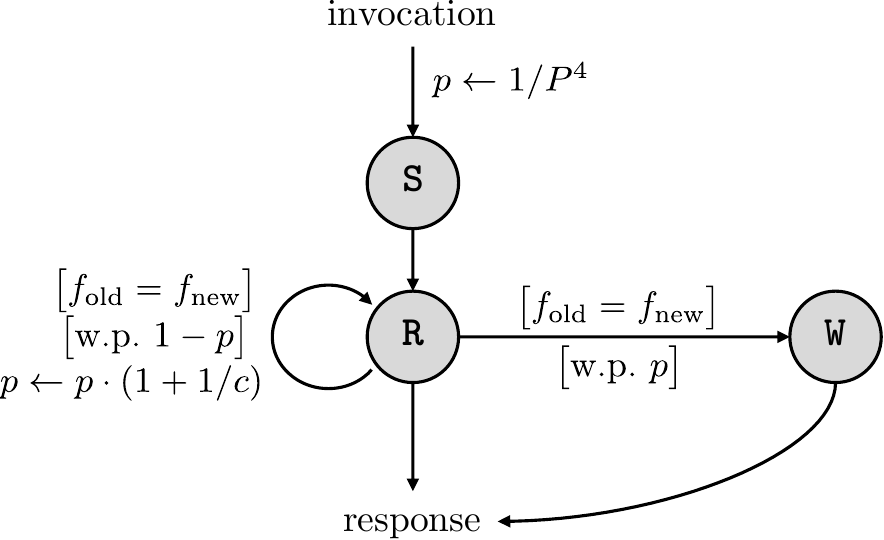}
    \Description{State machine for the write operation with states S, R, and W.}
  \end{minipage}
  \caption{Pseudocode for the $\Write$ operation of our read/write register, and its representation as a state machine. The comments ``\texttt{S}'', ``\texttt{R}'', and ``\texttt{W}'' in the pseudocode indicate the beginning of states \st{S}, \st{R}, \st{W} (resp.).
    In state \st{S}, the process reads the register. In state \st{R} it loops and re-reads the register, until either detecting a change and returning,
    or transitioning to state \st{W}, where it invokes a $\StoreRandom$
    and returns.}
  \label{fig:write}
  \label{fig:state-machine-register}
\end{figure*}

Linearizability of our register is straightforward to see:
$\Read$ operations are linearized when their $\Load$ is applied;
$\Write$ operations that invoke a hardware $\Store$ are linearized
when the \Store is applied; and $\Write$ operations that return as a result of observing a change in value are linearized immediately before the operation whose $\Store$ instruction changed the value.

\subsubsection{Analysis of the Latency}
\label{sec:overview/register-analysis}

The main idea of the latency analysis is to show that the scheduler cannot (except with low probability) cause a ``pile-up'' of \Write operations that all invoke \Store{}s at approximately the same time, causing significant contention.
To prove this, we analyze the back-on mechanism of the algorithm, and show that except with small probability, no more than $O(\log P)$ operations are posed to invoke hardware \Store{}s simultaneously, so the scheduler cannot cause contention of more than $O(\log P)$ instructions.

\paragraph{Defining and bounding a potential.}

We say an operation $q$ is \defn{active} at time $t$ if it is in state \instR and the value of cell \SharedCell has not changed since $q$ first \Load{}ed it in state \instS. Only active operations can transition to state \instW and invoke a hardware \Store, because by our unique-value assumption, other operations will see a changed value and return immediately from state \instR. Let $p_t^{(q)}$ be the local variable $p$ in \cref{alg:shared-register-write} at time $t$ for operation $q$, which indicates the probability that $q$ transitions to state \instW the next time it iterates through state \instR, assuming it does not quit due to detecting a change in \SharedCell's value. Using this notation, we define the potential at time $t$ as%
\footnote{This is analogous to what is often referred to as the \emph{broadcast sum} in wireless broadcast algorithms.}
\[
  \phi_t = \sum_{\substack{\text{active operation $q$} \\ \text{at time $t$}}} p_t^{(q)}.
\]

Intuitively, the potential at time $t$ equals the expected number of operations that will transition to state \instW and invoke a hardware \Store in the next constant number of timesteps, assuming they do not quit due to an earlier change in \SharedCell's value. Thus, bounding the potential allows us to bound the number of $\Store$ instructions that are invoked at roughly the same time.

We show that the potential at any time $t$ is at most $\phi_t \le \beta \log P = O(\log P)$ w.h.p.,
for a large constant $\beta$.
This holds due to a combination of two factors. First, the potential can never increase too fast---newly invoked operations only contribute a total of $P \cdot (1 / P^4) = 1 / P^3$ to the potential, and ongoing operations increase their write probability by at most a constant factor $(1 + 1/c)$ per timestep. Second, whenever the potential exceeds a threshold $\alpha \log P$ where $\alpha < \beta$ is a smaller constant, there is a ``downward pressure'' that causes the potential to decrease: in expectation, $\Omega(\log P)$ operations will transition to \instW and invoke a \Store in the next constant number of timesteps, so with high probability, at least one operation will do so. This \Store changes the value of cell \SharedCell, rendering all active operations inactive and decreasing the potential to zero.
Together, this shows that the potential cannot increase past $\alpha \log P$ for a long enough time to reach the higher threshold of $\beta \log P$, except with small probability in $P$.

As a corollary, in any timestep $t$, the number of operations that transition from state \instR to state \instW is at most $O(\log P)$ w.h.p.; this will be useful to bound the contention, as we see next.

\paragraph{Busy intervals.}
We say that a time interval $[t_0, \, t_1)$ is \defn{busy} if
the instruction queue $Q_{\SharedCell}$ is empty
at times $t_0$ and $t_1$,
but non-empty at all internal times $t_0 < t < t_1$, so that one \Store is applied on \SharedCell in every timestep $t \in [t_0, \, t_1)$.
We claim that w.h.p., each busy interval has length at most $O(\log P)$.

In order for a long busy interval $[t_0, \, t_1)$
of length $T = t_1 - t_0$ to occur,
there must be a large ``pile-up'' of operations that
transition from \instR to \instW
at approximately the same time (see Fig.~\ref{fig:busy_interval}):
We know that the instruction queue
is empty at time $t_0$,
but contains a total of $T$
\Store instructions over the duration of the busy interval.
We claim that each such instruction is invoked by a process that is in state \instR prior to
time $t_0$,
transitions to state \instW
prior to time $t_0$,
and then does not take another step until some time $t \in [t_0, \, t_1)$ inside the busy interval,
at which point it invokes a \Store instruction
that is added to the queue.
This is because during a busy interval the value of \SharedCell changes in every timestep, as the instruction queue is always non-empty;
therefore, any operation that
is in state \instR
and takes a step \emph{inside} the busy interval
observes the changed value and returns without invoking \Store.

To complete the analysis,
we consider each timestep $t_0 - j$ before the busy interval (for $j > 0$),
and bound the number of operations that
(a) transition from \instR to \instW in timestep $t_0 - j$, and
(b) do not take another step until after time $t_0$.
By the upper bound on the potential, the number of operations that transition from \instR to \instW in timestep $t_0 - j$ is at most $O(\log P)$ w.h.p.; by the random-delay rule of the stochastic scheduler (\cref{def:random-delay-informal}), the probability that each of these operations is not scheduled until time $t_0$ is at most $2^{-\Omega(j / \tau)} = 2^{-\Omega(j)}$.
Summing over all $j > 0$, we see that the expected number of operations that transition from \instR to \instW before time $t_0$
but do not take another step until after time $t_0$ is at most $\sum_{j = 1}^{\infty} 2^{-\Omega(j)} \cdot O(\log P) = O(\log P)$. The desired bound on the length of the busy interval then follows by concentration inequalities.
(This informal explanation ignores potential dependencies between the scheduler's and the user's adaptive decisions and the algorithm's randomness. In the full analysis, we handle these dependencies carefully, to obtain the independence properties that are needed for concentration.)

\begin{figure*}
  \centering
  \includegraphics[width=0.7\linewidth]{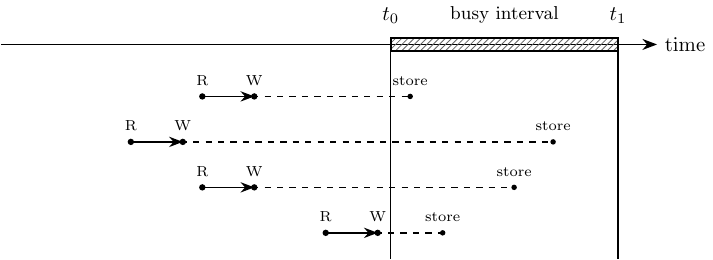}
  \Description{Timeline showing a busy interval and transitions from state R to state W before the interval.}
  \caption{
    A busy interval $[t_0, t_1)$. An arrow from \instR to \instW represents a timestep when the operation applies a \Load and transitions from state \instR to state \instW. Only operations that transition from \instR to \instW \emph{before} time $t_0$ can invoke $\Store$ instructions inside the busy interval,
    and there must be at least $t_1 - t_0$ such operations.
  }
  \label{fig:busy_interval}
\end{figure*}

\paragraph{Bounding the latency.}

A \Read operation completes after the process is scheduled once,
so its latency is $O(\log P)$ w.h.p.\ due to the stochastic scheduler.
As for $\Write$s, each \Write operation is scheduled for at most $O(\log P)$ times before it either returns or invokes a \Store, which takes at most $O(\log P)$ time w.h.p. In the latter case, the operation waits for the \Store to be applied, but the waiting time is bounded by the length of the busy interval during which the \Store is invoked, which is in turn bounded by $O(\log P)$ with high probability. Adding both parts together, the latency of \Write operations is w.h.p.\ $O(\log P)$.

\paragraph{Removing the unique-value assumption.}
So far, our overview has assumed that no two $\Write$ operations write the same value.
This simplifying assumption allowed us to argue that inside a busy period, every $\Store$ instruction applied to memory immediately releases all operations that take a step in state \instR to return without invoking a $\Store$ themselves,
as they observe a change of the register's value.

In the full analysis (\cref{sec:register}), we do not assume that each \Write operation is invoked with a unique value.
Instead, a random fingerprint of $\eps \log P$ bits is attached to the value of the register in the cell \SharedCell. Due to the possibility of fingerprint collisions, when the value of \SharedCell changes, we can no longer assert that all ongoing \Write operations will quit without invoking \Store{}s. We handle this issue by modifying the definition of \emph{active operations}, removing the requirement that the value of \SharedCell has not changed since the operation started. The definition of the potential is also modified accordingly. Then we use a more careful argument to show the ``downward pressure'' that causes the potential to decrease (\cref{lem:potential-drop-after-store}), and the desired bound on the potential follows (\cref{lem:potential-upper-bound}).

\subsection{Constructing a CAS Register}
\label{sec:sketch_CAS}

Our CAS register construction and its analysis are more involved than our read/write register.
A key difference between the semantics of the two primitives is that for the read/write register, there is a built-in ``contention-sensing'' mechanism: whenever a hardware $\Store$ instruction is applied, the value of the register changes, and this change can be seen by processes that re-$\Load$
the register.%
\footnote{We assumed for simplicity that values are never repeated, but as we briefly mentioned in the previous section,
  even if this is not the case, we can use random fingerprints to sense $\Store$ instruction applications with high probability.}
This is not true for $\CAS$: a failed hardware $\CAS$ still causes contention, but it does not change the value, so other processes cannot observe it.
We must be careful to avoid a high-contention scenario where the hardware instruction queue becomes full of ``invisible'' $\CAS$ instructions
that are bound to fail when applied,
with more processes coming in and invoking even more $\CAS$ instructions
because they cannot sense the high contention.

In the remainder of this section, we present our CAS register construction in two stages, and briefly discuss its analysis.
We begin by constructing a \emph{short-lived} CAS register that is similar to our read/write register, and performs well w.h.p.\ for $\poly(P)$ operations. We also explain why this simple strategy is \emph{not} suitable for obtaining a long-lived CAS register that performs well for arbitrarily long executions.
Then, we show how to modify the construction to obtain a \emph{long-lived} CAS register, where each operation has high-probability latency $O(\log P)$ no matter how many operations are invoked before it.
The full construction and analysis are given in \cref{sec:cas-register}.

\subsubsection{Short-Lived CAS Register}

Our construction of a short-lived CAS register
is given in Algorithm~\ref{alg:basic-cas-register}.
As in the read/write register, $\CellC$ stores a pair $(x,f)$: $x$ is the current register value and $f$ is a fingerprint.
The construction differs in two main ways:
First, the probability of invoking a hardware $\CAS$ starts out significantly smaller than for the read/write register, but increases more quickly (for technical reasons).
Second, we take care to prune out certain types of ``invisible'' $\CAS$ operations:
If a high-level operation $\opCAS(C, \xExp, \xNew)$ is invoked, and at any point, it reads a value $x \neq \xExp$ from register $C$,
then we pretend that the $\opCAS$ was applied at that point but failed, and return immediately;
if $\opCAS(C, \xExp, \xNew)$ is invoked with $\xExp = \xNew$, then we simply pretend that it was successfully applied and return immediately.
By pruning out such operations as soon as possible, we ensure that all operations that try to invoke the hardware $\CAS$ at least have the \emph{potential} to change the register's value when they are applied.

As with the read/write register, we use a fingerprinting mechanism to detect successful \CAS instructions, even if they end up changing the register to a value that was previously observed by some process (the so-called ``ABA problem'').
The fingerprinting mechanism is different from that of our read/write register, and
unlike the read/write register, fingerprinting
causes some minor subtleties in the linearizability proof:
For example, a changed fingerprint could cause a $\opCAS(C, \xExp, \xNew)$ to fail even though the current \emph{logical} value of the $\CAS$ register is indeed $\xExp$, but the fingerprint does not match the latest value the process read from $\SharedCell$.
We handle this by linearizing such failed $\opCAS$ operations \emph{before} the preceding successful $\opCAS$ operation, and we show that, because the logical value of the \CAS register changes at every successful $\opCAS$, this is indeed proper.

\begin{figure}[t]
  \centering
  \begin{minipage}[t]{0.47\textwidth}
    \vspace{0pt}
    \makeatletter\@twocolumnfalse\makeatother\begin{algorithm2e}[H]
      \caption{Basic CAS Register}\label{alg:basic-cas-register}
      \DontPrintSemicolon
      \SetKwFunction{BasicCAS}{\textsf{\textbf{BasicCAS}}}
      \Fn{\BasicCAS{$\xExpect$, $\xNew$}} {
        $p \gets 1 / P^{2c^3}$ \; \label{line:basic-cas/defn-c}
        $(\xOld, \fOld) \gets \Load(\CellC)$ \tcp{S} \label{line:basic-cas/S}\label{line:basic-cas/3}
        \If{$\xOld \ne \xExpect$}{ \label{line:basic-cas/4}
          \Return false \; \label{line:basic-cas/5}
        }
        \If{$\xNew = \xOld$}{ \label{line:basic-cas/6}
          \Return true \; \label{line:basic-cas/7}
        }
        \Loop{
          $(\xNow, \fNow) \gets \Load(\CellC)$ \tcp{R}
          \label{line:basic-cas/9}
          \If{$\xNow \ne \xOld$ or $\fNow \ne \fOld$}{
            \label{line:basic-cas/10}
            \Return false \;
            \label{line:basic-cas/11}
          }
          \WithProb{$p$}{
            \label{line:basic-cas/12}
            $\fNew \gets (\fOld + 1) \bmod (\log^2 P)$ \;
            \label{line:basic-cas/13}
            result $\gets \CAS(\CellC, \; (\xOld, \fOld),$ $\; (\xNew, \fNew))$ \tcp{C}
            \label{line:basic-cas/14}
            \Return result \;
            \label{line:basic-cas/15}
          }
          $p \gets 2p$ \;
          \label{line:basic-cas/16}
        }
      }
    \end{algorithm2e}
  \end{minipage}%
  \hfill
  \begin{minipage}[t]{0.51\textwidth}
    \vspace{0pt}
    \makeatletter\@twocolumnfalse\makeatother\begin{algorithm2e}[H]
      \caption{Long-Lived CAS Register} \label{alg:improved-cas-body}
      \DontPrintSemicolon
      \SetKwFunction{ImprovedCAS}{\textsf{\textbf{ImprovedCAS}}}
      \Fn{\ImprovedCAS{$\xExpect$, $\xNew$}} {
        \tcpNoSpace{Waiting Phase}
        $\yOld \gets \Load(\CellW)$ %
        \Loop{
          wait for $2c \log P$ steps \; %
          $\yNow \gets \Load(\CellW)$ \; %
          \If{$\yNow = \yOld$}{
            \Break \;
          }
          $\yOld \gets \yNow$ \;
        }
        \tcpNoSpace{Calling Phase}
        result $\gets \BasicCAS(\CellC, \, \xExpect, \, \xNew)$ \;
        \tcpNoSpace{Writing Phase}
        \If{\BasicCAS applied a \CAS instruction}{
          $\StoreRandom(\CellW, \; \Unif(\BK{0,1}^{2 \log P}))$ %
        }
        \Return result \;
      }
    \end{algorithm2e}
  \end{minipage}
\end{figure}

\paragraph{Analysis of the latency.}

We define the potential in the same way as for the read/write register (see \cref{sec:overview/register-analysis}).
However, an essential property that was true for the read/write register is \emph{not} true for the CAS register:
For the read/write register, whenever the potential exceeds some threshold $\alpha \log P$,
there is a ``downward pressure'' that causes the potential to decrease w.h.p.,
which implies $\phi_t \le O(\log P)$ w.h.p.\ at any point in the execution.
For the CAS register, this is only partially true:
When the potential exceeds $\alpha \log P$ \emph{but is not too large},
there is a strong ``downward pressure'' that causes the potential to quickly decrease w.h.p., emptying the instruction queue and reducing the potential to $1/P^{\Theta(1)}$.
However, if the potential ever increases past a logarithmic level, then an \emph{upward pressure} applies instead: the queue fills up instead of emptying, and arriving operations continue to back on, increasing the potential.
Although this is a low-probability event, in a long enough execution, it is bound to occur, and from this point on, the register does not return to a ``healthy'' state (except with exponentially small probability).
Thus, we can only show that this version of the \CAS register works well for $\poly(P)$ operations.

To analyze the short-lived \CAS register we define a ``healthy state'' for the CAS register, which approximately means that the potential is $1/\poly(P)$ and the instruction queue is empty.
Then we prove that if we are in a healthy state at time $t$, then w.h.p.\ we will again be in a healthy state at time $t + O(\log P)$.
This implies that the instruction queue empties every $O(\log P)$ timesteps (w.h.p.),
which means that the latency of operations is $O(\log P)$ (again, w.h.p.).

\subsubsection{Long-Lived CAS Register}

Before showing how we modify the short-lived \CAS register to make it long-lived,
let us first explain why the short-lived \CAS fails over sufficiently long executions.
Consider an adaptive adversary that repeatedly does the following:
\begin{itemize}
  \item Invokes a $\Read(\CellR)$ operation to get the current value of the \CAS register $\CellR$. Let $x$ be this current value.
  \item Concurrently invokes $\opCAS(\CellR, x, x+1)$ operations on every process that is currently free.
\end{itemize}
In the \emph{common case},
the size of the instruction queue does not grow beyond $O(\log P)$,
and the queue empties before the next hardware $\CAS$ is invoked.
However,
at each iteration of the adversary's pattern,
the following sequence of events occurs with probability at least $2^{-O(\log^2 P)}$:
first, the potential grows to $\Theta(\log^2 P)$ before any high-level $\opCAS(\CellR, x, x+1)$ operation invokes a hardware $\CAS$ instruction.
Then, $\Theta(\log^2 P)$ concurrent high-level $\opCAS$ operations invoke hardware $\CAS$ operations at the same time;
these instructions all enter the hardware instruction queue.
In the next timestep, a hardware \CAS succeeds, changing the value of \CellC from $x$ to $x+1$.
However, in the $\Theta(\log^2 P)$ subsequent timesteps,
the remaining $\Theta(\log^2 P)$ enqueued \CAS instructions are applied, and they all fail, so the value of \CellC remains unchanged for $\Theta(\log^2 P)$ timesteps.
Then, the adversary can invoke $\Theta(P)$ operations $\opCAS(C, \, x+1, \, x+2)$ at the beginning of the period when \CellC is unchanged, and these operations complete their back-on processes in $\Theta(\log P)$ time without detecting a change in the value of \CellC, so they all invoke \CAS instructions, blowing up the instruction queue size to $\Theta(P)$.
This situation will repeat except with exponentially small probability, causing $\Theta(P)$ latency per operation for a long time.

To prevent the aforementioned failure mode, we must be more proactive about managing contention, and specifically about detecting failed \CAS instructions, which delay other \CAS instructions but do not change the content of \CellC.
To do this, we use a separate memory cell $\CellW$, and have each operation that applies a hardware \CAS---no matter if successful or not---\Store a random string to $\CellW$, so that any application of hardware \CAS subsequently changes the content of $\CellW$ (w.h.p.).
Newly-invoked high-level $\opCAS$ operations are now required to wait before trying to carry out their $\CAS$ until they see a ``quiet period'' of $O(\log P)$ steps during which the value of \CellW does not change.
See \cref{alg:improved-cas-body} for the pseudocode.

A priori, it may not be obvious whether adding a contention-detection mechanism that involves writing to memory improves performance or makes it worse: cell $\CellW$ is itself subject to the same queuing effect as cell \CellC, so requiring processes to write to it after every time they apply a hardware \CAS risks creating equally bad contention.
However, $\CellW$ is written to using hardware \Store instructions, not \CAS.
This means that \emph{every write to $\CellW$ is visible}: every time a \Store instruction is applied, the cell changes value (except for collisions in the random string).
We show that this creates a self-regulating effect where during periods of relatively high contention, new operations do not ``pile on'' by trying to invoke hardware \CAS{}s, but instead wait until the instruction queue empties.
In addition, we show that the new waiting period imposed on newly-invoked processes is not too long: only $O(\log P)$ steps (w.h.p.). Combining this with the latency analysis for the basic CAS register yields the desired $O(\log P)$ high-probability latency guarantee for the long-lived CAS register.

\subsection{Lower Bound Overview}
\label{sec:overview-lower-bound}

To lower bound the expected latency of a wide class of shared-memory primitives, we consider the following input: among all $P$ processes, a subset $S$ of them will each receive an update operation in the first timestep, and the rest will stay idle; no operations will be invoked after the first timestep. The key observation is that each individual process in $S$ cannot tell whether it is the only process that receives an update operation until some process invokes a state-changing instruction to the shared memory. This leads to two cases: if each process tends to invoke a state-changing instruction too early, it is expected that many processes will invoke state-changing instructions in exactly the same timestep, leading to large latencies for some of the processes; otherwise, each process tends to not invoke state-changing instructions until a certain time threshold $T$, implying that its expected latency is at least $\Omega(T)$ if it is the only process that receives an update operation. In \cref{sec:lower-bound}, we formalize this intuition and set proper parameters to prove the desired lower bound.

\section{Basic Properties of the Model}
\label{sec:terminology}

This section proves additional properties of the stochastic CRQW model that will facilitate the analysis of algorithms in this model. We also introduce notation and terminology that will be used in the sequel.

\subsection{Stochastic Schedulers and Properties}

We start by defining notation for the execution history and the algorithm's randomness.

We define the history $\HistTrue(t)$ as all events that happened before time $t$, including the invocation and application of all instructions, random bits generated by the algorithm, and the enqueuing order of state-changing instructions.
In each timestep $t$, the scheduler takes $\HistTrue(t)$ and the scheduler's private randomness as input, and decides the set of scheduled processes.

For the sake of analysis, we denote by $\Randomness$ the private randomness of the algorithm: we imagine that there is an infinite tape of random bits for each process that is fixed prior to the execution, and each process reads random bits from its tape instead of generating random bits on the fly; then $\Randomness$ includes all the bits on these random tapes, including random bits that will be used in the future.

Next, we provide a formal definition of the random-delay property for stochastic schedulers. Throughout the remainder of the paper, we use \defn{(fair) coin flips} as an alias for Bernoulli random variables with expectation $1/2$.

\begin{definition}
  \label{def:scheduling-coin}
  \label{def:random-delay}
  We say the scheduler satisfies the \defn{$\tau$-random-delay} property if there are independent fair coin flips $\BK{\coin_{p, j}}_{p \in [P], j \ge 0}$ coupled with the execution of the algorithm (called the \defn{scheduling coins} of the scheduler) such that
  \begin{itemize}
    \item For every $j \ge 0$, all coins $\BK{\coin_{p, j'}}_{p \in [P], j' \ge j}$ are independent of the history $\HistTrue(j \tau)$ and of the algorithm's randomness $\Randomness$.
    \item For every $p \in [P]$ and $j \ge 0$, if $p$ is \emph{ready} at time $j \tau$ and there is $\coin_{p, j} = 1$, then $p$ is scheduled at least once in $[j \tau, \, (j + 1) \tau)$.
  \end{itemize}
\end{definition}

For convenience, we also define
\[
  \CoinSum_{p, [t_1, \, t_2)} \defeq \sum_{j \in \N \,:\, j \tau \in [t_1, \, t_2 - \tau]} \coin_{p, j}
\]
as the sum of scheduling coins of $p$ whose corresponding time windows completely reside in $[t_1, \, t_2)$.

\begin{proposition}
  \label{prop:coin-implies-scheduled}
  For process $p$, time interval $[t_1, \, t_2)$, and parameter $k \ge 1$, if $\CoinSum_{p, [t_1, \, t_2)} \ge k$,
  then either $p$ is scheduled at least $k$ times in time interval $[t_1, \, t_2)$, or $p$ is not ready at some time $t' \in [t_1, \, t_2)$.
\end{proposition}
\begin{proof}
  Suppose $p$ is ready at all times $t' \in [t_1, \, t_2)$. Then,
  \begin{align*}
      & \phantom{{}\ge{}} \#\BK[\big]{\text{scheduled timesteps of $p$ in } [t_1, \, t_2)} \\
      & \ge \sum_{j \tau \in [t_1, \, t_2 - \tau]} \#\BK[\big]{\text{scheduled timesteps of $p$ in } [j \tau, \, (j + 1) \tau)} \\
      & \ge \sum_{j \tau \in [t_1, \, t_2 - \tau]} \coin_{p, j}
    = \CoinSum_{p, [t_1, \, t_2)}
    \ge k,
  \end{align*}
  implying that $p$ is scheduled at least $k$ times in $[t_1, \, t_2)$.
\end{proof}

The independence of scheduling coins allows us to obtain basic facts about the stochastic scheduling via concentration bounds.

\begin{proposition}
  \label{prop:coin-sum-independence}
  Let $\BK{(p_i, \l_i, r_i)}_{i = 1}^k$ be a collection of tuples, such that for any $i \ne j$, we have either $p_i \ne p_j$ or $[\l_i, \, r_i) \cap [\l_j, \, r_j) = \emptyset$. Then,
  $\BK{\CoinSum_{p_i, [\l_i, \, r_i)}}_{i = 1}^k$ are independent. Moreover, this is true conditioned on $\HistTrue\bk[\big]{\min_{i = 1}^k \l_i}$ and \Randomness.
\end{proposition}

\begin{proof}
  The sets of scheduling coins used to determine $\CoinSum_{p_i, [\l_i, \, r_i)}$ for different $i$ are disjoint.
  Then, the statement follows directly from the independence of scheduling coins (see \cref{def:random-delay}).
\end{proof}

\begin{proposition}
  \label{prop:process-schedule-prob}
  Conditioned on $\HistTrue(t)$ and \Randomness, for every process $p$ at time $t$,
  \begin{align*}
    & \Pr\Big[
    \bk[\big]{\textup{$p$ is scheduled at least $k$ times in $[t, \, t + 4 k \tau)$}}
    \lor {} \\
      & \qquad \big(\textup{$p$ is not \emph{ready} at some time $t' \in [t, \, t + 4 k \tau)$}\big)
    \Big] \ge 1 - 2^{-\Omega(k)}
  \end{align*}
  over the randomness of $\CoinSum_{p, [t, \, t + 4k \tau)}$.
\end{proposition}

\begin{proof}
  $\CoinSum_{p, [t, \, t + 4k \tau)}$ is a sum of at least $(4k - 1) \ge 3k$ independent fair coin flips, which implies $\E[\CoinSum_{p, [t, \, t + 4k \tau)}] \ge 3k / 2$. By Chernoff bound, with probability at least $1 - 2^{-\Omega(k)}$, we have $\CoinSum_{p, [t, \, t + 4k \tau)} \ge k$. When this inequality holds, \cref{prop:coin-implies-scheduled} implies that either $p$ is scheduled at least $k$ times in $[t, \, t + 4 k \tau)$, or $p$ is not ready at some time $t' \in [t, \, t + 4 k \tau)$.
\end{proof}

\cref{prop:coin-sum-independence,prop:process-schedule-prob} together allow us to apply concentration bounds on a collection of events of the type in \cref{prop:process-schedule-prob}. This technique will be useful in our analysis of algorithms.

Intuitively, a larger parameter $\tau$ in \cref{def:random-delay} corresponds to a weaker guarantee on the scheduler. This is formalized as follows.

\begin{proposition}
  \label{prop:tau-monotone}
  Suppose the scheduler satisfies the $\tau$-random-delay property. Then, for any $\tau' \ge 2 \tau$, the scheduler also satisfies the $\tau'$-random-delay property.
\end{proposition}
\begin{proof}
  Let $\BK{\coin_{p, j}}_{p \in [P], j \ge 0}$ be the scheduling coins, with respect to which the scheduler satisfies the $\tau$-random-delay property. We define the scheduling coins $\BK{\coin'_{p, j}}_{p \in [P], j \ge 0}$ for the $\tau'$-random-delay property as follows. Every interval $I = [j \tau', \, (j + 1) \tau')$ of length $\ge 2 \tau$ contains an aligned subinterval $I' = [k \tau, \, (k + 1) \tau)$ of length $\tau$; we let $I'$ be the leftmost such subinterval. We then define $\coin'_{p, j} \defeq \coin_{p, k}$. It is easy to verify that these scheduling coins satisfy the desired properties of the $\tau'$-random-delay property.
\end{proof}

\subsection{Local Simulation and State Machine}
\label{sec:state-machine}

We now describe a way to simplify the analysis of algorithms by ignoring the details of local instructions.  We will define a \emph{reduced state machine} of an algorithm that keeps track of only the shared instructions.  Time bounds proven on the reduced state machine will also apply to the original algorithm.

Let $A$ be an algorithm that runs in the stochastic CRQW model.
For the sake of analysis, after applying a shared instruction $i$ in timestep $t$, we can immediately \emph{simulate} the execution of local instructions until the next shared instruction, and regard these instructions as all applied atomically in timestep $t$ together with the shared instruction $i$. We formalize this idea as follows.

\begin{definition}[Reduced State Machine]
  Let $A$ be an algorithm in the stochastic CRQW model represented as a list of instructions, including shared instructions, local instructions, and \Nop instructions that do nothing. The \defn{reduced state machine} of $A$ is an automaton $\tilde{A} = (S, \delta)$ where the set $S$ of states consists of all shared and \Nop instructions in $A$, and $\delta$ is the transition function.
  When the algorithm is at state $i$ and gets scheduled in timestep $t$, the instruction $i$ is invoked, which will be applied in timestep $t' \ge t$ and return a value $\mathrm{ret}_i$.
  Then, the state machine transitions to the next state $i_{\textup{next}}$ according to the transition function
  \[
    \delta: (i, \, \mathrm{ret}_i, \, M, \, r) \mapsto
    (i_{\textup{next}}, \, M_{\textup{next}}),
  \]
  where $M$ and $M_{\textup{next}}$ denote the local memory state of the algorithm before and after the transition, and $r$ denotes the private randomness of the algorithm.
\end{definition}

There is a natural correspondence between the execution of the original algorithm $A$ and the reduced state machine $\tilde{A}$---the state machine is scheduled in timestep $t$ if and only if the original algorithm is scheduled to invoke a shared instruction (or \Nop) in timestep $t$. Fixing any scheduler for the original algorithm $A$, this correspondence induces a scheduler for the reduced state machine $\tilde{A}$, which we call the \defn{induced scheduler}.

We use $\HistExt(t)$ to denote the history of the execution of the reduced state machine $\tilde{A}$ up to time $t$. $\HistExt(t)$ includes all information in $\HistTrue(t)$, but it may contain more information due to the local simulation---a local instruction applied in timestep $t$ in the original algorithm $A$ may take effect earlier than timestep $t$ in the reduced state machine $\tilde{A}$, since the transition containing that local instruction atomically happened together with the previous shared instruction (or \Nop).
On the other hand, if we know $\HistTrue(t)$ and the randomness $\Randomness$ of the algorithm, we already have enough information to run the local simulation to infer $\HistExt(t)$. Combining both directions, we know $(\HistTrue(t), \Randomness)$ and $(\HistExt(t), \Randomness)$ are equivalent in terms of the information they contain.

\begin{theorem}
  \label{thm:ind-scheduler-random-delay}
  Suppose the scheduler for the original algorithm $A$ satisfies the $\tau$-random-delay property,
  and let $K \ge 1$ be the maximum number of local instructions that can be applied between a pair of shared instructions (and/or \Nop) in the original algorithm $A$.
  Then, the induced scheduler for the reduced state machine $\tilde{A}$ satisfies the $(2 K + 1) \tau$-random-delay property.
\end{theorem}

\begin{proof}
  Let $\BK{\coin_{p, j}}_{p \in [P], j \ge 0}$ be the scheduling coins of the original scheduler. We define the scheduling coins $\BK{\coin'_{p, j}}_{p \in [P], j \ge 0}$ for the induced scheduler as follows:
  \[
    \coin'_{p, j} \defeq \ind\Bk*{\CoinSum_{p, [j \tau', \, (j + 1) \tau')} \ge K + 1} = \ind\Bk*{\sum_{k = 0}^{2K} \coin_{p, (2 K + 1) j + k} \ge K + 1}.
  \]
  Then, we verify the desired properties of the scheduling coins. Let $\tau' \defeq (2K + 1) \tau$ for convenience.
  \begin{itemize}
    \item Each $\coin'_{p, j}$ is a fair coin flip (i.e., it has expectation $1/2$) via direct calculation.\footnote{The event $\coin'_{p, j} = 1$ is equivalent to that, among $2K + 1$ fair coin flips $\BK{\coin_{p, (2 K + 1) j + k}}_{k = 0}^{2K}$, at least $K + 1$ of them are $1$. Then, $\E[\coin'_{p, j}] = 1/2$ follows by symmetry.}
    \item Conditioning on $\HistExt(j \tau')$ and \Randomness, which is equivalent to conditioning on $\HistTrue(j \tau')$ and \Randomness, we know that $\BK{\coin_{p, j'}}_{j' \ge j \cdot (2 K + 1)}$ are independent coin flips by \cref{def:random-delay}. Since all coins in $\BK{\coin'_{p, j'}}_{j' \ge j}$ are determined by disjoint subsets of $\BK{\coin_{p, j'}}_{j' \ge j \cdot (2 K + 1)}$, we know that $\BK{\coin'_{p, j'}}_{j' \ge j}$ are also independent coin flips.
    \item For every $p \in [P]$ and $j \ge 0$, if process $p$ is ready at time $j \tau'$ and there is $\coin'_{p, j} = 1$, we show that $p$ is scheduled at least once in $[j \tau', \, (j + 1) \tau')$.

    Recall that when $\coin'_{p, j} = 1$, we have
    $\CoinSum_{p, [j \tau', \, (j + 1) \tau')} \ge K$.
    By \cref{prop:coin-implies-scheduled}, we know that either $p$ is scheduled at least $K$ times in $[j \tau', \, (j + 1) \tau')$, or $p$ becomes not ready at some time $t' \in [j \tau', \, (j + 1) \tau')$. In the former case, $p$ invokes a shared instruction (or \Nop) in $[j \tau', \, (j + 1) \tau')$ because the original algorithm $A$ can apply at most $K$ local instructions in a row. The latter case can only happen if $p$ invokes a shared instruction in $[j \tau', \, (j + 1) \tau')$ and then becomes \emph{waiting}. Combining both cases, we know that $p$ is scheduled at least once in the reduced state machine $\tilde{A}$ during $[j \tau', \, (j + 1) \tau')$. \qedhere
  \end{itemize}
\end{proof}

\cref{thm:ind-scheduler-random-delay} states that the induced scheduler is $\tau'$-random-delay with $\tau' = O(K \tau)$, where $K$ is the maximum number of local instructions that can be applied between a pair of shared instructions (and/or \Nop) in the original algorithm $A$. In many natural algorithms, including all algorithms we will present in this paper, $K$ is a constant, since otherwise we can manually insert \Nop into the original algorithm $A$ to break the chain of local instructions into smaller pieces. Combining $K = O(1)$ with \cref{thm:ind-scheduler-random-delay} and \cref{prop:tau-monotone}, we know that there exists a global constant $\tau' = O(\tau)$ such that both the original scheduler for $A$ and the induced scheduler for $\tilde{A}$ satisfy the $\tau'$-random-delay property. As a slight abuse of notation, throughout the remainder of the paper, we will directly use $\tau$ to denote this $\tau'$ that gives random-delay properties to both schedulers. Similarly, we will directly use $\coin_{p, j}$ to denote the scheduling coins for the induced scheduler, and use $\CoinSum_{p, [t_1, \, t_2)}$ to denote the sum of scheduling coins for the induced scheduler. Since the conditions in \cref{def:random-delay} hold for the induced scheduler, \cref{prop:process-schedule-prob,prop:coin-implies-scheduled,prop:coin-sum-independence} all hold for the induced scheduler as well.

The reduced state machine $\tilde{A}$ provides a convenient way to analyze the original algorithm $A$ by ignoring the details of local instructions. Prior to the analysis, we will first specify the set $S$ of instructions that appear in the reduced state machine, which is usually the set of all shared instructions, but we also have the choice to add \Nop into the set. At any time $t$ in the execution, we say a process $p$ is \defn{pending on} instruction $i \in S$ if $p$ is waiting at state $i$ in the reduced state machine $\tilde{A}$. When we say a process $p$ is \defn{scheduled} in timestep $t$, it means $p$ is scheduled in the reduced state machine $\tilde{A}$ in timestep $t$, i.e., the process invokes the instruction $i \in S$ that it is pending on. Any performance guarantee of the algorithm proven under the view of the reduced state machine can be directly transferred to the original algorithm.

\subsection{From Processes to Operations}

When we study concurrent data structures, we will often conflate an operation $q$ with the process $p$ that is executing it, and transfer the notation/terminology for processes to the operations. We list some notable notations here.

We say an operation $q$ is \defn{waiting} at time $t$ if it has invoked an instruction that has not been applied. Otherwise, we say the operation is \defn{ready} to invoke an instruction at time $t$. We use $R_t$ to denote the set of \emph{ready} operations at time $t$.

In \cref{def:random-delay}, we defined scheduling coins with process IDs as subscripts. When we analyze concurrent data structures, it is more convenient to use scheduling coins with respect to operations. This is formalized as follows.

\begin{proposition}
  \label{prop:random-delay-for-operations}
  If the scheduler satisfies the $\tau$-random-delay property, then there are independent fair coin flips $\BK{\coin_{q, j}}_{q \in \N, \, j \ge 0}$ coupled with the execution of the data structure (called the \defn{scheduling coins} for operations) such that
  \begin{itemize}
    \item For every $j \ge 0$, all coins $\BK{\coin_{q, j'}}_{q \in \N, \, j' \ge j}$ are independent of the history $\HistExt(j \tau)$ and of the algorithm's randomness $\Randomness$.
    \item For every $q \in \N$ and $j \ge 0$, if operation $q$ is \emph{ready} at time $j \tau$ and there is $\coin_{q, j} = 1$, then $q$ is scheduled at least once in $[j \tau, \, (j + 1) \tau)$.
  \end{itemize}
\end{proposition}

\begin{proof}
  Let $p(q)$ denote the process that executes operation $q$.
  By \cref{def:random-delay}, the execution can be coupled with independent fair coin flips $\BK{\coin'_{p, j}}_{p \in [P], \, j \ge 0}$ such that
  \begin{itemize}
    \item For every $j \ge 0$, all coins $\BK{\coin'_{p, j'}}_{p \in [P], \, j' \ge j}$ are independent of the history $\HistExt(j \tau)$ and of the algorithm's randomness $\Randomness$.%
    \footnote{This step uses the fact that $(\HistExt(j \tau), \Randomness)$ contains exactly the same information as $(\HistTrue(j \tau), \Randomness)$.}
    \item For every $p \in [P]$ and $j \ge 0$, if process $p$ is \emph{ready} at time $j \tau$ and there is $\coin'_{p, j} = 1$, then $p$ is scheduled at least once in $[j \tau, \, (j + 1) \tau)$.
  \end{itemize}
  Then, we define the scheduling coins for operations as $\coin_{q, j} \defeq \coin'_{p(q), j}$ if operation $q$ is ongoing at time $j \tau$; otherwise, we let $\coin_{q, j}$ be a coin flip independent of the entire execution and of all other scheduling coins. It is easy to verify that these scheduling coins satisfy the desired properties in \cref{prop:random-delay-for-operations}.
\end{proof}

Throughout the remainder of this paper, we will exclusively use the scheduling coins $\BK{\coin_{q, j}}_{q \in \N, \, j \ge 0}$ for operations, and use $\CoinSum_{q, [t_1, \, t_2)}$ to denote the sum of scheduling coins for operation $q$ in time interval $[t_1, \, t_2 - \tau]$. Note that the scheduling coins for operations satisfy the same guarantee as \cref{def:random-delay}, with the only difference being the range of subscripts for the coins. Therefore, \cref{prop:coin-implies-scheduled,prop:coin-sum-independence,prop:process-schedule-prob} all hold for scheduling coins for operations, with processes replaced by operations.

\subsection{On the Definition of High-Probability Latency}
\label{sec:latency_def}

We discuss some subtleties related to defining a notion of latency for a given implementation.

Because we allow the user to be adaptive, the notion of ``each operation'' is not well defined a priori.  This is because whether and when an operation $q$ gets invoked can depend on the responses of earlier operations.
But if we condition on an operation $q$ being invoked, it may become impossible to derive any performance guarantee for $q$.
For example, the user could invoke a special operation $q$ only when they infer that the data structure is in a bad state, e.g., a very long instruction queue has been formed. In this case, conditioned on $q$ being invoked, $q$ is expected to incur very large latency.%
\footnote{We point out that the same challenge exists even in sequential data structures when the user is adaptive and has access to a clock. For example, in a sequential hash table, the user could query for an element $x$ only after detecting that inserting $x$ took a long time. In this case, the latency of the query is expected to be very large (conditioned on the query having been invoked).}

Another natural performance metric that one might think of is the \emph{expected latency} of each operation. However, attempts to define expected latency run into the same issue: we need to condition on the invocation of a particular operation. In fact, the issue is even worse for expectation than it is for high-probability bounds, because typically one wants to \emph{add up} the expected latency of different operations, and quantifying over time does not allow us to do this. In other words, while there exist definitions of expectation that we can analyze for our implementations in this paper, they are not definitions that we consider \emph{useful}.

\subsection{Conventions in Analysis}
\label{sec:conventions}

We conclude this section by introducing some conventions that we will follow in the analysis of the algorithms.

Recall that the entire execution is a complex random process, in which the algorithm, the scheduler, and the adaptive user all interact with each other. In the analysis, we will often prove statements of the form: conditioned on an arbitrary history $\HistExt(t)$, some property $X$ of the execution holds with a certain probability $p$. When we apply this statement, we may only condition on a shorter history $\HistExt(t')$ for some $t' \le t$, and the same property $X$ still holds with the same probability $p$ by the law of total probability. In other words, conditioning on less information only weakens the statement. In the rest of the paper, we omit this argument when it is clear from context---that is, when we switch from conditioning on a longer history to conditioning on a shorter one, we will not explicitly say that the rule of total probability was used.

We will also frequently prove that for any fixed time $t \ge 0$, some property $X(t)$ of the execution holds with high probability in $P$. A priori, this type of statement can only be applied when $t$ is a time that is fixed in advance, but we sometimes need to apply this statement to a time $t'$ that is a random variable determined by the random execution. In this case, if we can show a fixed time interval $[\l, r]$ of length $P^{O(1)}$ that contains $t'$ with high probability in $P$, then we can take a union bound of $X(t'')$ over all $t'' \in [\l, r]$ to show that $X(t')$ holds with high probability in $P$. Again, we will not repeat this argument; we will directly apply the statement to a random time $t'$ as long as it is shown to be bounded by a polynomial range with high probability in $P$.

\section{Read/Write Register}
\label{sec:register}

In this section, we study the \emph{read/write register} problem.

\begin{definition}
  A \defn{read/write register} (\defn{register} for short) is a data structure that maintains a single $\l$-bit integer $x$, and supports two types of operations:
  \begin{itemize}
    \item \Read: Return the current value of $x$.
    \item $\Write(\xNew)$: Overwrite $x$ with $\xNew$.
  \end{itemize}
\end{definition}

The most straightforward implementation of a read/write register is to use provided atomic instructions \Load and \Store. It is easy to see that this implementation is linearizable. However, the time complexity for each operation is problematic when $P$ operations attempt to apply \Store instructions simultaneously. According to our stochastic CRQW model, all $P$ \Store instructions will wait in the queue and apply one by one, causing $\Theta(P)$ latency on average for the $P$ operations. Since each operation can take up to $\Theta(P)$ time, this implementation is not \emph{scalable}.

Our main result in this section is a scalable implementation of a read/write register that takes $O(\log P)$ time for each operation with high probability in $P$, as stated in the following theorem.

\begin{theorem}[Theorem~\ref{thm:register-main-intro}, restated]
  \label{thm:register-main}
  We can construct a linearizable, wait-free $\l$-bit read/write register that works under the stochastic CRQW model with adaptive inputs, such that:
  \begin{itemize}
    \item Each operation has latency at most $O(\log P)$ with high probability in $P$.
    \item The data structure uses $O(1)$ machine words with word size $w \ge \l + \eps \log P$ bits in the shared memory, where $\eps > 0$ is a fixed constant, and $O(1)$ words of local memory for each operation.
  \end{itemize}
\end{theorem}

\subsection{Algorithm Overview}
\label{sec:register-algorithm}

In this subsection, we describe our algorithm for the read/write register problem, and analyze its linearizability and latency.

Our algorithm maintains a single cell in the shared memory, called the \SharedCell, which stores a pair $(x, f)$. Here, $x$ represents the current value of the register, and $f$ is a random \emph{fingerprint} of $\eps \log P$ bits, where $\eps > 0$ is a fixed constant. Intuitively, the fingerprint allows operations to detect whether the cell has been overwritten by another operation; when this happens, an operation can abort by linearizing itself just before the successful overwrite.

The \Read operation is straightforward: we simply use \Load to read the value $(x, f)$ of the shared cell, and return $x$. Regardless of the state of the data structure, the \Read operation completes in $O(1)$ time in expectation, and $O(\log P)$ time with high probability in $P$, according to \cref{prop:random-delay-for-operations,prop:process-schedule-prob}.

The \Write operation uses a \emph{back-on} strategy: each operation starts with a small probability $p = 1 / P^4$ of invoking a \StoreRandom instruction to overwrite the shared cell immediately, while other operations wait and try to detect changes to the shared cell. Once a change is detected, the operation can abort. Otherwise, the operation multiplies $p$ by a constant factor and repeats. The pseudocode for the \Write operation is shown in \cref{alg:shared-register-write}, where $c \gg \max\{\eps^{-1}, \tau\}$ on line~\ref{line:register-write/defn-c} is a sufficiently large constant.
Line~\ref{line:reg/W} uses the special instruction \StoreRandom to atomically sample a uniformly random $\eps \log P$-bit fingerprint $\fNew$ and overwrite \SharedCell with the pair $(\xNew, \fNew)$.

There are three shared instructions in \cref{alg:shared-register-write}, which we refer to as \instS, \instR, and \instW, respectively. There are a constant number of local instructions between each pair of shared instructions. As mentioned in \cref{sec:state-machine}, we adopt the view of the \emph{reduced state machine} throughout this section: each ongoing operation $q$ is pending on one of \instS, \instR, and \instW; when $q$ is \emph{scheduled}, it invokes the shared instruction that it is pending on, and upon application of the shared instruction, all subsequent local instructions are applied immediately, and the operation $q$ \emph{transitions} to the next shared instruction. See \cref{fig:state-machine-register} for an illustration of the reduced state machine.

\subsection{Linearizability and Wait-Freedom}
\label{sec:register/linearizability}

\begin{restatable}{lemma}{RegisterLinearizable}
  \label{lem:register-linearizable}
  \cref{alg:shared-register-write} is a linearizable, wait-free implementation of a read/write register.
\end{restatable}

\begin{proof}
  Consider an infinite execution $\alpha$ of the read/write register. We linearize the operations in $\alpha$ as follows:
  \begin{itemize}
    \item $\Read$ operations are linearized at the point where they $\Load$ the value of the read/write register.
    \item A $\Write$ operation that applies the $\StoreRandom$ instruction in line~\ref{line:reg/W}
    is linearized when that instruction is applied.
    It is convenient to refer to such writes as \defn{successful writes}.
    \item A $\Write$ operation that does not execute line~\ref{line:reg/W} must return in line~\ref{line:reg/7},
    after executing in line~\ref{line:reg/5} a $\Load$ instruction that returned $(\xNow, \fNow)$
    where $\fNow \neq \fOld$.
    We refer to such writes as \defn{aborted writes}.

    Let $q$ be an aborted \Write operation.
    Prior to executing line~\ref{line:reg/5},
    $q$ executes line~\ref{line:reg/3},
    where it $\Load$s values $(\xOld, \fOld)$
    from the shared cell \SharedCell.
    As $q$ returns without applying a $\StoreRandom$ instruction (line~\ref{line:reg/W})
    and $\fOld \neq \fNow$,
    between the time that $q$ executes lines~\ref{line:reg/3} and~\ref{line:reg/5}
    the value of the shared cell must have been changed by some other successful \Write operation $q' \neq q$ that applied a $\StoreRandom$ instruction.
    We linearize $q$ immediately before $q'$ applies its $\StoreRandom$ instruction.
    (There may be multiple successful \Write{}s between the execution of lines~\ref{line:reg/3} and~\ref{line:reg/5} of $q$; in this case, we choose $q'$ arbitrarily.)
  \end{itemize}

  Incomplete \Read operations and \Write operations that have not applied line~\ref{line:reg/W} are not linearized.
  Incomplete \Write operations that have applied line~\ref{line:reg/W} are linearized as above.

  To show that this linearization is proper, consider a \Read operation $r$ that returns a value $x$.
  Let $t$ be the time when $r$ executes the \Load
  instruction that returns $(x, f)$ for some $f$.
  If $(x, f)$ is the initial content of \SharedCell and no \StoreRandom instruction has been applied yet, then the \Read is linearized before all \Write{}s,
  so it respects the semantics.
  Otherwise, let $w$ be the last \Write operation that applied the \StoreRandom instruction to \SharedCell at some time $t' < t$.
  The \Read $r$ is linearized at time $t$, and the \Write $w$ is linearized at time $t'$.
  No \StoreRandom instruction is applied between time $t'$ and time $t$ (by choice of $w$ and the fact that $r$ reads the value $(x, f)$ from \SharedCell).
  Thus, no \Write operation is linearized between time $t'$ and time $t$,
  and the \Read returns the value written by the last write that precedes it, as required.

  Wait-freedom is also easy to see. \Read operations return after executing a single instruction. \Write operations return upon reaching line~\ref{line:reg/abort}, or upon reading $\fNow \neq \fOld$ in line~\ref{line:reg/5} and returning in line~\ref{line:reg/7}.
  Each iteration through the loop increases $p$ by a constant multiplicative factor, so $p$ reaches 1 after $O(\log P)$ iterations,
  at which time the operation executes line~\ref{line:reg/W} and returns, if it has not aborted earlier by executing line~\ref{line:reg/7}.
\end{proof}

\subsection{Latency Analysis}
\label{sec:register-analysis}

In this subsection, we analyze the latency of our read/write register.
We start by defining some notation and terminology.
\begin{itemize}
  \item An ongoing operation is \defn{active} if it is pending on \instR. We use $A_t$ to denote the set of active operations at time $t$.
  \item The local variable $p$ in \cref{alg:shared-register-write} is called the \defn{invocation probability} of the operation. We use $p^{(q)}_t$ to denote the invocation probability of operation $q$ at time $t$.
\end{itemize}

Our proof tracks the \defn{potential} of the data structure, which is defined as
\[
  \phi_t \defeq \sum_{q \in A_t} p^{(q)}_t.
\]

The rest of the proof consists of three parts. First, in \cref{sec:register/potential}, we show a high-probability upper bound on the potential $\phi_t$ at any fixed time $t$. Based on this bound, in \cref{sec:register/busy-interval}, we further bound the lengths of \emph{busy intervals}---maximal intervals where the instruction queue of \SharedCell is non-empty. Finally, in \cref{sec:register/latency}, we derive the desired latency guarantee from the results on busy intervals.

\subsubsection{Bounding the Potential}
\label{sec:register/potential}

We start by analyzing how the potential $\phi_t$ changes over time. To develop intuition, imagine that $\phi_t = \Omega(\log P)$ at some time $t$. Given this, we will show that with high probability in $P$, one of the active operations will invoke a \StoreRandom instruction within the next $O(1)$ timesteps, changing the fingerprint of \SharedCell to a random string. Within another $O(1)$ timesteps, other active operations are likely to detect the change of fingerprint and abort the operation without invoking \StoreRandom{}s, which causes the potential to drop significantly. The main objective of this part is to show that the latter effect prevents the potential from growing beyond $\Theta(\log P)$ with high probability in $P$.

\begin{claim}
  \label{clm:potential-growth-bound}
  For any integer $k \ge 1$, $\phi_{t + k} \le (\phi_t + c / P^3) \cdot (1 + 1 / c)^k$.
\end{claim}
\begin{proof}
  At each timestep, the invocation probability of any active operation can increase by a factor of at most $1 + 1 / c$, and new active operations contribute at most $P \cdot (1 / P^4) = 1 / P^3$ to the potential. Therefore,
  \[
    \phi_{t + 1} \le \phi_t \cdot (1 + 1 / c) + 1 / P^3.
  \]
  Applying this recursively for $k$ steps, we obtain
  \begin{align*}
    \phi_{t + k} & \le \phi_t \cdot (1 + 1 / c)^k + \frac{1}{P^3} \sum_{i=0}^{k-1} (1 + 1 / c)^i      \\
      & = \phi_t \cdot (1 + 1 / c)^k + \frac{1}{P^3} \cdot \frac{(1 + 1 / c)^k - 1}{1 / c} \\
      & \le \Bigl(\phi_t + \frac{c}{P^3}\Bigr) \cdot (1 + 1 / c)^k. \qedhere
  \end{align*}
\end{proof}

\begin{lemma}
  \label{lem:store-probability}
  Conditioned on $\HistExt(t)$,
  \[
    \Pr\Bk[\Big]{\bk[\big]{\phi_{t + 2c} \ge \phi_t / 100} \land \bk[\big]{\textup{no \StoreRandom is applied in time interval } [t, \, t + 2c)}} \le 2^{-\Omega(\phi_t)}.
  \]
\end{lemma}
\begin{proof}
  We partition the set $A_t$ of active operations at time $t$ into four disjoint groups according to their behavior in the time interval $[t, \, t + c)$:
  \begin{enumerate}
    \item $\ASlow$: Operations that do not apply any shared instruction before time $t + c$.
    \item $\AAbort$: Operations that apply at least one \instR before time $t + c$, where the first \instR observes a changed fingerprint ($\fNow \ne \fOld$), causing these operations to abort.
    \item $\AStay$: Operations that apply at least one \instR before time $t + c$, where the first \instR observes $\fNow = \fOld$, and they remain pending on \instR after the first \instR.
    \item $\ATransfer$: Operations that apply at least one \instR before time $t + c$, where the first \instR observes $\fNow = \fOld$, and the operation transitions to \instW after the first \instR. Note that these operations are not required to actually invoke \instW in $[t, \, t + c)$.
  \end{enumerate}
  Let $\pSlow, \pAbort, \pStay, \pTransfer$ denote the sum of invocation probabilities of operations in the respective groups at time $t$. We have $\pSlow + \pAbort + \pStay + \pTransfer = \phi_t$.

  \begin{claim}
    $\phi_{t + c} \le \pSlow + e \cdot \pStay + e / P^3$.
  \end{claim}
  \begin{proof}
    By definition, $\phi_{t + c}$ is the sum of invocation probabilities of all operations that are active at time $t + c$. The operations active at time $t + c$ consist of:
    \begin{itemize}
      \item Operations from $A_t$ that remain active at time $t + c$ (i.e., $A_t \cap A_{t + c}$).
      \item Operations that newly become active in $(t, \, t + c]$ (i.e., $A_{t + c} \setminus A_t$).
    \end{itemize}

    At time $t + c$, operations in groups $\AAbort$ and $\ATransfer$ are no longer active (they have either aborted or transitioned to \instW). Therefore, $A_t \cap A_{t + c} \subseteq \ASlow \cup \AStay$. Thus, the total invocation probability of operations in $A_t \cap A_{t + c}$ at time $t + c$ is at most
    \[
      \pSlow + (1 + 1/c)^c \cdot \pStay \le \pSlow + e \cdot \pStay,
    \]
    where the factor $(1 + 1/c)^c$ accounts for the growth of invocation probabilities over $c$ timesteps.

    The operations that newly become active in $(t, \, t + c]$ start with invocation probability $1/P^4$ and can grow by a factor of at most $(1 + 1/c)^c \le e$. Since at most $P$ new operations can become active, their total invocation probability at time $t + c$ is at most
    \[
      P \cdot (1/P^4) \cdot e = e/P^3.
    \]
    Adding two parts together yields $\phi_{t + c} \le \pSlow + e \cdot \pStay + e/P^3$, as claimed.
  \end{proof}

  Combining this claim with \cref{clm:potential-growth-bound}, we have
  \[
    \phi_{t + 2c} \le \Bigl(\phi_{t + c} + \frac{c}{P^3}\Bigr) \cdot (1 + 1 / c)^c \le e \cdot \pSlow + e^2 \cdot \pStay + O(1 / P^3).
    \numberthis \label{eq:phi_t+2c}
  \]

  We then prove the lemma by bounding the probability of the following \defn{bad event}:
  \begin{itemize}
    \item $\phi_{t + 2c} \ge \phi_t / 100$, and
    \item there is no \StoreRandom instruction applied in $[t, \, t + 2c)$.
  \end{itemize}
  When the bad event happens, according to \eqref{eq:phi_t+2c}, we have
  \[
    \phi_t / 100 \le \phi_{t + 2c} \le e \cdot \pSlow + e^2 \cdot \pStay + O(1 / P^3).
  \]
  This implies
  \[
    \pSlow + \pStay \ge \frac{\phi_t}{100 e^2} - O(1 / P^3) \ge \phi_t / 1000,
    \numberthis \label{eq:pSlow+pStay}
  \]
  where the last inequality holds because we can assume $\phi_t \ge 1$ without loss of generality (otherwise the lemma is trivial). It further implies that either $\pSlow \ge \phi_t / 2000$ or $\pStay \ge \phi_t / 2000$. Based on this condition, there are three cases for the bad event to occur:
  \begin{enumerate}
    \item $\pSlow \ge \phi_t / 2000$;
    \item $\pStay \ge \phi_t / 2000$ and $|\ATransfer| < \phi_t / 3000$;
    \item $|\ATransfer| \ge \phi_t / 3000$ and no \StoreRandom is applied in $[t, \, t + 2c)$.
  \end{enumerate}
  In the rest of the proof, we will bound the probability of these cases separately.

  \paragraph{Case 1: $\pSlow \ge \phi_t / 2000$.}
  According to \cref{prop:process-schedule-prob}, conditioned on $\HistExt(t)$, every operation $q \in R_t$ has a probability of at least $1 - 2^{-\Omega(c)}$ to invoke at least one shared instruction in $[t, \, t + c)$, independently of other operations.%
  \footnote{Throughout the paper, by saying ``event $X$ has probability $p$ to occur independently of random variables $Y_1, \dots, Y_n$'', we mean that conditioned on an arbitrary realization of $Y_1, \dots, Y_n$, the probability of $X$ occurring is at least $p$.}
  Taking a summation over all active operations $A_t$, the expectation of $\pSlow$ is at most $\phi_t \cdot 2^{-\Omega(c)} \le \phi_t / 3000$. By a Chernoff bound over operations in $A_t$, we have
  \[
    \Pr[\pSlow \ge \phi_t / 2000] \le 2^{-\Omega(\phi_t)}.
    \numberthis \label{eq:pSlow}
  \]

  \paragraph{Case 2: $\pStay \ge \phi_t / 2000$ and $|\ATransfer| < \phi_t / 3000$.}
  Fixing an operation $q$ and conditioning on $q \in \AStay \cup \ATransfer$, i.e., $q$ applies an \instR and observes $\fNow = \fOld$, the probability that it immediately transitions to \instW (i.e., $q \in \ATransfer$) is exactly $p^{(q)}_t$, and these events are independent across different operations. Over all operations in $\AStay \cup \ATransfer$, the expected number of operations that immediately transition to \instW is exactly $\pStay + \pTransfer$. By a Chernoff bound, we have
  \[
    \Pr\Bk[\Big]{|\ATransfer| < \phi_t / 3000 \;\Big|\; \pStay + \pTransfer \ge \phi_t / 2000} \le 2^{-\Omega(\phi_t)}.
  \]
  This implies
  \[
    \Pr\Bk[\Big]{\bk[\big]{|\ATransfer| < \phi_t / 3000} \land \bk[\big]{\pStay \ge \phi_t / 2000}} \le 2^{-\Omega(\phi_t)}.
    \numberthis \label{eq:pStay+pTransfer}
  \]

  \paragraph{Case 3: $|\ATransfer| \ge \phi_t / 3000$ and no \StoreRandom is applied in $[t, \, t + 2c)$.}
  Similar to Case 1, every operation $q \in \ATransfer$ has an independent probability of at least $1 - 2^{-\Omega(c)} \ge 1/2$ to apply \instW before time $t + 2c$. Therefore,
  \[
    \Pr\Bk[\big]{\text{no \StoreRandom is applied in } [t, \, t + 2c)} \le 2^{-|\ATransfer|}.
  \]
  As a special case,
  \begin{align*}
      & \phantom{\le{}} \Pr\Bk[\Big]{\bk[\big]{\text{no \StoreRandom is applied in } [t, \, t + 2c)} \land \bk[\big]{|\ATransfer| \ge \phi_t / 3000}} \\
      & \le \Pr\Bk[\Big]{\text{no \StoreRandom is applied in } [t, \, t + 2c) \;\Big|\; |\ATransfer| \ge \phi_t / 3000}
    \le 2^{-\Omega(\phi_t)}.
    \numberthis \label{eq:no-store}
  \end{align*}

  Combining \eqref{eq:pSlow}, \eqref{eq:pStay+pTransfer}, and \eqref{eq:no-store}, we know that the probability of the bad event is at most $2^{-\Omega(\phi_t)}$, which completes the proof.
\end{proof}

\begin{lemma}
  \label{lem:potential-drop-after-store}
  Suppose a \StoreRandom instruction is applied in $[t - 2c, \, t)$. Conditioned on $\HistExt(t)$,
  \[
    \Pr\Bk[\big]{\phi_{t + c} \le \phi_t / 200} \ge 1 - 2^{-\Omega(\phi_t)} - O(P^{-\eps}).
  \]
\end{lemma}
\begin{proof}
  The lemma is conditioned on a \StoreRandom instruction being applied in the time interval $[t - 2c, \, t)$. Let $t_s \in [t - 2c, \, t)$ be the timestep when this \StoreRandom is applied. This instruction writes a new, uniformly random fingerprint to the shared cell.

  We partition the set $A_t$ of active operations at time $t$ into three disjoint groups:
  \begin{itemize}
    \item $\ASlow$: Operations that were active at time $t$ and did not apply any shared instruction in $[t, \, t + c)$.
    \item $\AOld$: Operations that apply at least one \instR in $[t, \, t + c)$ and started (applied their \instS instructions) before time $t - 2c$.
    \item $\ANew$: Operations that apply at least one \instR in $[t, \, t + c)$ and started after time $t - 2c$.
  \end{itemize}
  Let $\pSlow, \pNew, \pOld$ be the total invocation probabilities at time $t$ for operations in the three groups, respectively. We bound the potential $\phi_{t + c}$ by analyzing the contributions of these three groups separately.

  \paragraph{$\ASlow$.}
  According to \cref{prop:process-schedule-prob}, each operation in $A_t$ has probability at least $1 - 2^{-\Omega(c)}$ to invoke at least one shared instruction in $[t, \, t + c)$, independently of other operations. By a Chernoff bound over all operations in $A_t$, we have
  \[
    \Pr[\pSlow \ge \phi_t / 1000] \le 2^{-\Omega(\phi_t)}.
    \numberthis \label{eq:pSlow-2}
  \]
  Operations in $\ASlow$ remain active at time $t + c$, which contributes at most $\phi_t / 1000$ to the potential $\phi_{t + c}$ with probability $1 - 2^{-\Omega(\phi_t)}$.

  \paragraph{$\AOld$.}
  Operations in $\AOld$ started before time $t - 2c$. Since a \StoreRandom was applied after time $t - 2c$, all operations in $\AOld$ applied their \instS instructions before this \StoreRandom was applied. Therefore, their $\fOld$ values were read before time $t_s$ and were not affected by the \StoreRandom. When such an operation $q \in \AOld$ applies its first \instR in $[t, \; t + c)$, it observes the current fingerprint, which was written by a \StoreRandom applied in some timestep in $[t_s, \; t + c)$. Operation $q$ then aborts unless the observed fingerprint collides with $\fOld$. The collision implies that there exists a timestep $t' \in [t_s, \; t + c)$ such that (1) a \StoreRandom was applied in timestep $t'$, and (2) the fingerprint written by the \StoreRandom in timestep $t'$ collides with $\fOld$. Each $t'$ satisfies (2) with probability $P^{-\eps}$, so by a union bound, the probability of the collision is at most $3c \cdot P^{-\eps} = O(P^{-\eps})$.%
  \footnote{Although in every timestep $t'$ when a \StoreRandom is applied, the fingerprint changes to a uniformly random value, we cannot infer that the fingerprint observed by operation $q$'s first \instR in $[t, \; t + c)$ is uniformly random, because the timestep that operation $q$ applies its first \instR in $[t, \; t + c)$ is a random variable. The union bound addresses this issue.}
  It is worth noting that these collision events are \emph{not} independent across different operations, as they can depend on the same random fingerprint.

  Let $\ACont \subseteq \AOld$ be the set of operations in $\AOld$ that experience a fingerprint collision at the first \instR in $[t, \; t + c)$, so they did not abort immediately. We let $\pCont$ be the total invocation probability at time $t$ for operations in $\ACont$. By linearity of expectation,
  \[
    \E[\pCont] = \pOld \cdot P^{-\eps} \le \phi_t \cdot P^{-\eps}.
  \]
  By Markov's inequality,
  \[
    \Pr[\pCont \ge \phi_t/1000] \le \frac{\E[\pCont]}{\phi_t/1000} = \frac{1000 \phi_t P^{-\eps}}{\phi_t} = 1000 P^{-\eps} = O(P^{-\eps}).
    \numberthis \label{eq:pCont-2}
  \]
  Operations in $\ACont$ might remain active at time $t + c$, and their invocation probabilities grow by a factor of at most $e$ during $[t, \, t + c)$. They contribute at most $\phi_t \cdot e / 1000$ to the potential $\phi_{t + c}$ with probability $1 - O(P^{-\eps})$. Operations in $\AOld \setminus \ACont$ are already aborted before time $t + c$.

  \paragraph{$\ANew$.}
  Operations in $\ANew$ started after time $t - 2c$, so by time $t + c$ they have been active for at most $3c$ timesteps. The number of operations in this group is at most $P$, and each starts with invocation probability $1/P^4$. Their invocation probabilities can grow by a factor of at most $(1 + 1/c)^{3c} \le e^3$ over these $3c$ timesteps. Thus, the total potential contribution from $\ANew$ at time $t + c$ is at most
  \[
    P \cdot (1/P^4) \cdot e^3 = e^3/P^3 = O(1/P^3).
    \numberthis \label{eq:pNew-2}
  \]

  \paragraph{Combining the bounds.}
  Combining the bounds from the three groups, we know that with probability at least $1 - 2^{-\Omega(\phi_t)} - O(P^{-\eps})$,
  \[
    \sum_{q \in A_t \cap A_{t + c}} p^{(q)}_{t + c} \le \phi_t \cdot \frac{e + 1}{1000} + O(1 / P^3).
    \numberthis \label{eq:phi_t+c_partial}
  \]
  Similar to \eqref{eq:pNew-2}, the total invocation probability of operations in $A_{t + c} \setminus A_t$---which applied their \instS instructions after time $t$---is at most $O(1/P^3)$. This combined with \eqref{eq:phi_t+c_partial} gives
  \[
    \phi_{t + c} \le \phi_t \cdot \frac{e + 1}{1000} + O(1 / P^3) \le \phi_t / 200,
  \]
  where the last inequality holds because we can assume $\phi_t \ge 1$ without loss of generality (otherwise the probability bound is trivial). This completes the proof.
\end{proof}

Combining \cref{lem:store-probability,lem:potential-drop-after-store}, we obtain the following corollary.

\begin{corollary}
  \label{cor:potential-falling-probability}
  Conditioned on $\HistExt(t)$,
  \[
    \Pr[\phi_{t + 3c} \le \phi_t / 20] \ge 1 - 2^{-\Omega(\phi_t)} - O(P^{-\eps}).
  \]
\end{corollary}
\begin{proof}
  By \cref{lem:store-probability}, with probability at least $1 - 2^{-\Omega(\phi_t)}$, at least one of the following two events occurs:
  \begin{itemize}
    \item $\phi_{t + 2c} < \phi_t / 100$. In this case, by \cref{clm:potential-growth-bound},
    \[
      \phi_{t + 3c} \le \phi_{t + 2c} \cdot e + O(1 / P^3) \le \phi_t \cdot e / 100 + O(1 / P^3) \le \phi_t / 20,
    \]
    where the last inequality holds because we can assume $\phi_t \ge 1$ without loss of generality.

    \item $\phi_{t + 2c} \ge \phi_t / 100$, and there is a \StoreRandom instruction applied in $[t, \, t + 2c)$. By \cref{lem:potential-drop-after-store} applied at time $t + 2c$, we have
    \[
      \phi_{t + 3c} \le \phi_{t + 2c} / 200
    \]
    with probability at least
    \[
      1 - 2^{-\Omega(\phi_{t + 2c})} - O(P^{-\eps}) \ge 1 - 2^{-\Omega(\phi_t)} - O(P^{-\eps}).
    \]
    By \cref{clm:potential-growth-bound}, $\phi_{t + 2c} \le \phi_t \cdot e^2 + O(1/P^3)$. Thus,
    \[
      \phi_{t + 3c} \le \frac{\phi_t \cdot e^2 + O(1/P^3)}{200} \le \phi_t / 20.
    \]
  \end{itemize}
  Combining the two cases concludes the proof.
\end{proof}

\begin{lemma}
  \label{lem:potential-upper-bound}
  For any fixed time $t \ge 0$, conditioned on $\HistExt(t - P)$,\footnote{For $t' < 0$, we define $\HistExt(t')$ to be the empty history.} $\phi_t \le O(\log P)$ with high probability in $P$.
\end{lemma}

\begin{proof}
  Let $\phiThreshold = \Theta(\log P)$ be a threshold such that for $\phi_t \ge \phiThreshold$, the probability bound $1 - 2^{-\Omega(\phi_t)} - O(P^{-\eps})$ in \cref{cor:potential-falling-probability} is at least $1 - O(P^{-\eps})$.

  Let $K \ge c$ be a sufficiently large constant, and suppose $\phi_t \ge 2^K \cdot \phiThreshold$ for some time $t$ in the execution. We let $t_0 \in [t - P, \, t)$ be the earliest time where the potential stays at least $\phiThreshold$ during time interval $[t_0, \, t]$. By definition, we have either $\phi_{t_0 - 1} < \phiThreshold$ or $t_0 = t - P$.

  We then partition the time interval $[t_0, \, t)$ into disjoint intervals of length $3c$, each being $[t_i, \, t_{i+1}) \defeq \bigl[t_0 + 3ci, \, t_0 + 3c(i+1)\bigr)$ for $0 \le i < k \defeq \floor[\big]{(t - t_0) / 3c}$. Note that the final $< 3c$ timesteps may not belong to any interval due to rounding.

  We say an interval $[t_i, \, t_{i+1})$ is \defn{falling} if $\phi_{t_{i+1}} \le \phi_{t_i} / 20$. Otherwise, we have
  $\phi_{t_i} / 20 \le \phi_{t_{i+1}} \le \phi_{t_i} \cdot (e^3 + O(1/P^3)) \le \phi_{t_i} \cdot 30$,
  and we say the interval is \defn{rising}.
  Conditioned on $\HistExt(t_i)$, according to \cref{cor:potential-falling-probability}, the probability that the interval $[t_i, \, t_{i+1})$ is falling is at least $1 - 2^{-\Omega(\phi_t)} - O(P^{-\eps}) = 1 - O(P^{-\eps})$ according to the definition of $\phiThreshold$.

  This definition produces a sequence of $k \le P / 3c$ symbols in \{\Rising, \Falling\}. Let $\kRise$ and $\kFall$ denote the number of rising and falling intervals among the $k$ intervals. We derive relations between $\kRise$ and $\kFall$ by considering two cases in the definition of $t_0$.

  \begin{itemize}
    \item $\phi_{t_0 - 1} < \phiThreshold$. In this case, the potential has increased by a factor of $O(2^K)$ from the beginning of the $k$ intervals (time $t_0$) to the end of these intervals (time $t_k = t_0 + 3ck > t - 3c$), so we have
    \begin{align*}
      \kRise \log 30 - \kFall \log 20 & \ge K - O(1) \ge K / 2   \\
      \kRise \cdot 5 - \kFall \cdot 4 & \ge K / 2                \\
      \kRise                          & \ge \kFall / 2 + K / 10.
      \numberthis \label{eq:kRise-kFall-1}
    \end{align*}
    \item $t_0 = t - P$. In this case, we have $\kRise + \kFall = k = \floor{P / 3c} > P / 4c$; by the trivial bound $\phi_{t_0} \le P$, we know that the potential has decreased by a factor of at most $O(P)$ from $t_0$ to $t_k$, so we have
    \begin{align*}
      \kRise \log 30 - \kFall \log 20   & \ge -\log P - O(1) > -(\kRise + \kFall) / 1000 \\
      \kRise \cdot (\log 30 + 1 / 1000) & \ge \kFall \cdot (\log 20 + 1 / 1000)          \\
      \kRise                            & \ge \kFall / 2.
    \end{align*}
    Combined with $\kRise + \kFall \ge P / 4c$, we have
    \[
      \kRise \ge \frac{P / 4c}{3} = P / 12c.
      \numberthis \label{eq:kRise-kFall-2}
    \]
  \end{itemize}

  The event of $\phi_t \ge 2^K \cdot \phiThreshold$ is witnessed by one such sequence of \{\Rising, \, \Falling\} that satisfies \eqref{eq:kRise-kFall-1} for $t_0 > t - P$ or \eqref{eq:kRise-kFall-2} for $t_0 = t - P$. Let $I_1, \dots, I_k \in \BK{\Rising, \, \Falling}$ be a sequence with $\kRise$ \Rising{}s and $\kFall$ \Falling{}s. Applying \cref{cor:potential-falling-probability} repeatedly, we know that conditioned on $\HistExt(t_0) \supseteq \HistExt(t_0 - P)$, the probability that $I_1, \dots, I_k$ occurs as a valid witness is bounded by $O(P^{-\eps})^{\kRise}$.
  Taking a summation over all valid sequences, we obtain
  \begin{align*}
      & \Pr\Bk[\big]{\exists \, \textup{witness with } t_0 > t - P}                                                                \\
      & \le \sum_{\kFall \ge 0} \;\; \sum_{\kRise \ge \kFall/2 + K/10} \binom{\kFall + \kRise}{\kFall} \cdot O(P^{-\eps})^{\kRise} \\
      & \le \sum_{\kFall \ge 0} \;\; \sum_{\kRise \ge \kFall/2 + K/10} 2^{\kFall + \kRise} \cdot O(P^{-\eps})^{\kRise}             \\
      & \le \sum_{\kRise \ge K/10} \;\; \sum_{\kFall = 0}^{2\kRise} 2^{\kFall + \kRise} \cdot O(P^{-\eps})^{\kRise}                \\
      & \le \sum_{\kRise \ge K/10} O(2^{3\kRise}) \cdot O(P^{-\eps})^{\kRise}                                                      \\
      & = \sum_{\kRise \ge K/10} O(P^{-\eps})^{\kRise}                                                                             \\
      & \le P^{-\Omega(K \eps)} = 1 / \poly(P),
    \numberthis \label{eq:witness-pr-1}
  \end{align*}
  where the last inequality holds because $K$ is a sufficiently large constant. Similarly, we obtain
  \begin{align*}
      & \Pr\Bk[\big]{\exists \, \textup{witness with } t_0 = t - P}                              \\
      & \le \sum_{\kRise \ge P / 12c} \binom{\floor{P / 3c}}{\kRise} \cdot O(P^{-\eps})^{\kRise} \\
      & \le \sum_{\kRise \ge P / 12c} 2^{O(\kRise)} \cdot O(P^{-\eps})^{\kRise}                  \\
      & \le P^{-\Omega(P)} < 1 / \poly(P).
    \numberthis \label{eq:witness-pr-2}
  \end{align*}
  Combining \eqref{eq:witness-pr-1} and \eqref{eq:witness-pr-2} yields the desired bound.
\end{proof}

\subsubsection{Bounding the Length of Busy Intervals}
\label{sec:register/busy-interval}

\begin{definition}
  \label{def:reg/busy-interval}
  A \defn{busy interval} $[t_1, t_2)$ is a \emph{maximal} interval where
  (1) one \StoreRandom is applied on \SharedCell in every timestep $t \in [t_1, t_2)$;
  (2) the instruction queue of \SharedCell is non-empty at each \emph{internal} time $t \in (t_1, t_2)$.
\end{definition}

In the following, we use the potential upper bound from \cref{lem:potential-upper-bound} to bound the length of the busy interval.

\begin{lemma}
  \label{lem:busy-interval-length}
  For any fixed time $t \ge 0$, conditioned on $\HistExt(t - \Theta(P))$, the length of the busy interval starting at $t$ (if exists) is at most $O(\log P)$ with high probability in $P$.
\end{lemma}

\begin{proof}
  Let $C \gg \eps^{-1}$ be a sufficiently large constant.
  We bound the length of the busy interval by counting the number of \StoreRandom{}s that get applied in the busy interval. Specifically, we define the following sets:
  \begin{itemize}
    \item For $j \ge 1$, $S_j$ is the set of operations that applied \instR in timestep $t - j$.
    \item $T_j \subseteq S_j$ is the set of operations $q$ that applied \instR in timestep $t - j$, and according to its private randomness, $q$ will transition to \instW exactly after its next $C$ \instR instructions, unless it quits earlier due to a changed fingerprint being detected.
  \end{itemize}

  We start by bounding the size of $T_j$.

  \begin{claim}
    \label{clm:Tj-size-bound}
    Conditioned on $\HistExt(t - \Theta(P))$, with high probability in $P$, we have $|T_j| \le \Theta(\log P)$ for all $j \in [1, \, P]$.
  \end{claim}

  \begin{proof}
    When an operation $q \in S_j$ applies \instR in timestep $t - j$, it has probability $p^{(q)}_{t - j}$ to transition to \instW immediately, and the probability that it transitions to \instW within its next $C$ \instR instructions is at most
    \begin{align*}
        & 1 - \bk[\big]{1 - p^{(q)}_{t - j}} \cdot \bk[\big]{1 - (1 + 1/c) \cdot p^{(q)}_{t - j}} \cdot \bk[\big]{1 - (1 + 1/c)^2 \cdot p^{(q)}_{t - j}} \cdots \bk[\big]{1 - (1 + 1/c)^{C-1} \cdot p^{(q)}_{t - j}} \\
        & \qquad {} \le p^{(q)}_{t - j} + (1 + 1/c) \cdot p^{(q)}_{t - j} + \cdots + (1 + 1/c)^{C-1} \cdot p^{(q)}_{t - j} \le (1 + 1/c)^{C} \cdot c \cdot p^{(q)}_{t - j}.
    \end{align*}
    This event is a necessary condition for $q \in T_j$.
    Taking a summation over all operations in $S_j$, we know that the expected size of $T_j$ is at most $(1 + 1/c)^C \cdot c \cdot \phi_{t - j}$. By \cref{lem:potential-upper-bound}, conditioned on $\HistExt(t - j - P) \supseteq \HistExt(t - j - \Theta(P))$, we have that $\phi_{t - j} \le O(\log P)$ with high probability in $P$.
    Moreover, the events $q \in T_j$ for different operations $q \in S_j$ are independent.
    By a Chernoff bound, we obtain
    \[
      \Pr\Bk[\big]{|T_j| > \Omega(\log P)} \le 1 / \poly(P) + 2^{-\Omega(\log P)} = 1 / \poly(P).
    \]
    That is, for a given $j > 0$, conditioned on $\HistExt(t - j - P) \supseteq \HistExt(t - j - \Theta(P))$, we have $|T_j| = O(\log P)$ with high probability in $P$. A union bound over $j \le P$ completes the proof.
  \end{proof}

  In the rest of the proof, we assume $|T_j| \le \Theta(\log P)$ holds for all $j \le P$.

  Next, we use $T$ to denote the length of the busy interval starting at $t$, and count the number of \StoreRandom instructions that are invoked in $[t, \, t + \log^2 P)$ and before the busy interval ends. The number of such \StoreRandom{}s is at least $\min\BK{T, \log^2 P}$.
  We partition the operations that invoke these \StoreRandom{}s into the following categories:
  \begin{enumerate}
    \item \BoldHeading{Operations that were invoked before time $t - P$.}
    Each operation either invokes a \StoreRandom instruction or quits before it invokes $O(\log P)$ shared instructions. According to \cref{prop:process-schedule-prob}, conditioned on $\HistExt(t - P)$, each operation $q$ that is ongoing at time $t - P$ (which implies $q$ was invoked before time $t - P$) has an independent probability of $2^{-\Omega(P)}$ to either invoke $\Omega(\log P)$ shared instructions or becomes not ready during time $[t - P, \; t)$. In both cases, $q$ has either invoked \StoreRandom or quit before time $t$, so it cannot invoke a \StoreRandom instruction after time $t$, and therefore cannot contribute to the busy interval. Since there are at most $P$ ongoing operations at time $t - P$, there is no operation in this category with high probability in $P$.

    In the remaining categories, we only need to consider operations that are invoked in $[t - P, \; t + P)$.

    \item \BoldHeading{Operations that transition to \instW within the first $C$ \instR{}s.}
    These operations do not appear in any $T_j$, but the probability of any operation $q$ being in this category is at most $\sum_{k = 0}^{C - 1} (1 + 1/c)^k / P^4 = O(1 / P^4)$, independently of other operations. By a Chernoff bound over all $O(P^2)$ operations invoked in $[t - P, \; t + \log^2 P)$, we know that the number of operations in this category is at most $O(\log P)$ with high probability in $P$.

    \item \BoldHeading{Operations that apply at least $C + 1$ \instR{}s after $t$ before they invoke \instW.}
    To bound the number of operations in this category, we need a key observation: during the busy interval, in every timestep, a \StoreRandom instruction gets applied, which changes the fingerprint $f$ of \SharedCell to a uniformly random $\eps \log P$-bit string that is independent of all earlier fingerprints.
    In order for any operation $q$ to apply at least $C + 1$ \instR{}s during the busy interval and before time $t + \log^2 P$ without quitting, there must be $C$ timesteps $t_1, t_2, \dots, t_C \in [t, \; t + \log^2 P)$ such that the \StoreRandom{}s applied in these timesteps write exactly the same fingerprint to \SharedCell. The probability of this event is bounded by
    \[
      \binom{\log^2 P}{C} \cdot (P^{-\eps})^{C - 1} \le (\log^2 P)^{C} \cdot (P^{-\eps})^{C - 1} \le P^{-\eps C / 2} = 1/\poly(P),
    \]
    where the last inequality holds because $C$ is a sufficiently large constant. Therefore, with high probability in $P$, there is no operation in this category.

    \item \BoldHeading{Operations that apply at most $C$ \instR{}s after $t$ before invoking \instW, but apply at least $C + 1$ \instR{}s in total.}
    Every operation $q$ in this category appears in a unique $T_j$, where $t - j$ is the timestep when $q$ applies the $(C + 1)$-th last \instR instruction. Fixing $j \le P$ and $q \in T_j$, we define $X_{q, j}$ as an indicator variable that $q$ invokes \instW or quits before time $t$. Since $X_{q, j}$ implies that $q$ cannot contribute a \StoreRandom to the busy interval, we know
    \[
      \widebar{X} \defeq \sum_{j = 1}^{P} \sum_{q \in T_j} (1 - X_{q, j})
    \]
    is an upper bound on the number of operations in this category.

    Recall the definition of \emph{scheduling coins} from \cref{prop:random-delay-for-operations}. For every $q \in T_j$, we define
    \[
      Y_{q, j} \defeq \ind\Bk[\big]{\CoinSum_{q, \, [t - j + 1, \, t)} < C + 1}.
    \]
    By \cref{prop:coin-implies-scheduled}, if $Y_{q, j} = 0$, then $q$ either invokes at least $C + 1$ shared instructions during $[t - j + 1, \, t)$, or becomes not ready during $[t - j + 1, \, t)$. Because $q \in T_j$, both cases imply that $q$ either invokes \instW or quits before time $t$, and therefore $X_{q, j} = 1$. This means that $X_{q, j} \ge 1 - Y_{q, j}$ for all $q, j$, so
    \[
      \widebar{X} \le \sum_{j = 1}^{P} \sum_{q \in T_j} Y_{q, j}.
      \numberthis \label{eq:widebar-X}
    \]

    We define a total order $\prec$ on the pairs $(q, j)$ where $(q, j) \prec (q', j')$ if and only if $j < j'$ or $j = j'$ and $q < q'$. Then, the following claim bounds the expected value of $Y_{q, j}$ conditioned on all other $Y$-variables before it in the order.

    \begin{claim}
      Fixing $j \le P$, conditioned on $\BK{Y_{q', j'}}_{(q', j') \prec (q, j)}$, we have
      \[
        \E\Bk[\Big]{Y_{q, j} \;\Big|\; \BK{Y_{q', j'}}_{(q', j') \prec (q, j)}, \, \HistExt(t - j + 1), \, \Randomness} \le 2^{-\Omega(j)}.
      \]
    \end{claim}
    \begin{proof}
      We fix an arbitrary realization of $\HistExt(t - j + 1)$ and \Randomness, which determines $\BK{T_{j'}}_{j' \le j}$ and therefore the set of $Y$-variables before $Y_{q, j}$ in the order. We also fix an arbitrary realization of the scheduling coins used to determine the previous $Y$-variables $\BK{Y_{q', j'}}_{(q', j') \prec (q, j)}$.
      Then, because every operation $q$ can only appear in one $T_j$, the scheduling coins used to determine $Y_{q, j}$ are all unrevealed. By \cref{prop:random-delay-for-operations}, these coins are independent fair coin flips conditioned on $\HistExt(t - j + 1)$, \Randomness, and all revealed scheduling coins of other operations $q' \ne q$.
      The claim follows from \cref{prop:process-schedule-prob}.
    \end{proof}

    The expected summation over all $Y$-variables is at most
    \[
      \sum_{j = 1}^P \sum_{q \in T_j} \E\Bk[\Big]{Y_{q, j} \;\Big|\; \BK{Y_{q', j'}}_{(q', j') \prec (q, j)}, \, \HistExt(t - j + 1), \, \Randomness} \le \sum_{j = 1}^P |T_j| \cdot 2^{-\Omega(j)} = O(\log P).
    \]
    By the \emph{multiplicative Azuma's inequality}~\cite[Corollary 6]{kuszmaul2025multiplicative}, conditioned on $\HistExt(t - P)$ and \Randomness, with high probability in $P$,
    \[
      \sum_{j = 1}^P \sum_{q \in T_j} Y_{q, j} \le O(\log P).
    \]
    This further implies that $\widebar{X} \le O(\log P)$ with high probability in $P$, i.e., the number of operations in this category is at most $O(\log P)$.
  \end{enumerate}
  Combining all categories together, we know that conditioned on $\HistExt(t - P)$ and \Randomness, with high probability in $P$, there are at most $O(\log P)$ operations that invoke \StoreRandom instructions in $[t, \, t + P)$ and before the busy interval ends, so the busy interval length is at most $O(\log P)$.
\end{proof}

\begin{lemma}
  \label{lem:busy-interval-length-large}
  Conditioned on an arbitrary $\HistExt(t)$, the busy interval containing $t$ (if exists) ends before time $t + O(P)$ with high probability in $P$.
\end{lemma}

\begin{proof}
  In order for the busy interval to extend beyond $t + O(P)$, at least $\Omega(P)$ \StoreRandom{}s need to be applied during $[t, \; t + O(P))$ and before the busy interval ends. These \StoreRandom{}s consist of two categories:
  \begin{itemize}
    \item \StoreRandom{}s from operations invoked before $t$. There are at most $P$ such \StoreRandom{}s because all these operations are ongoing at time $t$.
    \item \StoreRandom{}s from operations invoked after $t$. There are at most $O(P^2)$ operations invoked in $[t, \; t + O(P))$, which are possible to invoke such \StoreRandom{}s.

    \begin{claim}
      Conditioned on $\HistExt(t')$ for $t' \in [t, \; t + O(P))$, for any operation $q$ invoked in timestep $t'$, the probability that $q$ applies \instW before the busy interval ends is at most $O(1 / P^4)$.
    \end{claim}
    \begin{proof}
      By \cref{prop:process-schedule-prob}, conditioned on $\HistExt(t')$, operation $q$ either quits or invokes \instW before time $t' + O(\log P)$. We assume this condition holds for the rest of the proof.

      For every $j \ge 1$, we bound the probability that operation $q$ applies exactly $j$ \instR instructions before it transitions to \instW, all before the busy interval ends. In order for this to happen, there must be $j$ timesteps $t_1, t_2, \dots, t_j \in [t', \; t' + O(\log P))$ such that the \StoreRandom{}s applied in these timesteps write the same fingerprint to \SharedCell. The probability of this event is bounded by
      \[
        \min\BK*{\binom{\log^2 P}{j} \cdot (P^{-\eps})^{j - 1}, \; 1} \le \min\BK*{(\log^2 P)^{j} \cdot (P^{-\eps})^{j - 1}, \; 1} \le P^{-\eps (j - 1) / 2}.
      \]
      Assuming $q$ does not quit in the first $j$ \instR instructions, the probability that it transitions to \instW after the $j$-th \instR instruction is at most $(1 + 1/c)^j / P^4$. Taking a summation over $1 \le j \le O(\log P)$, we know that the probability that $q$ transitions to \instW before the busy interval ends is at most
      \[
        \sum_{j = 1}^{O(\log P)} \min\BK*{\frac{(1 + 1/c)^j}{P^4}, \; 1} \cdot P^{-\eps (j - 1) / 2} \le O(1 / P^4). \qedhere
      \]
    \end{proof}

    A union bound over all $O(P^2)$ operations invoked in $[t, \; t + O(P))$ shows that there are no \StoreRandom{}s in this category with probability at least $1 - O(1 / P^2)$.
  \end{itemize}
  In summary, with probability at least $1 - O(1 / P^2)$, there are at most $P$ \StoreRandom{}s getting applied during $[t, \; t + O(P))$ that can contribute to the busy interval. Therefore, conditioned on $\HistExt(t)$, the busy interval terminates before $t + O(P)$ with probability at least $1 - O(1 / P^2)$. Repeating this argument for a large constant times yields the lemma.
\end{proof}

\begin{corollary}
  \label{cor:busy-interval-endpoint}
  For every time $t \ge 0$, conditioned on $\HistExt(t - \Theta(P))$, the busy interval containing $t$ (if exists) ends before time $t + O(\log P)$ with high probability in $P$.
\end{corollary}

\begin{proof}
  Assume $t$ is in a busy interval (otherwise the statement holds trivially). Let $t_1 \le t$ be the starting point of the busy interval containing $t$.
  By \cref{lem:busy-interval-length-large}, conditioned on $\HistExt(t - \Theta(P))$, the busy interval containing $t - \Theta(P)$ (if exists) does not extend beyond time $t$, which implies that $t_1 \ge t - \Theta(P)$.
  By \cref{lem:busy-interval-length} at time $t_1$,%
  \footnote{Note that $t_1$ is a random variable that is bounded by a polynomial range with high probability in $P$, which allows us to apply \cref{lem:busy-interval-length} at time $t_1$ via an implicit union bound as discussed in \cref{sec:conventions}.}
  we know that the busy interval containing $t$ ends before $t_1 + O(\log P) \le t + O(\log P)$ with high probability in $P$.
\end{proof}

\subsubsection{Bounding the Latency}
\label{sec:register/latency}

Finally, we use \cref{cor:busy-interval-endpoint} to bound the latency of the algorithm.

\begin{lemma}
  \label{lem:register/latency}
  Let $t$ be a fixed timestep and let $q$ be an arbitrary \Write operation invoked in timestep $t$. Then, conditioned on $\HistExt(t - \Theta(P))$, operation $q$ completes within $O(\log P)$ time with high probability in $P$.
\end{lemma}

\begin{proof}
  After being invoked, operation $q$ performs at most $O(\log P)$ shared instructions before it either quits or invokes a \StoreRandom instruction. Let $\TThreshold = \Theta(\log P)$ be a sufficiently large multiple of $\log P$. By \cref{prop:process-schedule-prob}, conditioned on $\HistExt(t - \Theta(P))$, with high probability in $P$, operation $q$ is scheduled at least $\Omega(\log P)$ times before time $t + \TThreshold$, thus it either quits or invokes a \StoreRandom.
  \begin{itemize}
    \item If operation $q$ quits, its latency is at most $\TThreshold = O(\log P)$.
    \item If operation $q$ invokes a \StoreRandom, we denote by $t_1 < t + \TThreshold$ the timestep when $q$ invokes the \StoreRandom, and let $t_2 > t_1$ be the first time after $t_1$ when the instruction queue of \SharedCell is empty. By \cref{cor:busy-interval-endpoint}, conditioned on $\HistExt(t - \Theta(P))$, with high probability in $P$, the busy interval containing $t + \TThreshold \ge t_1 + 1$ ends before $t + \TThreshold + O(\log P) \le t + 2 \TThreshold$, so we have $t_2 \le t + 2 \TThreshold$. At time $t_2$, the instruction queue for \SharedCell is empty, so the \StoreRandom invoked by operation $q$ has already been applied. It implies that the operation $q$ completes before time $t + 2 \TThreshold$ with high probability in $P$, or equivalently, its latency is at most $2 \TThreshold = O(\log P)$.
  \end{itemize}
  Combining both cases completes the proof.
\end{proof}

\subsection{Putting the Pieces Together}

In this subsection, we put pieces together to prove \cref{thm:register-main}.

\begin{proofof}{\cref{thm:register-main}}
  Recall that our algorithm for register in \cref{sec:register-algorithm} maintains a single machine word in the shared memory that stores a pair $(x, f)$, where $x$ is a $\l$-bit integer that represents the current value of the register, and $f$ is a random $\eps \log P$-bit fingerprint. The \Read operation is directly implemented using a \Load instruction, and the \Write operation is implemented as in \cref{alg:shared-register-write}.

  We now verify the desired properties of our algorithm.
  \begin{itemize}
    \item \textbf{Linearizability and wait-freedom.} \cref{lem:register-linearizable} states that our algorithm is linearizable and wait-free.
    \item \textbf{Memory usage.} It is easy to see that our algorithm uses a single machine word in the shared memory with word size at least $\l + \eps \log P$ bits; every operation also uses at most $O(1)$ words of local memory.
    \item \textbf{High-probability latency.} A \Read operation always completes after invoking a \Load instruction, and thus completes within $O(\log P)$ time with high probability in $P$. \cref{lem:register/latency} states that, for every fixed time $t \ge 0$, with high probability in $P$, every write operation invoked in timestep $t$ completes within $O(\log P)$ time. \qedhere
  \end{itemize}
\end{proofof}

\paragraph{Why shorter random fingerprints fail.}

In \cref{alg:shared-register-write}, we use a random fingerprint of $\eps \log P$ bits to detect changes to the shared cell. We show that this fingerprint length is necessary.

\begin{theorem}
  \label{thm:register/short-fingerprint-fail}
  \cref{alg:shared-register-write} with random fingerprints of $F = o(\log P)$ bits fails to achieve a high-probability latency guarantee of $O(\log P)$.
\end{theorem}

\begin{proofsketch}
  We show a possible execution that occurs with probability at least $1 / P^{o(1)}$, in which some operation incurs $\omega(\log P)$ latency.
  In this execution, we assume a greedy scheduler, i.e., in every timestep, all ready operations are scheduled. The input in this execution is simple: the user invokes $P$ \Write operations in the first timestep, all trying to write the same value $0$ to the register, and invokes no operations afterwards.

  Let $\alpha$ be a superconstant parameter to be determined later. Then, the following sequence of events occurs with probability at least $1 / P^{o(1)}$:
  \begin{itemize}
    \item In the first $\Theta(\log P)$ timesteps, no \StoreRandom instruction is invoked, and the potential $\phi_t$ increases to $\Theta((\log P) / \alpha)$.
    This happens with probability at least $1 / 2^{O((\log P) / \alpha)} = 1 / P^{o(1)}$ by a direct calculation.
    \item $\Theta((\log P) / \alpha)$ operations invoke \StoreRandom simultaneously, which enter the instruction queue and apply one by one. This happens with constant probability by a Chernoff bound.
    \item The first $\alpha = \omega(1)$ \StoreRandom instructions that get applied write the same fingerprint $0$ to \SharedCell. This happens with probability $1 / 2^{F \alpha} = 1 / 2^{o(\log P)} = 1 / P^{o(1)}$, where the first equality holds by choosing $\alpha = \sqrt{\frac{\log P}{F}}$.
    \item In the first $\alpha$ timesteps, the value of \SharedCell is always $0$ and has not changed, so no operation can abort. Therefore, during the first $\alpha$ timesteps, the potential $\phi_t$ increases by a factor of $(1 + 1/c)^\Theta(\alpha) = 2^{\Theta(\alpha)}$, which reaches at least $\frac{\log P}{\alpha} \cdot 2^{\Theta(\alpha)} = \omega(\log P)$.
    \item $\omega(\log P)$ operations invoke \StoreRandom instructions during the first $\alpha$ timesteps, which all enter the instruction queue, causing some operations to have $\omega(\log P)$ latency.
  \end{itemize}
  This execution occurs with probability at least $1 / P^{o(1)}$ and causes some operation to have $\omega(\log P)$ latency, which completes the proof.
\end{proofsketch}

\section{CAS Register}
\label{sec:cas-register}

In this section, we study the \defn{CAS register} problem.

\begin{definition}
  A \defn{CAS register} is a data structure that maintains a single $\l$-bit integer $x$, and supports two (logical) operations:
  \begin{itemize}
    \item \Read: Return the current value of $x$.
    \item $\opCAS(\xExpect, \xNew)$: If $x = \xExpect$, overwrite $x$ with $\xNew$, return true; otherwise, do nothing and return false.
  \end{itemize}
\end{definition}

Similar to the read/write register problem, the most straightforward implementation of a CAS register is to use the provided atomic instructions \Load and \CAS, but this implementation can have time complexity $\Theta(P)$ on specific inputs. Our main result is the following theorem.

\begin{restatable}[Theorem~\ref{thm:cas-register-main-intro}, restated]{theorem}{CasRegisterMain}
  \label{thm:cas-register-main}
  We can construct a linearizable, lock-free $\l$-bit CAS register that works under the stochastic CRQW model with adaptive inputs, such that:
  \begin{itemize}
    \item Each operation has latency at most $O(\log P)$ with high probability in $P$.
    \item The data structure uses $O(1)$ machine words with word size $w \ge \max\BK{\l + 2 \log \log P, \; 2 \log P}$ bits in the shared memory, and $O(1)$ words of local memory for each operation.
  \end{itemize}
\end{restatable}

\subsection{Warmup: Basic CAS Register}
\label{sec:basic-cas-register}

As a warmup, we first present and analyze a CAS register algorithm that is similar to \cref{alg:shared-register-write}. This algorithm will later be used as a subroutine in our final algorithm for the CAS register problem.

\paragraph{Algorithm description.}
The algorithm maintains a single machine word in the shared memory, called \CellC, which stores a pair $(x, f)$. Here, $x$ represents the current value of the CAS register, and $f$ is a $2 \log \log P$-bit fingerprint. Every time one operation applies a \CAS instruction that changes $x$, it also increases $f$ by one, so that the change can be detected by other operations.

Similar to \cref{alg:shared-register-write}, the \Read operation is implemented directly by a \Load instruction; the \opCAS operation uses a \emph{back-on} strategy, which starts with a small probability $p$ of invoking a \CAS instruction, and the probability doubles in each iteration. When the operation detects that the \CellC changes its value, it aborts by linearizing itself just before or after the change point. See \cref{alg:basic-cas-register} for details,
where $c$ on line~\ref{line:basic-cas/defn-c} is a sufficiently large constant satisfying $c \ge 2^{4\tau+1}$.

There are three shared instructions in \cref{alg:basic-cas-register}: \instS, \instR, and \instC. Similar to \cref{sec:register}, throughout this section, we adopt the view of the \emph{reduced state machine}, and say an operation is \emph{scheduled} if it applies a shared instruction and transitions to the next state.

Recalling that \cref{alg:shared-register-write} lets every operation complete within $O(\log P)$ time with high probability in $P$, one might think \cref{alg:basic-cas-register} gives the same guarantee. However, the same guarantee is true only when the length of the execution is bounded by a polynomial in $P$.

\begin{restatable}{lemma}{BasicCasLatencyGuarantee}
  \label{lem:basic-cas/latency-guarantee}
  In an execution of \cref{alg:basic-cas-register} with at most $P^{O(1)}$ timesteps, every operation completes within $O(\log P)$ time with high probability in $P$.
\end{restatable}

In the remainder of this subsection, we show that \cref{alg:basic-cas-register} is linearizable and wait-free, and prove \cref{lem:basic-cas/latency-guarantee}.

\subsubsection{Linearizability and Wait-Freedom}
\label{sec:basic-cas/linearizability}

\begin{lemma}
  \label{lem:basic-cas/linearizable}
  \cref{alg:basic-cas-register} is linearizable and wait-free.
\end{lemma}

\begin{proof}
  Wait-freedom is easy to see: each iteration through the loop in lines~\ref{line:basic-cas/9}--\ref{line:basic-cas/16} doubles the probability $p$, so $p$ reaches 1 after $O(\log P)$ iterations.
  At this point, the
  \opCAS operation invokes a hardware \CAS instruction in line~\ref{line:basic-cas/14} and returns after it is applied, if the operation has not returned earlier.

  We now turn to linearizability.
  Recall that \CellC maintains a pair $(x, f)$. Throughout the linearizability proof, we refer to the first entry $x$ as the \defn{value} of the register, and to the second entry $f$ as the \defn{fingerprint}. \Read operations are linearized at their \Load instructions.
  For an operation $q=\opCAS(\xExp, \xNew)$, let $(\xOld, \fOld)$ be the pair returned
  by the \instS instruction of $q$, and we linearize $q$ as follows.
  \begin{itemize}
    \item \emph{Immediate returns.}
    If $q$ returns at line~\ref{line:basic-cas/5} ($\xOld \neq \xExp$) or
    at line~\ref{line:basic-cas/7} ($\xNew = \xOld$),
    we linearize $q$ when it applies \instS.
    \item \emph{Return on value mismatch at \instR.}
    If $q$ returns at line~\ref{line:basic-cas/11} upon reading in line~\ref{line:basic-cas/9} that $\xNow \neq \xOld$,
    we linearize $q$ when it applies \instR.
    \item \emph{Successful hardware \CAS.}
    If $q$ applies a successful \instC in line~\ref{line:basic-cas/14}, we linearize $q$ when it applies \instC.
    \item \emph{Failed hardware \CAS due to value mismatch.}
    If $q$ applies \instC in line~\ref{line:basic-cas/14} and it fails \emph{because the current value of the register was not $\xOld$} (regardless of the fingerprint),
    we linearize $q$ when it applies \instC.
    \item \emph{Fingerprint-only mismatches.}
    There are two cases where $q$ observes $\xNow = \xOld$ but $\fNow \neq \fOld$:
    (i) $q$ returns at line~\ref{line:basic-cas/11} after reading in line~\ref{line:basic-cas/9} that $\xNow = \xOld$ and $\fNow \neq \fOld$; or
    (ii) $q$ returns at line~\ref{line:basic-cas/15} after failing \instC in line~\ref{line:basic-cas/14} \emph{because the current value was $\xOld$ but the current fingerprint was not $\fOld$}.
    In either situation, let $t^\star$ be the time of the \emph{last} successful \instC (by any operation) that occurs \emph{before} the \instR/\instC where the fingerprint mismatch is observed. We linearize $q$ \emph{immediately before} $t^\star$.
  \end{itemize}
  A pending \opCAS operation that has successfully applied a hardware \CAS is linearized when the \CAS is applied; other pending operations are not linearized.

  We claim that this linearization respects the semantics of a \CAS
  register w.r.t.\ the \emph{value} of the hardware \CAS register
  (ignoring the fingerprint).
  This is shown by induction on the number of successful hardware
  \CAS instructions that are linearized before the operation we are interested in.
  For simplicity, it is convenient to model the initialization of the register as a successful $\CAS$ applied right before time $0$,
  which sets the initial value of $\CellC$. This provides the base case of the induction (with zero operations linearized before time $0$).

  Now suppose the claim holds up to some timestep $t \geq 0$
  when a successful hardware \CAS instruction changes the contents of the register to
  $(x, f)$,
  and consider the next timestep $t' > t$
  when a successful hardware \CAS instruction changes the contents of the register
  to $(x', f')$.
  Note that $x' \neq x$, because
  any operation that reaches line~\ref{line:basic-cas/14} must have $\xOld = \xExp$
  and $\xNew \neq \xExp$ (otherwise it would have returned in line~\ref{line:basic-cas/5} or line~\ref{line:basic-cas/7}).

  The operation that applied the successful \CAS at time $t'$
  returns true, and
  is linearized at that point. It must have $\xOld = x$ and $\xNew = x'$,\
  so this respects the semantics of a successful \CAS.
  Other operations linearized after time $t$ and before time $t'$
  can include:
  \begin{itemize}
    \item Operations that linearize at a $\Load$ instruction
    applied in timesteps $[t + 1, \, t']$,
    and return false: each such operation reads the value $x$
    at that $\Load$,
    and then has either $\xExp \neq x$ in line~\ref{line:basic-cas/4},
    or $\xOld = \xExp$ in line~\ref{line:basic-cas/4} but $x = \xNow \neq \xOld$ in line~\ref{line:basic-cas/10},
    which again implies that $\xExp \neq x$.
    In both cases, the value observed at that $\Load$ differs from $\xExp$, so linearizing the operation at this point respects the semantics of a failed \opCAS.

    \item Operations that linearize at a $\Load$ instruction
    applied in timesteps $[t + 1, \; t']$
    and return true:
    Each such operation loads the value $x$ into $\xOld$
    and have $\xOld = \xExp$ in line~\ref{line:basic-cas/4}
    and also $\xNew = \xOld$ in line~\ref{line:basic-cas/6}.
    This means that when the $\Load$ instruction was applied, the value of the register was $x = \xExp = \xNew$,
    so the linearization respects the semantics of a successful \opCAS
    operation that did not change the value of the register (because $\xNew = \xExp$).
    \item Operations that linearize at a failed hardware $\CAS$ instruction in line~\ref{line:basic-cas/14}, and return false:
    By definition of the linearization, these operations fail because the value of \CellC differs from $\xOld$ at the time their $\CAS$ is applied,
    so this respects the semantics of a failed \opCAS.

    \item Operations that linearize immediately before the successful hardware
    \CAS at time $t'$:
    Each such operation
    observes a fingerprint mismatch but not a value mismatch \emph{after} time $t'$ at either line~\ref{line:basic-cas/9} or line~\ref{line:basic-cas/14}. Let $(\xNow, \fNow)$ be the content of \CellC that the operation observes at that time.
    Then, it either has $x' = \xNow = \xOld$ but $f' = \fNow \neq \fOld$ in line~\ref{line:basic-cas/10}, or it fails the hardware \CAS in line~\ref{line:basic-cas/14} with $x' = \xOld$ but $\fOld \neq f'$. In both cases, we have that $x' = \xOld$, and that the operation previously executed line~\ref{line:basic-cas/4} and observed that $\xExp = \xOld$, which together implies $\xExp = x'$. Therefore, at the linearization point (immediately before timestep $t'$), the value stored in \CellC is $x \neq x' = \xExp$, so the linearization respects the semantics of a failed \opCAS. \qedhere
  \end{itemize}
\end{proof}

\subsubsection{Analyzing \CAS Instruction Behavior}

We start the proof of \cref{lem:basic-cas/latency-guarantee} by defining terminologies for the behavior of \CAS instructions.

\begin{definition}[Status of \CAS Instructions]
  \label{defn:basic-cas/instruction-status}
  \phantom{a}
  \begin{itemize}
    \item We say an operation \defn{proposes} a \CAS instruction in timestep $t$ if the operation applies \instR in timestep $t$ and transitions to \instC. The arguments of this \CAS instruction are already determined by the operation.
    \item A \CAS instruction that has been proposed but not applied is \defn{awaiting}. Equivalently, if a \CAS instruction is awaiting, either its corresponding operation is pending on \instC, or the instruction has been invoked but not yet applied (i.e., the instruction is waiting in the queue).
    \item An awaiting \CAS instruction is \defn{valid} if its source value $(\xOld, \fOld)$ matches the current value of \CellC. A valid \CAS instruction is guaranteed to succeed when it gets applied. The execution of a valid \CAS instruction changes the cell value and \defn{invalidates} all awaiting \CAS instructions.%
    \footnote{There is a small probability that an invalid awaiting \CAS instruction can become valid again, as our $2 \log \log P$-bit fingerprint can wrap around after being incremented for $\log^2 P$ times. We will show that this probability is negligible.}
  \end{itemize}
\end{definition}

\begin{observation}
  \label{obs:basic-cas/invoke->change}
  Suppose there is no invalid awaiting \CAS instruction at time $t$,
  and $t_1 \ge t$ is the first timestep after $t$ in which a \CAS instruction is invoked. Then, the value of \CellC changes in timestep $t_1$.
\end{observation}
\begin{proof}
  The invoked \CAS instruction is enqueued in timestep $t_1$, which implies that at least one \CAS instruction gets applied in timestep $t_1$.
  This \CAS instruction must be valid when it gets applied, because (1) every \CAS instruction proposed in $[t, \, t_1)$ is valid before time $t_1$ (as the cell value changes only after $t_1$), and (2) every \CAS instruction proposed before $t$ is valid at time $t$ (as stated in the condition). Hence, applying this \CAS instruction changes the cell value.
\end{proof}

\begin{definition}[Potential]
  \label{defn:basic-cas/potential}
  \phantom{a}
  \begin{itemize}
    \item An operation $q$ is \defn{active} at time $t$ if it is pending on \instR, and \CellC has not changed its value since $q$ applied \instS. Denote by $A_t$ the set of active operations at time $t$.
    \item The local variable $p$ in \cref{alg:basic-cas-register} is called the \defn{invocation probability} of the operation. Denote by $p^{(q)}_t$ the invocation probability of operation $q$ at time $t$.
    \item The \defn{potential} $\phi_t$ of the data structure is the total invocation probability of all active operations at time $t$, i.e.,
    \[
      \phi_t \defeq \sum_{q \in A_t} p^{(q)}_t.
    \]
  \end{itemize}
\end{definition}

\cref{defn:basic-cas/potential} is similar to the definition of potential in \cref{sec:register-analysis}, except that it additionally requires that \CellC has not changed its value since the operation applied \instS in order for the operation to be active. We have the following claim as an analog of \cref{clm:potential-growth-bound}.

\begin{claim}
  \label{clm:basic-cas/potential-growth-bound}
  For any integer $k \ge 1$, $\phi_{t + k} \le (\phi_t + 1 / P^{2c^3 - 1}) \cdot 2^k$.
\end{claim}

\begin{proof}
  In each timestep, the invocation probability of any active operation can at most double, and new active operations have a total invocation probability of at most $P \cdot (1 / P^{2c^3}) = 1 / P^{2c^3 - 1}$. Therefore,
  \[
    \phi_{t + 1} \le \phi_t \cdot 2 + 1 / P^{2c^3 - 1}.
  \]
  Applying this repeatedly for $k$ times yields the claim.
\end{proof}

\begin{lemma}
  \label{lem:basic-cas/cell-change}
  Suppose there is no invalid awaiting \CAS instruction at time $t$. Conditioned on $\HistExt(t)$, with probability at least $1 - 2^{-\Omega(\phi_t)}$, the value of \CellC is changed during time interval $[t, \; t + 4\tau)$.
\end{lemma}

\begin{proof}
  Assume there are $k$ active operations $q_1, \dots, q_k$ at time $t$. For each $q_i$, we define
  \[
    X_i \defeq \ind\Bk[\big]{\CoinSum_{q_i, \, [t, \; t + 4\tau)} \ge 2}, \qquad
    Y_i \defeq \ind\Bk[\big]{q_i \text{ transitions to \instC after the next \instR}}.
  \]
  Here, $X_i = 1$ implies that $q_i$ is scheduled at least twice during $[t, \; t + 4 \tau)$. It is easy to see that all variables $X_i, Y_i$ are independent%
  \footnote{To see this, note that $Y_i$s are determined solely by the algorithm's randomness \Randomness, and are mutually independent; $X_i$s only depend on the scheduling coins. The independence of $X_i$s and $Y_i$s conditioned on $\HistExt(t)$ then follow from \cref{prop:random-delay-for-operations}.}
  conditioned on $\HistExt(t)$, and that
  \[
    \Pr[X_i Y_i = 1] = \Pr[X_i = 1] \cdot \Pr[Y_i = 1] \ge \frac{1}{4} \cdot p^{(q_i)}_t.
  \]
  It further implies that $\E\Bk[\big]{\sum_{i=1}^k X_i Y_i} \ge \phi_t / 4$. By a Chernoff bound, we know that with probability at least $1 - 2^{-\Omega(\phi_t)}$, there is $\sum_{i=1}^k X_i Y_i > 0$, which means that at least one active operation $q_i$ transitions to \instC and then invokes a \CAS instruction during $[t, \; t + 4\tau)$. This combined with \cref{obs:basic-cas/invoke->change} implies that the cell value is changed before $t + 4\tau$.
\end{proof}

\begin{definition}
  We say an operation $q$ is \defn{heavily-delayed} at time $t$ if it applied \instS before time $t - \tfrac{1}{2} \log^2 P$ and is still pending on \instR at time $t$. (Note that $q$ can possibly be inactive at time $t$.)
\end{definition}

Note that each operation is no longer pending on \instR after being scheduled for $\Theta(\log P)$ times, so intuitively, heavily-delayed operations are very unlikely to occur. In the following few lemmas, we will take as assumption that there is no heavily-delayed operation at a particular time $t$.

The key to our analysis is the following lemma.

\begin{lemma}
  \label{lem:basic-cas/no-many-awaiting-cas}
  Suppose at time $t$,
  \begin{enumerate}[label=\textup{(\roman*)}]
    \item $0 < \phi_t \le c \log P$;
    \item there is no awaiting \CAS instruction;
    \item\label{item:basic-cas/no-slow-operation} there is no \emph{heavily-delayed} operation at time $t$.
  \end{enumerate}
  Let $t_1$ be the first timestep after $t$ when the value of \CellC changes. Then, conditioned on $\HistExt(t)$, with high probability in $P$, we have:
  \begin{enumerate}[label=\textup{(\alph*)}]
    \item $t_1 < t + c^4 \log P$;
    \item\label{item:basic-cas/no-many-awaiting-cas} the number of awaiting \CAS instructions at time $t_1 + 1$ is at most $c^3 \log P$.
  \end{enumerate}
\end{lemma}

\begin{proof}
  We first show that $t_1 \le t + c^4 \log P$ with high probability in $P$. Since $\phi_t > 0$, there exists an active operation $q \in A_t$, and it will invoke a \CAS instruction after getting scheduled for $\log P^{2c^3} + 2$ times, unless the cell value changes before then. By \cref{prop:process-schedule-prob}, we know that with high probability in $P$, this operation $q$ invokes a \CAS instruction before time
  \[
    t + (\log P^{2c^3} + 2) \cdot (4 \tau) < t + 10 c^3 \tau \log P < t + c^4 \log P,
  \]
  unless the cell value changes before then.
  By \cref{obs:basic-cas/invoke->change}, we know that the cell value changes before time $t + c^4 \log P$ with high probability in $P$.

  Then, we point out that \emph{inactive} operations will not propose \CAS instructions before time $t + c^4 \log P$. This is because the fingerprint $f$ of \CellC can only increase by one in each timestep. If an operation $q$ pending on \instR becomes invalid due to an increase of the fingerprint $f$, it needs to wait for at least $\log^2 P$ timesteps for the fingerprint $f$ to wrap around, at which point it is possible again for $q$ to propose a \CAS instruction; otherwise, $q$ will abort on its next \instR. Therefore, condition~\ref{item:basic-cas/no-slow-operation} ensures that no inactive operation can propose a \CAS instruction before time $t + \tfrac{1}{2} \log^2 P > t + c^4 \log P$.

  In the remainder of the proof, we enumerate all possible values of $t_1 \in [t, \; t + c^4 \log P)$, and bound the probability that condition~\ref{item:basic-cas/no-many-awaiting-cas} is violated with $t_1$ being the first timestep when the value of \CellC changes. There are two cases:

  \paragraph{Case 1: There exists a time $t' \in [t, \; t_1 - 4 \tau]$ with $\phi_{t'} \ge c \log P$.}

  By \cref{lem:basic-cas/cell-change}, conditioned on $\HistExt(t')$, the value of \CellC changes during $[t', \; t' + 4 \tau) \subseteq [t, \, t_1)$ with probability $1 - 2^{-\Omega(\phi_{t'})} \ge 1 - 2^{-\Omega(c \log P)} \ge 1 - 1 / \poly(P)$. Thus, the event that \CellC did not change before $t_1$ is witnessed by a low-probability event (i.e., conditioned on $\HistExt(t')$, \CellC did not change during $[t', \; t' + 4 \tau)$).

  \paragraph{Case 2: For all $t' \in [t, \; t_1 - 4 \tau]$, we have $\phi_{t'} < c \log P$.}

  We first analyze $\max_{t'' \in [t, \, t_1]} \phi_{t''}$. There are two subcases:
  \begin{itemize}
    \item $t_1 < t + 4 \tau$, so the range for $t'$ is empty.
    \item $t_1 \ge t + 4 \tau$, and $\phi_{t_1 - 4 \tau} < c \log P$.
  \end{itemize}
  Define $t_{\text{mid}} \defeq \max\BK{t_1 - 4 \tau, \; t}$. In both subcases, we have $\phi_{t_{\text{mid}}} < c \log P$. This combined with \cref{clm:basic-cas/potential-growth-bound} further implies that for every $t'' \in [t_{\text{mid}}, \, t_1]$, there is
  \[
    \phi_{t''} \le \bk*{\phi_{t_{\text{mid}}} + 1 / P^{2c^3 - 1}} \cdot 2^{t'' - t_{\text{mid}}} \le 2c \log P \cdot 2^{4 \tau} \le c^2 \log P,
  \]
  where the last inequality holds because we assume $c \ge 2^{4 \tau + 1}$. Combining this with the assumption of Case 2, we know that $\phi_{t'} < c^2 \log P$ for all $t' \in [t, \, t_1]$.

  For every $t' \in [t, \, t_1]$, we define $S_{t'}$ to be the set of operations that applied \instR in timestep $t'$; define $T_{t'} \subseteq S_{t'}$ to be the set of operations that applied \instR in timestep $t'$ and then immediately transitioned to \instC (i.e., they proposed \CAS instructions in timestep $t'$). Every operation $q \in S_{t'}$ has probability $p^{(q)}_{t'}$ to transition to \instC after applying \instR, i.e., it has $p^{(q)}_{t'}$ probability to appear in $T_{t'}$. The expected size of $T_{t'}$ is thus at most $\sum_{q \in S_{t'}} p^{(q)}_{t'} \le \phi_{t'} \le c^2 \log P$. Moreover, for a fixed $t'$, the events $q \in T_{t'}$ for different operations $q \in S_{t'}$ are independent. By a Chernoff bound, with high probability in $P$, we have
  \[
    |T_{t'}| \le 2 c^2 \log P, \qquad \forall t' \in [t, \, t_1].
    \numberthis \label{eq:basic-cas/Tt-size-bound}
  \]

  Denote by $d_{t'}$ the number of awaiting \CAS instructions at time $t' \in [t, \; t_1 + 1)$.
  Assuming \eqref{eq:basic-cas/Tt-size-bound} holds, during time interval $[t, \; t_1 + 1)$, at most $2 c^2 \log P$ new \CAS instructions can be proposed in each timestep, i.e.,
  \[
    d_{t' + 1} \le d_{t'} + 2 c^2 \log P.
    \numberthis \label{eq:basic-cas/d-growth-bound}
  \]
  We define $t_2 \in [t, \; t_1 + 1]$ to be the first time with $d_{t_2} \ge c^2 \log P$. (In order for \ref{item:basic-cas/no-many-awaiting-cas} to be violated, $t_2$ must exist, since otherwise $d_{t_1 + 1} \le c^2 \log P$.) By \eqref{eq:basic-cas/d-growth-bound}, we know that $c^2 \log P \le d_{t_2} \le 3 c^2 \log P$, so in order for \ref{item:basic-cas/no-many-awaiting-cas} to be violated (i.e., $d_{t_1 + 1} > c^3 \log P$), we need
  \[
    (t_1 + 1) - t_2 > \frac{c^3 \log P - 3 c^2 \log P}{2 c^2 \log P} \ge c/4 > 4\tau + 2.
  \]
  By \cref{prop:random-delay-for-operations}, each of the $d_{t_2}$ awaiting \CAS instructions at time $t_2$ has an independent probability of at least $1/2$ to be applied before $t_2 + 4\tau$, so the probability that none of them gets applied before $t_2 + 4\tau$ is at most $2^{-d_{t_2}} \le 2^{-c^2 \log P} \le 1 / \poly(P)$. When any of these \CAS instructions gets applied, by \cref{obs:basic-cas/invoke->change}, the cell value changes before $t_2 + 4\tau < t_1$, which contradicts with the definition of $t_1$.

  In summary, the event that \ref{item:basic-cas/no-many-awaiting-cas} gets violated is witnessed by either of the following low-probability events: (1) \cref{eq:basic-cas/Tt-size-bound} does not hold, or (2) there exists $t_2 \in [t, \; t_1 - 4\tau]$ with $d_{t_2} \ge c^2 \log P$, such that starting from time $t_2$, none of the $d_{t_2}$ awaiting \CAS instructions gets applied before time $t_2 + 4\tau$. Both events occur with low probability in $P$.

  \paragraph{Combining both cases.}

  Combining both cases together, we know that the event that \ref{item:basic-cas/no-many-awaiting-cas} gets violated is witnessed by $t_1 \in [t, \; t + c^4 \log P)$ and a low-probability event depending on $t_1$. Taking a union bound over all possible values of $t_1$ yields the lemma.
\end{proof}

\subsubsection{Defining and Analyzing Healthy States}

In the next part of the proof, we will define the notion of \emph{healthy states}, and then show that starting from any time $t$ when the data structure is in a healthy state, it will enter a healthy state again within $\Theta(\log P)$ timesteps with high probability in $P$. This property will later be used to derive the latency guarantee in \cref{sec:basic-cas/latency-guarantee}.

\begin{definition}[Healthy State]
  \label{defn:basic-cas/healthy-state}
  We say the data structure is in a \defn{healthy state} at time $t$ if (1) there is no awaiting \CAS instruction, and (2) the potential $\phi_t$ is at most $1 / P^{c^2}$.
\end{definition}

\begin{claim}
  \label{clm:basic-cas/healthy-state-keep-silent}
  Suppose the data structure is in a healthy state at time $t$, and there is no heavily-delayed operation at time $t$.
  Then, conditioned on $\HistExt(t)$, with high probability in $P$, no \CAS instruction is invoked or applied during $[t, \, t + \tfrac{c^2}{2} \log P)$.
\end{claim}

\begin{proof}
  Let $D \defeq \tfrac{c^2}{2} \log P$.
  As argued in the proof of \cref{lem:basic-cas/no-many-awaiting-cas}, since $D < \tfrac{1}{2} \log^2 P$ for sufficiently large $P$, no inactive operation can propose a \CAS instruction during $[t, \, t + D)$. It follows that every proposed \CAS instruction during $[t, \, t + D)$ must come from an active operation. For each timestep $t' \in [t, \, t + D)$, the probability that some active operation proposes a \CAS instruction in timestep $t'$ is at most $\phi_{t'}$. By \cref{clm:basic-cas/potential-growth-bound} and the healthy-state bound $\phi_t \le 1/P^{c^2}$,
  \[
    \sum_{t' = t}^{t + D - 1} \phi_{t'}
    \le
    \sum_{j = 0}^{D - 1}
    \bk*{\frac{1}{P^{c^2}} + \frac{1}{P^{2c^3 - 1}}} \cdot 2^j
    \le
    \frac{2^{D}}{P^{c^2}}
    +
    \frac{2^{D}}{P^{2c^3 - 1}}
    =
    P^{-\Omega(1)}.
  \]
  By a union bound, with high probability in $P$, no active operation proposes a \CAS instruction during $[t, \, t + D)$.

  Since the data structure is in a healthy state at time $t$, there is no awaiting \CAS instruction at time $t$. Hence, if no operation proposes a \CAS instruction during $[t, \, t + D)$, then no \CAS instruction is invoked or applied during this interval. This completes the proof.
\end{proof}

\begin{lemma}
  \label{lem:basic-cas/enter-healthy-state}
  Suppose at time $t$,
  \begin{enumerate}[label=\textup{(\arabic*)}]
    \item $\phi_t \le c \log P$;
    \item there is no awaiting \CAS instruction;
    \item\label{item:basic-cas/no-slow-operation-2} there is no \emph{heavily-delayed} operation at time $t$.
  \end{enumerate}
  Then, conditioned on $\HistExt(t)$, with high probability in $P$, the data structure enters a healthy state at some time $t_2 \in [t + 1, \; t + 2c^4 \log P]$.
\end{lemma}

\begin{proof}
  We first prove the lemma assuming $\phi_t > 0$, i.e., there is at least one active operation at time $t$. We will handle the case of $\phi_t = 0$ at the end of the proof.

  Let $t_1$ be the first timestep after $t$ where the value of \CellC changes. $t_1$ is guaranteed to exist because there is at least one active operation at time $t$. By \cref{lem:basic-cas/no-many-awaiting-cas}, conditioned on $\HistExt(t)$, with high probability in $P$, we have $t_1 < t + c^4 \log P$, and the number of awaiting \CAS instructions at time $t_1 + 1$ is at most $c^3 \log P$.

  We argue that at time $t_2 \defeq t_1 + (c^3 + c) \log P \le t + 2c^4 \log P$, the data structure is in a healthy state with high probability in $P$, by verifying the required conditions:
  \begin{itemize}
    \item \textbf{Potential bound.}
    Since the value of \CellC changed in timestep $t_1$, we know that $\phi_{t_1 + 1} \le 1 / P^{2c^3 - 1}$ (because only operations invoked in timestep $t_1$ can be active at time $t_1 + 1$). By \cref{clm:basic-cas/potential-growth-bound}, we have that for any $t' \in [t_1 + 1, \; t_2]$,
    \[
      \phi_{t'} \le \bk*{\phi_{t_1 + 1} + \frac{1}{P^{2c^3 - 1}}} \cdot 2^{t' - t_1 - 1} \le 2 \cdot \frac{1}{P^{2c^3 - 1}} \cdot 2^{(c^3 + c) \log P - 1} = \frac{1}{P^{c^3 - c - 1}} < \frac{1}{P^{c^2}}.
      \numberthis \label{eq:basic-cas/phi-t2-bound}
    \]
    Thus, $\phi_{t_2} < 1 / P^{c^2}$, as required.

    \item \textbf{No awaiting \CAS instruction.}
    Similar to the proof of \cref{lem:basic-cas/no-many-awaiting-cas}, condition~\ref{item:basic-cas/no-slow-operation-2} ensures that no inactive operation can propose a \CAS instruction before time $t_2$. For active operations, the probability that any active operation proposes a \CAS instruction in timestep $t' \in [t_1 + 1, \; t_2)$ is bounded by the potential $\phi_{t'} \le 1 / P^{c^2}$ according to \eqref{eq:basic-cas/phi-t2-bound}. A union bound over all $t' \in [t_1 + 1, \; t_2)$ yields that, with high probability in $P$, no active operation proposes a \CAS instruction during $[t_1 + 1, \; t_2)$. Combining both cases yields that, with high probability in $P$, no operation proposes \CAS instructions during $[t_1 + 1, \; t_2)$.

    Let $S$ be the set of awaiting \CAS instructions at time $t_1 + 1$. We know from \cref{lem:basic-cas/no-many-awaiting-cas} that $|S| \le c^3 \log P$ with high probability in $P$. By a Chernoff bound and \cref{prop:random-delay-for-operations}, we know that every instruction in $S$ that has not been invoked yet (i.e., its corresponding operation is pending on \instC) has probability at least $1 - 1 / \poly(P)$ to be invoked before time $t_1 + c \log P$. A union bound over all instructions in $S$ yields that, with high probability in $P$, no \CAS instruction is waiting to be invoked at time $t_1 + c \log P$, and the number of \CAS instructions waiting in the queue is at most $c^3 \log P$ at time $t_1 + c \log P$. This combined with the property that no operation proposes \CAS instructions during $[t_1 + 1, \; t_2)$ yields that, with high probability in $P$, all awaiting \CAS instructions at time $t_1 + 1$ are already applied before time $t_2$, so there is no awaiting \CAS instruction at time $t_2$, as required.
  \end{itemize}
  Therefore, the data structure enters a healthy state at time $t_2$ with high probability in $P$, which concludes the proof when $\phi_t > 0$.

  \paragraph{Case of $\phi_t = 0$.} When $\phi_t = 0$, the data structure is already in a healthy state at time $t$. By \cref{clm:basic-cas/healthy-state-keep-silent}, with high probability in $P$, no \CAS instruction is invoked or applied during $[t, \; t + 1)$. Therefore, there is still no awaiting \CAS instruction at time $t + 1$, and $\phi_{t + 1} \le 1 / P^{2c^3 - 1} < 1 / P^{c^2}$. Hence, the data structure is in a healthy state at time $t + 1$.
\end{proof}

In \cref{lem:basic-cas/no-many-awaiting-cas,lem:basic-cas/enter-healthy-state}, in order to prevent inactive operations from proposing \CAS instructions, we set an additional requirement that there is no heavily-delayed operation at time $t$. The following claim guarantees that this requirement is met with high probability in $P$ for any given time $t$.

\begin{claim}
  \label{clm:basic-cas/no-slow-operation}
  For every time $t \ge 0$, conditioned on $\HistExt(t - \frac{1}{2} \log^2 P)$, with high probability in $P$, there is no heavily-delayed operation at time $t$.
\end{claim}

\begin{proof}
  After applying \instS, every operation either proposes a \CAS instruction or aborts after being scheduled for $2c^3 \log P + 3 = \Theta(\log P)$ times. For any operation $q$ that applies \instS before time $t - \tfrac{1}{2} \log^2 P$, the probability that $q$ is scheduled for less than $\Theta(\log P)$ times during $\bigl[t - \tfrac{1}{2} \log^2 P, \; t\bigr)$---which is necessary for $q$ to remain pending on \instR at time $t$---is at most $2^{-\Omega\bk{\frac{1}{2} \log^2 P}} = P^{-\omega(1)}$. Taking a union bound over at most $P$ ongoing operations at time $t - \tfrac{1}{2} \log^2 P$ yields that, the probability of any such operation remaining pending on \instR at time $t$ is at most $P^{-\omega(1)}$.
\end{proof}

\subsubsection{Proving Latency Guarantee}
\label{sec:basic-cas/latency-guarantee}

Finally, we proceed to prove \cref{lem:basic-cas/latency-guarantee}.

\BasicCasLatencyGuarantee*

\begin{proofof}{\cref{lem:basic-cas/latency-guarantee}}
  For every time $t \le P^{O(1)}$, \cref{clm:basic-cas/no-slow-operation} states that, with high probability in $P$, there is no ongoing operation at time $t$ that applied its \instS more than $\tfrac{1}{2} \log^2 P$ timesteps ago. By a union bound, we know that this property holds for all time $t$ in the execution with high probability in $P$. In the rest of the proof, we assume this property holds.

  According to \cref{alg:basic-cas-register}, every operation $q$ will be scheduled for at most $O(\log P)$ times before it either quits or invokes a \CAS instruction, which takes at most $O(\log P)$ time with high probability in $P$ according to \cref{prop:process-schedule-prob}. Therefore, if $q$ quits without invoking a \CAS instruction, its latency is at most $O(\log P)$ with high probability in $P$. If $q$ invokes a \CAS instruction, it will wait in the queue until the \CAS gets applied. Therefore, with high probability in $P$,
  \[
    \Latency(q) \le O(\log P) + (\text{time that the } \CAS \text{ instruction waits in the queue}).
    \numberthis \label{eq:basic-cas/latency-single-op}
  \]

  At time $t^*_0 \defeq 0$, there is no ongoing operation, so the data structure is in a healthy state. Applying \cref{lem:basic-cas/enter-healthy-state} on time $t^*_0$ yields a time $t^*_1 \in [t^*_0 + 1, \; t^*_0 + 2c^4 \log P]$ with high probability in $P$ when the data structure is again in a healthy state. We repeat this process to get a sequence of times $t^*_0 < t^*_1 < \cdots < t^*_k$ such that
  \begin{itemize}
    \item the data structure is in a healthy state at every $t^*_i$;
    \item $t^*_{i + 1} - t^*_i \le 2c^4 \log P = \Theta(\log P)$.
    \item $t^*_k$ is the end of the execution.
  \end{itemize}
  For any \CAS instruction invoked in timestep $t \in [t^*_i, \; t^*_{i + 1})$, we know that the \CAS instruction is completed before $t^*_{i + 1}$, as there is no awaiting \CAS instruction at time $t^*_{i + 1}$. This implies that the \CAS instruction waits in the queue for at most $t^*_{i + 1} - t \le \Theta(\log P)$ time, and this bound holds for all \CAS instructions invoked in the execution simultaneously. Combining this with \eqref{eq:basic-cas/latency-single-op} yields that the latency of each operation is at most $O(\log P)$ with high probability in $P$, which further implies the lemma via a union bound over all operations in the execution.
\end{proofof}

\subsubsection{Failure on Long Execution}
\label{sec:basic-cas/failure-on-long-execution}

In the proof of \cref{lem:basic-cas/latency-guarantee}, we showed that when the data structure is in a healthy state at time $t^*_i$, it will enter the next healthy state within $O(\log P)$ time with high probability in $P$. However, there is a small probability that the data structure enters an ``ill state'' and never recovers to healthy states again. We illustrate this by constructing a superpolynomial-length adaptive input for which the expected average latency over all operations is $\Theta(P)$.

\begin{proposition}
  There exists an adaptive input such that, with high probability in $P$, \cref{alg:basic-cas-register} incurs an expected average latency of $\Theta(P)$, where the average is taken over all operations.
\end{proposition}

\begin{proof}
  We start by constructing the input. We let the user repeatedly:
  \begin{itemize}
    \item Invoke a \Read operation to get the current value $x$ of the register.
    \item Invoke as many \opCAS operations as possible, each trying to change the value of the register from $x$ to $x + 1$, subject to the constraint that the number of ongoing operations cannot exceed $P$.
  \end{itemize}
  We let the user repeat this process for $2^P$ timesteps.

  To show that this input pattern causes $\Theta(P)$ latency for each operation in expectation, we keep track of the size of the instruction queue after each time the value of the register changes. In common cases, the instruction queue size stays at most $c^3 \log P$ according to \cref{lem:basic-cas/no-many-awaiting-cas}, and the queue empties itself before the next operations invoke \CAS instructions. However, the following sequence of events happens with probability $2^{-O(\log^2 P)}$:
  \begin{enumerate}
    \item The potential grows to $\Theta(\log^2 P)$ before any operation invokes a \CAS instruction.
    \item $\Theta(\log^2 P)$ operations invoke \CAS instructions simultaneously in timestep $t$. This causes the size of the instruction queue to increase to $\Theta(\log^2 P)$.
    \item The head instruction in the queue gets applied and succeeds, changing the value of \CellC.
    \item\label{item:basic-cas/lb/bad-interval} In the next $\Theta(\log^2 P)$ timesteps, all applied \CAS instruction fail, so the value of \CellC does not change. We call this a ``bad interval''.
    \item\label{item:basic-cas/lb/queue-size} Each operation invoked in the bad interval gets scheduled for $\Theta(\log P)$ times in $\Theta(\log P)$ timesteps, but they do not detect a change of \CellC, so they all invoke \CAS instructions in $\Theta(\log P)$ timesteps after being invoked. This causes the queue size to increase to $\Theta(P)$.
  \end{enumerate}
  Then, steps \ref{item:basic-cas/lb/bad-interval} and \ref{item:basic-cas/lb/queue-size} repeat: the value of \CellC only changes every $\Theta(P)$ timesteps, which causes the queue size to stay at $\Theta(P)$ forever.

  Based on the above events, for every $\Theta(P)$ timesteps, there is a $2^{-O(\log^2 P)}$ probability that the queue size increases to $\Theta(P)$. This event happens in the first $1.5^P \gg 2^{O(\log^2 P)} \cdot \log P$ timesteps with high probability. Once this happens, the queue size will stay at $\Theta(P)$ until the end of the execution, causing all remaining operations to incur $\Theta(P)$ latency. Therefore, the average latency of all operations is at least $\Theta(P)$.
\end{proof}

\subsection{Improved CAS Register}
\label{sec:improved-cas-register}

In this section, we propose an improved CAS register algorithm that has $O(\log P)$ high-probability latency guarantee even for unbounded-length executions.

\paragraph{Algorithm description.}

In the improved algorithm, we use the basic CAS algorithm (\cref{alg:basic-cas-register}) as a subroutine to maintain the register value (and a fingerprint) in \CellC. In addition, we maintain an auxiliary cell called \CellW, which is used to detect that the instruction queue for \CellC is non-empty. After an operation $q$ applies a \CAS instruction in \cref{alg:basic-cas-register}, we let it write a random $2 \log P$-bit string in \CellW regardless of whether its \CAS instruction succeeds or fails. See \cref{alg:improved-cas} for details. Here, $c$ on line~\ref{line:improved-cas/wait} denotes the same constant as in \cref{alg:basic-cas-register}.

\begin{algorithm2e}[ht]
  \caption{Improved CAS Register}\label{alg:improved-cas}
  \DontPrintSemicolon
  \Fn{\ImprovedCAS{$\xExpect$, $\xNew$}} {
    \tcpNoSpace{Waiting Phase}
    $\yOld \gets \Load(\CellW)$ \tcp{\instSS}
    \Loop{
      wait for $2c \log P$ steps \; \label{line:improved-cas/wait}
      $\yNow \gets \Load(\CellW)$ \tcp{\instRR}
      \If{$\yNow = \yOld$}{
        \Break \;
        \label{line:improved-cas/break}
      }
      $\yOld \gets \yNow$ \;
    }
    \tcpNoSpace{Calling Phase}
    result $\gets \BasicCAS(\CellC, \, \xExpect, \, \xNew)$ \;
    \tcpNoSpace{Writing Phase}
    \If{\BasicCAS applied a \CAS instruction}{
      $\StoreRandom(\CellW, \; \Unif(\BK{0,1}^{2 \log P}))$ \tcp{\instWW}
    }
    \Return result \;
  }
\end{algorithm2e}

In \cref{alg:improved-cas}, in addition to the shared instructions in \cref{alg:basic-cas-register}, there are three additional shared instructions: \instSS, \instRR, and \instWW. Moreover, we treat line~\ref{line:improved-cas/wait} as a special \Nop instruction that gets repeated for $2c \log P$ times before the operation moves on to the next state. We again adopt the view of the \emph{reduced state machine} similarly to \cref{sec:register,sec:basic-cas-register}.

The execution of each \opCAS operation $q$ is divided into three phases: the Waiting Phase, the Calling Phase, and the Writing Phase. The Waiting Phase loops until in one iteration of line~\ref{line:improved-cas/wait} the value of \CellW does not change. Then, $q$ transitions to the Calling Phase. In the Calling Phase, $q$ calls \BasicCAS (\cref{alg:basic-cas-register}) to perform the \opCAS operation on \CellC. The \opCAS operation in \BasicCAS either quits directly or applies a \CAS instruction to \CellC. In the latter case, $q$ transitions to the Writing Phase, and writes a random $2 \log P$-bit string in \CellW (\instWW).

\begin{restatable}{lemma}{ImprovedCASLinearizable}
  \label{lem:improved-cas/linearizable}
  \cref{alg:improved-cas} is linearizable and lock-free.
\end{restatable}

\begin{proofof}{\cref{lem:improved-cas/linearizable}}
  Linearizability follows immediately from \cref{lem:basic-cas/linearizable}, as every operation ends up invoking \BasicCAS (no operation can return without doing so); the linearization given in the proof of \cref{lem:basic-cas/linearizable} applies.

  Lock-freedom follows from the wait-freedom of the basic CAS algorithm, together with the fact that not all processes can become stuck in the loop: otherwise, there would be an infinite suffix where no process writes to $\CellW$, allowing all processes to break out of the loop the next time they execute line~\ref{line:improved-cas/break}.
  Thus, some process eventually breaks out of the loop, invokes \BasicCAS, and returns.
\end{proofof}

In the remainder of this subsection, we prove \cref{lem:improved-cas/latency-guarantee}.

\begin{restatable}{lemma}{ImprovedCASLatencyGuarantee}
  \label{lem:improved-cas/latency-guarantee}
  Let $t$ be a fixed timestep and let $q$ be an arbitrary \opCAS operation invoked in timestep $t$. Then, operation $q$ completes within $O(\log P)$ time with high probability in $P$.
\end{restatable}

\subsubsection{Analyzing Calling Phase}

\begin{definition}
  \label{defn:improved-cas/bad-interval}
  We say the data structure is in a \defn{bad state} at time $t$ if the data structure satisfies one of the following conditions:
  \begin{itemize}
    \item There are awaiting CAS instructions.
    \item The potential $\phi_t$ is at least $c \log P$.
  \end{itemize}
  We use \defn{bad interval} to denote a maximal interval $[t_1, t_2)$ where the data structure is in a bad state at every time $t \in [t_1, t_2)$.
\end{definition}

\begin{lemma}
  \label{lem:improved-cas/bad-interval-end}
  Conditioned on $\HistExt(t)$, with probability $1 - P^{-\Omega(1)}$, the bad interval containing $t$ (if exists) does not extend beyond $t + P^2$.
\end{lemma}

\begin{proof}
  We start by pointing out two observations about bad intervals.

  \begin{observation}
    \label{obs:improved-cas/cas-after-bad}
    Conditioned on $\HistExt(t)$ where the data structure is in a bad state at time $t$, with high probability in $P$, a \CAS instruction is applied on \CellC during $[t, \; t + c \log P)$.
  \end{observation}

  \begin{proof}
    By \cref{defn:improved-cas/bad-interval}, either of the following conditions hold:
    \begin{itemize}
      \item \BoldHeading{There exists an awaiting \CAS instruction at time $t$.}
      If this instruction is waiting in the queue, then one instruction is applied immediately in timestep $t$. Otherwise, this instruction is waiting to be invoked (its corresponding operation is pending on \instC). By \cref{prop:process-schedule-prob}, we know that this instruction will be invoked before time $t + c \log P$ with high probability in $P$, which further implies that at least one \CAS instruction is applied before time $t + c \log P$.
      \item \BoldHeading{$\phi_t \ge c \log P$.}
      According to \cref{prop:random-delay-for-operations}, every active operation $q \in A_t$ has an independent probability core%
      \footnote{We say events $\BK{X_i}$ have independent probability cores $\BK{p_i}$ to occur if there are random variables $\BK{Y_i}$ such that $Y_i \le X_i$ and $\Pr[Y_i] \ge p_i$ for all $i$, and that $Y_i$ are mutually independent. We will use this language for convenience.}
      of $p^{(q)}_t / 4$ to apply \instR before time $t + \frac{c}{2} \log P$, transition to \instC, then apply \instC to invoke a \CAS instruction before time $t + c \log P$, unless \CellC changes its value before then. By a Chernoff bound, we know that with high probability in $P$, at least one \CAS instruction is invoked before time $t + c \log P$, and thus at least one \CAS instruction is applied before time $t + c \log P$.
    \end{itemize}
    The observation holds in both cases.
  \end{proof}

  \begin{observation}
    \label{obs:improved-cas/cell-w-change-after-bad}
    Conditioned on $\HistExt(t)$ where the data structure is in a bad state at time $t$, with high probability in $P$, a \StoreRandom on \CellW is applied during $[t, \; t + 2c \log P)$.
  \end{observation}

  \begin{proof}
    By \cref{obs:improved-cas/cas-after-bad}, with high probability in $P$, a \CAS instruction is applied on \CellC during $[t, \; t + c \log P)$. Then, the operation that applies this \CAS instruction becomes pending on \instWW, and it invokes \instWW before time $t + 2c \log P$ with high probability in $P$. This further implies that at least one \StoreRandom on \CellW is applied before time $t + 2c \log P$.
  \end{proof}

  Now, we set $T = P^2$ and analyze the probability that the bad interval extends beyond $t + T$. We assume the conditions in \cref{obs:improved-cas/cas-after-bad,obs:improved-cas/cell-w-change-after-bad} hold for all $t' \in [t, \; t + T)$ before the bad interval ends, which happens with high probability in $P$ by a union bound. Since a \CAS instruction is applied every $c \log P$ timesteps, there must be at least $T / (c \log P)$ \CAS instructions getting applied in the bad interval. These \CAS instructions consist of the following categories:
  \begin{itemize}
    \item \BoldHeading{Awaiting \CAS instructions at time $t$.}
    There are at most $P$ such instructions.
    \item \BoldHeading{\CAS instructions invoked by operations that enter the Calling Phase during $[t, \; t + 2c \log P)$.}
    Since each operation waits for at least $2c \log P$ timesteps before it enters the calling phase, these operations are already ongoing at time $t$. So there are at most $P$ such operations.
    \item \BoldHeading{\CAS instructions invoked by operations that enter the Calling Phase during $[t + 2c \log P, \; t + T)$.}
    We bound the expected number of operations in this category as follows. For any timestep $t' \in [t + 2c \log P, \; t + T)$, we bound the expected number of operations that apply \instRR in timestep $t'$ and transition to the Calling Phase. In order for any operation $q$ to apply \instRR in timestep $t'$, it must have applied its previous \instSS or \instRR instruction before time $t' - 2c \log P$, which records the old value $\yOld$ of \CellW. By \cref{obs:improved-cas/cell-w-change-after-bad}, at least one \StoreRandom is applied during $[t' - 2c \log P, \; t')$, each of which writes a uniformly random $2 \log P$-bit string to \CellW that is independent of $\yOld$.

    In order for $q$ to apply \instRR in timestep $t'$ and transition to the Calling Phase, there must be a timestep $t'' \in [t' - 2c \log P, \; t')$ such that the \StoreRandom applied in timestep $t''$ writes the same string as $\yOld$. The probability of this event is bounded by $\frac{2c \log P}{P^2}$. When this event happens, at most $P$ operations apply \instRR in timestep $t'$ and transition to the Calling Phase; when this event does not happen, no operation can apply \instRR in timestep $t'$ and transition to the Calling Phase. Therefore, the expected number of operations that apply \instRR in timestep $t'$ and transition to the Calling Phase is at most $P \cdot \frac{2c \log P}{P^2} = \frac{2c \log P}{P}$. Taking a summation over all $t' \in [t + 2c \log P, \; t + T)$, we know the expected number of operations in this category is at most $T \cdot \frac{2c \log P}{P} = 2c P \log P$. By Markov's inequality, the number of operations in this category is at most $P^{1.6}$ with probability at least $1 - 1 / \sqrt{P}$.
  \end{itemize}
  Combining all categories together, we know that the number of \CAS instructions applied in the bad interval is at most $P + P + P^{1.6} = O(P^{1.6}) = o(T)$ with probability at least $1 - O(1 / \sqrt{P})$. This further implies that the bad interval does not extend beyond $t + T$ with probability at least $1 - O(1 / \sqrt{P}) = 1 - P^{-\Omega(1)}$.
\end{proof}

\begin{lemma}
  \label{lem:improved-cas/worst-case-recovery}
  Conditioned on $\HistExt(t)$, with probability $1 - P^{-\Omega(1)}$, the data structure enters a healthy state in $[t, \; t + 2P^2]$.
\end{lemma}

\begin{proof}
  By \cref{clm:basic-cas/no-slow-operation} and a union bound, we know that conditioned on $\HistExt(t)$, with high probability in $P$, there is no \emph{heavily-delayed} operation at any time $t' \in [t + \log^2 P, \; t + 2P^2]$. We assume this condition holds in the rest of the proof.

  Applying \cref{lem:improved-cas/bad-interval-end} at time $t + \log^2 P$, we obtain that with probability $1 - P^{-\Omega(1)}$, there is a time $t_1 \in [t + \log^2 P, \; t + P^2 + \log^2 P]$ where the data structure is \emph{not} in a bad state. This combined with the fact that there is no heavily-delayed operation at time $t_1$ satisfies the required conditions for \cref{lem:basic-cas/enter-healthy-state}. \cref{lem:basic-cas/enter-healthy-state} further states that, conditioned on $\HistExt(t_1)$, the data structure enters a healthy state in $[t_1, \; t_1 + 2c^4 \log P] \subseteq [t, \; t + 2P^2]$ with high probability in $P$.
\end{proof}

\begin{corollary}
  \label{cor:improved-cas/worst-case-recovery}
  Conditioned on $\HistExt(t)$, the data structure enters a healthy state in $[t, \; t + O(P^2)]$ with high probability in $P$.
\end{corollary}

\begin{proof}
  Starting from time $t$ and repeatedly applying \cref{lem:improved-cas/worst-case-recovery} yields the corollary.
\end{proof}

\begin{lemma}
  \label{lem:improved-cas/enter-healthy-state}
  For any given time $t \ge 0$, conditioned on $\HistExt(t - \Theta(P^2))$, the data structure enters a healthy state in time interval $[t, \; t + O(\log P)]$ with high probability in $P$.
\end{lemma}

\begin{proof}
  Let $t_1$ be the last time when the data structure is in a healthy state before $t + 3c^4 \log P$. If $t_1 \ge t$, then the desired property holds. Otherwise, the event $t_1 < t$ is witnessed by low-probability events depending on $t_1$. Let $\TThreshold = \Theta(P^2)$ be a sufficiently large threshold.

  \paragraph{For $t_1 \ge t - \TThreshold$.}
  Conditioned on $\HistExt(t_1 - \TThreshold) = \HistExt(t - \Theta(P^2))$, with high probability in $P$:
  \begin{itemize}
    \item Applying \cref{clm:basic-cas/no-slow-operation} at time $t_1$, we know that there is no heavily-delayed operation at time $t_1$.
    \item By \cref{lem:basic-cas/enter-healthy-state}, the data structure enters a healthy state before $t_1 + 2c^4 \log P < t + 3c^4 \log P$, contradicting the definition of $t_1$.
  \end{itemize}

  \paragraph{For $t_1 < t - \TThreshold$.}
  We apply \cref{cor:improved-cas/worst-case-recovery} at time $t - \TThreshold$, which states that conditioned on $\HistExt(t - \TThreshold)$, the data structure enters a healthy state in $[t - \TThreshold, \; t - \TThreshold + O(P^2)] \subseteq [t - \TThreshold, \; t - 1]$ with high probability in $P$. This contradicts the definition of $t_1$ with high probability in $P$.

  In summary, we know that $t_1$ takes each value in $[t - \TThreshold, \; t - 1]$ with low probability in $P$, and the probability that $t_1 < t - \TThreshold$ is also at most $1 / \poly(P)$. A union bound yields the lemma.
\end{proof}

\begin{lemma}[Calling Phase Latency]
  \label{lem:improved-cas/calling-phase}
  For a fixed time $t \ge 0$, with high probability in $P$, any operation $q$ that enters the Calling Phase at time $t$ completes the Calling Phase within $O(\log P)$ time.
\end{lemma}

\begin{proof}
  The following holds with high probability in $P$:
  After entering the Calling Phase, the operation $q$ is scheduled for at most $\Theta(\log P)$ times before it either quits or invokes a \CAS instruction, which takes at most $\Theta(\log P)$ time by \cref{prop:process-schedule-prob}. Then, by \cref{lem:improved-cas/enter-healthy-state}, the next healthy state of the data structure happens within $O(\log P)$ time after the \CAS instruction is invoked (if $q$ does not quit before then), at which point the Calling Phase of $q$ is guaranteed to finish.
\end{proof}

\subsubsection{Analyzing Writing Phase}

Next, we analyze the latency of the Writing Phase. The proof is based on analyzing \emph{busy intervals} of \CellW.

\begin{definition}[Busy Interval]
  A \defn{busy interval} of \CellW is a \emph{maximal} interval $[t_0, t_1)$ where
  (1) one \StoreRandom is applied on \CellW in every timestep $t \in [t_0, t_1)$;
  (2) the instruction queue of \CellW is non-empty at each \emph{internal} time $t \in (t_0, t_1)$.
\end{definition}

\begin{claim}
  \label{clm:writing-phase/pending-w'}
  For any time $t \ge 0$, conditioned on an arbitrary $\HistExt(t - P)$, the number of operations pending on \instWW at time $t$ is at most $c \log P$ with high probability in $P$.
\end{claim}

\begin{proof}
  The proof is similar to \cref{lem:busy-interval-length}: in each timestep, at most one \CAS instruction gets applied, and its corresponding operation immediately transitions to the Writing Phase and becomes pending on \instWW. However, according to \cref{prop:process-schedule-prob}, the operation that transitions to \instWW in timestep $t - j$ for $j \ge 1$ has an independent probability core of $1 - 2^{-\Omega(j)}$ to invoke \instWW before time $t$, which means it is no longer pending on \instWW by time $t$. Combining with the fact that there are at most $P$ operations pending on \instWW at time $t - P$, a Chernoff bound on the aforementioned probability cores yields the claim.
\end{proof}

In the next part of the proof, we define $\psi_t$ to be the number of \defn{awaiting} \StoreRandom instructions on \CellW at time $t$, which includes both the \StoreRandom instructions waiting in the queue of \CellW and the operations pending on \instWW. One nice property of $\psi_t$ is that it is monotonically non-increasing in busy intervals: in every timestep in the busy interval, at most one operation transitions to \instWW, which increases $\psi$ by at most 1, while one \StoreRandom instruction gets applied, which decreases $\psi$ by 1. The busy interval ends when $\psi$ decreases to 0.

\begin{lemma}
  \label{lem:writing-phase/busy-interval-length}
  For any time $t \ge 0$, conditioned on $\HistExt(t - \Theta(P^2))$, the length of the busy interval starting at $t$ (if exists) is at most $O(\log P)$ with high probability in $P$.
\end{lemma}

\begin{proof}
  Conditioned on $\HistExt(t - \Theta(P^2))$, the following holds with high probability in $P$:
  By \cref{clm:writing-phase/pending-w'}, we know that $\psi_{t} \le c \log P$ at time $t$ (when the busy interval starts).
  Let $t_1 \ge t$ be the next time when the data structure enters a healthy state, where $t_1 \le t + O(\log P)$ by \cref{lem:improved-cas/enter-healthy-state}. By the monotonicity of $\psi$, we know that $\psi_{t_1} \le \psi_t \le c \log P$ unless the busy interval ends before $t_1$.
  By \cref{clm:basic-cas/no-slow-operation,clm:basic-cas/healthy-state-keep-silent}, with high probability in $P$, there is no heavily-delayed operation at time $t_1$, and no \CAS instruction is invoked or applied during $[t_1, \; t_1 + \tfrac{c^2}{2} \log P)$. Therefore, $\psi$ decreases by 1 in each timestep after $t_1$ and before the busy interval ends. This further implies that the busy interval ends before $t_1 + c \log P = t + O(\log P)$.
\end{proof}

\begin{lemma}
  \label{lem:writing-phase/busy-interval-length-P^3}
  For any time $t \ge 0$, conditioned on $\HistExt(t)$, the busy interval containing $t$ (if exists) ends before $t + O(P^3)$ with high probability in $P$.
\end{lemma}

\begin{proof}
  We apply \cref{cor:improved-cas/worst-case-recovery} on time $t + \frac{1}{2} \log^2 P$, which states that conditioned on $\HistExt\bk[\big]{t + \frac{1}{2} \log^2 P} \supseteq \HistExt(t)$, with high probability in $P$, the data structure enters a healthy state at time $t_1 \in \Bk[\big]{t + \frac{1}{2} \log^2 P, \; t + O(P^2)}$. By \cref{clm:basic-cas/no-slow-operation,clm:basic-cas/healthy-state-keep-silent}, conditioned on $\HistExt(t)$, with high probability in $P$, there is no heavily-delayed operation at time $t_1$, and no \CAS instruction is invoked or applied during $[t_1, \; t_1 + \tfrac{c^2}{2} \log P)$. In particular, no \CAS instruction is applied during $[t_1, \; t_1 + 2 c \log P)$, so $\psi$ decreases by at least $2 c \log P$ during this interval unless the busy interval ends before $t_1 + 2 c \log P$. This gives us a time $t' \defeq t_1 + 2 c \log P \le t + O(P^2)$ such that either $\psi_{t'} \le \psi_{t} - 2 c \log P$ or the busy interval ends before $t'$. Noting that $\psi_{t} \le P$, starting from time $t$ and repeatedly applying this argument for at most $O(P)$ times yields the desired bound.
\end{proof}

\begin{corollary}
  \label{cor:writing-phase/busy-interval-length}
  For a fixed time $t \ge 0$, conditioned on an arbitrary $\HistExt(t - \Theta(P^3))$, the busy interval containing $t$ (if exists) ends before time $t + O(\log P)$ with high probability in $P$.
\end{corollary}

This corollary can be obtained from \cref{lem:writing-phase/busy-interval-length,lem:writing-phase/busy-interval-length-P^3} via the same argument as in \cref{cor:busy-interval-endpoint}.

\begin{corollary}
  \label{cor:writing-phase/psi-bound}
  For any time $t \ge 0$, conditioned on an arbitrary $\HistExt(t - \Theta(P^3))$, $\psi_t \le c \log P$ with high probability in $P$.
\end{corollary}

\begin{proof}
  If the instruction queue of \CellW is empty at time $t$, then \cref{clm:writing-phase/pending-w'} directly yields the corollary. Otherwise, conditioned on $\HistExt(t - \Theta(P^3))$, with high probability in $P$, the following holds:
  \begin{itemize}
    \item We define $t_1 \le t$ as the starting point of the busy interval containing $t$, which means the instruction queue of \CellW is empty at time $t_1$. \cref{cor:writing-phase/busy-interval-length} implies $t_1 \ge t - O(\log P)$. Next, applying \cref{clm:writing-phase/pending-w'} at time $t_1$ yields that $\psi_{t_1} \le c \log P$, then the corollary follows by the monotonicity of $\psi$ in a busy interval. \qedhere
  \end{itemize}
\end{proof}

\begin{lemma}[Writing Phase Latency]
  \label{lem:improved-cas/writing-phase}
  For a fixed time $t \ge 0$, with high probability in $P$, any operation $q$ that enters the Writing Phase at time $t$ completes the Writing Phase within $O(\log P)$ time.
\end{lemma}

\begin{proof}
  Each operation $q$ that enters the Writing Phase first becomes pending on \instWW, which takes at most $O(\log P)$ time before it invokes \instWW with high probability in $P$. Let $t_1 \le t + O(\log P)$ be the timestep when $q$ invokes \instWW. Then, by \cref{cor:writing-phase/busy-interval-length} at time $t_1 + 1$, the busy interval containing $t_1 + 1$ (if exists) ends before $t_1 + O(\log P)$ with high probability in $P$, at which point the \instWW instruction is guaranteed to be applied, so $q$ completes the Writing Phase within $O(\log P)$ time.
\end{proof}

\subsubsection{Analyzing Waiting Phase and Total Latency}

\begin{lemma}[Waiting Phase Latency]
  \label{lem:improved-cas/waiting-phase}
  For a fixed time $t \ge 0$, with high probability in $P$, any operation $q$ that enters the Waiting Phase at time $t$ completes the Waiting Phase within $O(\log P)$ time.
\end{lemma}

\begin{proof}
  With high probability in $P$, we have the following.
  \begin{itemize}
    \item Let $t_1 \ge t$ be the first time at or after $t$ when the data structure is in a healthy state. By \cref{lem:improved-cas/enter-healthy-state}, there is $t_1 \le t + O(\log P)$.
    \item By \cref{clm:basic-cas/no-slow-operation,clm:basic-cas/healthy-state-keep-silent}, there is no heavily-delayed operation at time $t_1$, and no \CAS instruction is invoked or applied during $[t_1, \; t_1 + (c^2 / 2) \log P)$. This implies that no operation can start pending on \instWW before time $t_1 + (c^2 / 2) \log P$. Hence, in time interval $[t_1, \; t_1 + (c^2 / 2) \log P)$, $\psi$ will decrease by one for each timestep until the busy interval ends.
    \item \cref{cor:writing-phase/psi-bound} implies that $\psi_{t_1} \le c \log P$. This combined with the previous statement implies that the busy interval ends before $t_1 + c \log P$, and that $\psi_{t'}$ remains zero for all $t' \in [t_1 + c \log P, \; t_1 + (c^2 / 2) \log P]$, i.e., there is an interval of length $(c^2 / 2 - c) \log P$ where the value of \CellW does not change.
    \item By \cref{prop:process-schedule-prob}, during $[t, \; t + P]$, every iteration of the waiting loop of operation $q$ takes at most $(2 c \log P) \cdot (4 \tau) \le (c^2 / 8) \log P$ steps. Therefore, one of the iterations will reside in the interval $[t_1 + c \log P, \; t_1 + (c^2 / 2) \log P]$ and detects no change in \CellW. The Waiting Phase completes immediately after that iteration (if not earlier). This implies that $q$ completes the Waiting Phase within $O(\log P)$ time.
    \qedhere
  \end{itemize}
\end{proof}

Combining \cref{lem:improved-cas/waiting-phase,lem:improved-cas/calling-phase,lem:improved-cas/writing-phase} yields the desired lemma \ref{lem:improved-cas/latency-guarantee}.

\ImprovedCASLatencyGuarantee*

\subsection{Putting Pieces Together}
\label{sec:cas/conclusion}

We conclude the section by combining the pieces together to prove
\cref{thm:cas-register-main}.

\CasRegisterMain*

\begin{proofof}{\cref{thm:cas-register-main}}
  Our improved CAS register algorithm in \cref{sec:improved-cas-register} maintains two words in the shared memory: \CellC contains a pair $(x, f)$, where $x$ is a $\l$-bit integer that represents the current value of the CAS register, and $f$ is a $2 \log \log P$-bit fingerprint; \CellW contains a random $2 \log P$-bit string. Thus, the word size of the algorithm needs to be $w \ge \max\BK{\l + 2 \log \log P, \; 2 \log P}$. In addition to the two words in the shared memory, each operation only uses $O(1)$ words of local memory.

  In our algorithm, every \Read operation is completed by a \Load instruction to \CellC, and thus completes within $O(\log P)$ time with high probability in $P$, and $O(1)$ time in expectation; \opCAS operations are implemented using \cref{alg:improved-cas}. According to \cref{lem:improved-cas/latency-guarantee}, for any fixed time $t$, any \opCAS operation invoked at time $t$ completes within $O(\log P)$ time with high probability in $P$. Moreover, \cref{lem:improved-cas/linearizable} states that the improved CAS register algorithm is linearizable and lock-free.
  This completes the proof.
\end{proofof}

\section{Lower Bound}
\label{sec:lower-bound}

In this section, we prove a lower bound on the expected latency of \emph{1-bit max registers}. Each 1-bit max register maintains a single bit $x$ that is initially 0, and supports two types of operations: a \Read operation returns the current value of $x$, while a \Set operation sets $x$ to 1. This problem is arguably the simplest data structure problem that allows updates to the state of the data structure, and it can be reduced to a variety of fundamental and well-studied data structure problems in the concurrent setting, including but not limited to read/write registers, CAS registers, test-and-set objects, sticky bits, and more (see the textbook~\cite{HS12}).
Therefore, the lower bound we prove for 1-bit max registers directly applies to these problems as well.

\begin{theorem}[Lower Bound for 1-bit Max Registers]
  \label{thm:lower-bound-1-bit-max-reg}
  Suppose $\mathcal{A}$ is a \emph{sequentially correct} implementation of a 1-bit max register in the \emph{asynchronous CRQW} model that uses $M$ machine words in the shared memory, and for any execution of at most $P$ operations under \emph{one-shot} inputs, every \Set operation completes within $L \in \Omega(\log P)$ timesteps with high probability in $P$. Then, there exists a \emph{one-shot} input in which some \Set operation has expected latency $\Omega\bk[\big]{\log_{ML} P}$. In particular, if $M, L \le \polylog P$, then the expected latency becomes $\Omega\bk[\big]{\log P / \log \log P}$.
\end{theorem}

Before we prove the theorem, we make some comments on the setting.
\begin{itemize}
  \item The lower bound does not require the implementation to be \emph{linearizable}. It only needs an extremely weak correctness assumption, which we call \emph{sequential correctness}: if at any time in the execution, there is at most one ongoing operation, then the data structure must be correct. This forces each \Set operation to invoke a state-changing instruction to the shared memory, unless it detects that other operations have made any change in the shared memory.
  \item The lower bound holds in the \emph{asynchronous CRQW} model~\cite{GibbonsMR98b},
  which is the same as our \emph{stochastic CRQW} model except that in every timestep, all ready operations are guaranteed to be scheduled. Since the \emph{asynchronous CRQW} model is weaker than the stochastic CRQW model, the lower bound in \cref{thm:lower-bound-1-bit-max-reg} directly applies to the stochastic CRQW model as well.
  \item Instead of assuming the user gives adaptive inputs as in \cref{sec:model}, we only need a much weaker user model for this lower bound on latency: \defn{one-shot input}, in which the user can only invoke operations in the first timestep in the execution. This definition is weaker than many natural user models that one might consider, including the adaptive user model in \cref{sec:model}, so the lower bound applies to all these user models. In addition, recall that it is challenging to define expected latency without conditioning on the invocation of an operation when the user is adaptive. The one-shot input model does not suffer from this issue, as the operations invoked by the user are fixed in advance and not adaptive to the execution.
  \item The lower bound of $\Omega\bk[\big]{\log P / \log \log P}$ assumes a polylogarithmic space blow-up ($M \le \polylog P$) and a polylogarithmic high-probability latency ($L \le \polylog P$). Whether we can remove these assumptions remains an open question.
  \item As mentioned in \cref{sec:model}, in this lower bound, we assume each operation sees the index of the process that is running the operation.
\end{itemize}

\begin{proof}
  For each set $S \subseteq [P]$, let $\alpha_S$ denote the one-shot execution in which exactly the processes in $S$ begin a \Set operation at time $0$.

  For each process $i \in [P]$, let $T_i$ be the timestep in which process $i$ invokes its first state-changing instruction in the singleton execution $\alpha_{\{i\}}$. Since $\alpha_{\{i\}}$ contains only process $i$, the random variable $T_i$ depends only on $i$'s private randomness.

  It is convenient to analyze the following \emph{decoupled game}. For a fixed set $S$, each process $i \in S$ runs independently according to the singleton execution $\alpha_{\{i\}}$. Each process halts immediately after it invokes its first state-changing instruction, but the simulation does \emph{not} stop when the first such instruction is invoked: the other processes continue running independently until they reach their own first state-changing instructions. We define
  \[
    T^*(S) \defeq \min_{i \in S} T_i
  \]
  as the first timestep when any process in the decoupled game invokes a state-changing instruction, and define the \defn{loss} of the game by
  \[
    \mathcal{L}(S) \defeq \#\{i \in S : T_i = T^*(S)\},
  \]
  i.e., the number of processes that invoke the earliest state-changing instructions simultaneously.

  This decoupled game has the same first-invocation time and the same loss as the real execution $\alpha_S$. Indeed, before time $T^*(S)$, no process invokes a state-changing instruction in its singleton execution, so the shared memory in $\alpha_S$ remains in its initial state up to time $T^*(S)$. Therefore, every process $i \in S$ sees exactly the same execution prefix in $\alpha_S$ as in $\alpha_{\{i\}}$ up to time $T^*(S)$, and so process $i$ invokes its first state-changing instruction in $\alpha_S$ exactly in timestep $T_i$. After timestep $T^*(S)$, the loss is already determined, so the later interaction among the processes is irrelevant. A key property of the decoupled game is that the random variables $\BK{T_i}_{i \in S}$ are independent, which will facilitate the analysis.

  Sequential correctness implies that in $\alpha_S$, no process can complete its \Set before time $T^*(S)$: if some process returned earlier, then the same process would also return earlier in its singleton execution, without any state-changing instruction having been invoked, and a subsequent \Read would incorrectly return $0$ in a sequential execution.

  The high-probability latency bound implies that
  \[
    \Pr[\mathcal{L}(S) \le ML] \ge 1 - 1/\poly(P)
  \]
  for every set $S$. Indeed, if $\mathcal{L}(S) > ML$, then in the first timestep $T^*(S)$ when any state-changing instruction is invoked, some memory cell receives at least $L + 1$ such instructions by the pigeonhole principle. Since the instructions are queued and applied one at a time, one of them is applied only after time $L$, contradicting the high-probability latency bound.

  For each process $i$, define
  \[
    p_t(i) \defeq \Pr[T_i < t],
  \]
  and for each set $S \subseteq [P]$, define
  \[
    p_t(S) \defeq \sum_{i \in S} p_t(i).
  \]

  \begin{claim}
    \label{clm:lower-bound/no-large-jump}
    For every set $S \subseteq [P]$ and every time $t$, it is impossible to have
    \[
      p_t(S) \le 1/2
      \qquad\text{and}\qquad
      p_{t+1}(S) \ge 3ML.
    \]
  \end{claim}

  \begin{proof}
    For each $i \in S$, let $X_i \defeq \ind[T_i < t]$ and $Y_i \defeq \ind[T_i = t]$.
    Because the decoupled game runs the processes independently, the random variables $\BK{X_i, Y_i}_{i \in S}$ are independent.

    Let $X \defeq \sum_{i \in S} X_i$ and $Y \defeq \sum_{i \in S} Y_i$. Then $\E[X] = p_t(S) \le 1/2$, so Markov's inequality gives
    \[
      \Pr[X = 0] \ge 1/2. \numberthis \label{eq:lower-bound/prob-X-zero}
    \]
    Also, $\E[Y] = p_{t+1}(S) - p_t(S) \ge 3ML - 1/2 > 2ML$.
    Since $Y$ is a sum of independent Bernoulli random variables, a Chernoff bound and the assumption $L \in \Omega(\log P)$ imply
    \[
      \Pr[Y > ML] \ge 1 - 1/P^{\Omega(1)}. \numberthis \label{eq:lower-bound/prob-Y-large}
    \]
    Combining \eqref{eq:lower-bound/prob-X-zero} and \eqref{eq:lower-bound/prob-Y-large}, we know that $X = 0$ and $Y > ML$ holds simultaneously with constant probability. On this event, no process invokes a state-changing instruction before time $t$, while more than $ML$ processes do so in timestep $t$. Hence $T^*(S) = t$ and $\mathcal{L}(S) = Y > ML$, contradicting the high-probability latency guarantee.
  \end{proof}

  \begin{claim}
    \label{clm:growth-bound}
    For any set $S$ and time $t$, we have
    \[
      p_{t+1}(S) \le 12ML \cdot (p_t(S) + 1).
    \]
  \end{claim}

  \begin{proof}
    We partition $S$ into two disjoint subsets:
    \[
      \SLarge = \{i \in S : p_t(i) > 1/4\},
      \qquad
      \SSmall = \{i \in S : p_t(i) \le 1/4\}.
    \]

    For $\SLarge$, we simply use $p_{t+1}(i) \le 1 < 4p_t(i)$ for every $i \in \SLarge$, which gives
    \[
      p_{t+1}(\SLarge) < 4p_t(\SLarge).
      \numberthis \label{eq:merged1-large}
    \]

    For $\SSmall$, we greedily partition it into disjoint subsets $\tilde{S}_1, \dots, \tilde{S}_k$ such that
    \begin{itemize}
      \item $p_t(\tilde{S}_j) \in [1/4, 1/2]$ for all $j < k$, and
      \item $p_t(\tilde{S}_k) \in [0, 1/2]$.
    \end{itemize}
    Since each set $\tilde{S}_j$ for $j < k$ contributes at least $1/4$ to $p_t(\SSmall)$, we have
    \[
      k \le 4p_t(\SSmall) + 1.
    \]
    By \cref{clm:lower-bound/no-large-jump}, each $\tilde{S}_j$ satisfies $p_{t+1}(\tilde{S}_j) < 3ML$, and hence
    \[
      p_{t+1}(\SSmall)
      =
      \sum_{j=1}^k p_{t+1}(\tilde{S}_j)
      <
      3ML \cdot k
      \le
      12ML \cdot p_t(\SSmall) + 3ML
      <
      12ML \cdot (p_t(\SSmall) + 1).
      \numberthis \label{eq:merged1-small}
    \]
    Adding \eqref{eq:merged1-large} and \eqref{eq:merged1-small} yields the claim.
  \end{proof}

  Finally, we let $E$ be an upper bound on the expected latency of \Set operations over all one-shot executions, and show that $E = \Omega(\log_{ML} P)$.

  Consider the singleton execution $\alpha_{\{i\}}$. Since no \Set can complete before time $T_i$, if we had $p_{2E}(i) < 1/2$, then the \Set in $\alpha_{\{i\}}$ would have latency greater than $2E$ with probability more than $1/2$, contradicting the definition of $E$. Therefore,
  \[
    p_{2E}(i) \ge 1/2
    \qquad \text{for every } i \in [P].
  \]
  Summing over all processes gives
  \[
    p_{2E}([P]) \ge P/2.
    \numberthis \label{eq:merged1-p2e}
  \]

  Since $p_0([P]) = 0$, let $t_0$ be the first time such that $p_{t_0}([P]) \ge 1$. By \cref{clm:growth-bound}, we have $p_{t_0}([P]) < 24ML$. For every $t_0 \le t < 2E$ with $p_t([P]) \ge 1$, the same claim gives
  \[
    p_{t+1}([P]) \le 12ML \cdot (p_t([P]) + 1) \le 24ML \cdot p_t([P]).
  \]
  Therefore it takes $\Omega(\log_{ML} P)$ timesteps for $p_t([P])$ to grow from at most $24ML$ to at least $P/2$ as in \eqref{eq:merged1-p2e}. Hence $E = \Omega(\log_{ML} P)$, which proves the theorem.
\end{proof}

\section{Applications}
\label{sec:applications}

In this section, we show a general composition theorem that allows us to compose deterministic algorithms in the \emph{asynchronous shared-memory model} (without contention)
with primitives for the stochastic CRQW model,
to obtain corresponding algorithms in the stochastic CRQW model.
This is essentially a transformation from the asynchronous model without stalls into the stochastic CRQW model, which preserves the runtime up to a logarithmic factor in $P$.

\begin{theorem}
  \label{thm:transformation}
  Let $\mathcal{A}$ be an abstract data type (ADT),
  and suppose $\mathcal{D}$ is a deterministic linearizable concurrent implementation of $\mathcal{A}$ satisfying:
  \begin{itemize}
    \item In the asynchronous shared-memory model (without stalls),
    every operation of $\mathcal{D}$ completes within $T \le P^{O(1)}$ steps of the invoking process in the worst case; and
    \item $\mathcal{D}$ uses a total of $M \le P^{O(1)}$ shared read/write or CAS registers,
    each of $\l$ bits,
    and also $M_{\textup{local}}$ words of local memory for each operation.
  \end{itemize}
  Then there is a concurrent implementation $\mathcal{D}'$ such that
  \begin{itemize}
    \item $\mathcal{D}'$ is a {linearizable} and {lock-free} implementation of $\mathcal{A}$;
    \item In the stochastic CRQW model,
    for every fixed time $t \ge 0$, each operation of $\mathcal{D}'$
    invoked at time $t$ completes within $O(T \log P)$ time with high probability in $P$;
    \item The data structure uses $O(M)$ words of shared memory with word size $w \ge \max\BK{\l + \eps \log P, \; 2 \log P}$ bits, where $\eps > 0$ is a predetermined constant (that does not depend on $\mathcal{A}$ or $\mathcal{D}$), and $O(M_{\textup{local}})$ words of local memory for each operation.
  \end{itemize}
\end{theorem}

\begin{proof}
  The new implementation $\mathcal{D}'$ simulates $\mathcal{D}$,
  using the algorithms from \cref{sec:register} and \cref{sec:cas-register}
  in place of atomic read/write and CAS registers (respectively).
  Linearizability follows by composition~\cite{HW90}.%
  \footnote{Although our CAS implementation is not \emph{strongly linearizable}, this does not affect the correctness of our composition: we rely on the \emph{worst-case} properties of the deterministic outer algorithm $\mathcal{D}$, not its probabilistic behavior.}

  \paragraph{High-Probability Latency.}

  From now on, we use \defn{high-level operations} to refer to operations in $\mathcal{D}$ and $\mathcal{D}'$, and use \defn{low-level steps} to refer to invocations of register and CAS register operations---including \Read, \Write, \opCAS---and local instructions. The \defn{latency} of a low-level step is the time duration from the completion of the previous low-level step (or the invocation of the high-level operation) to the completion of the current low-level step. The latency of a high-level operation equals the sum of the latencies of the low-level steps it applies.

  Let $t$ be a fixed time, and let $q$ be an arbitrary operation of $\mathcal{D}'$ invoked at time $t$.
  By \cref{thm:register-main,thm:cas-register-main} (and a union bound), with high probability in $P$, all operations to all registers/CAS-registers invoked during $[t, \; t + T P]$ complete within $O(\log P)$ time after invocation, which implies that the latency of any low-level step of types \Read, \Write, and \opCAS is at most $O(\log P)$. Moreover, \cref{prop:process-schedule-prob} (and a union bound) states that, with high probability in $P$, any low-level step that is a local instruction and is invoked during $[t, \; t + T P]$ has latency at most $O(\log P)$. When both high-probability events occur, the latency of $q$ is at most $O(T \log P)$ deterministically. Therefore, $q$ has latency at most $O(T \log P)$ with high probability in $P$.

  \paragraph{Lock-Freedom.}
  Recall that our registers and CAS registers are lock-free,
  and that $\mathcal{D}$ is wait-free (as every high-level operation completes within at most $T$ low-level steps in the worst case).
  Suppose for the sake of contradiction that $\mathcal{D}'$ is \emph{not} a lock-free
  implementation of $\mathcal{A}$.
  Then there is an execution $\alpha$ where
  starting from some time $t \geq 0$,
  all processes take infinitely many steps in total,
  but no high-level operation of $\mathcal{D}'$ completes.
  Because every high-level operation
  returns within at most $T$ low-level steps of $\mathcal{D}$,
  each high-level operation must eventually become stuck inside a low-level step;
  in other words, starting from some time $t' \geq t$,
  there are pending invocations of \Read, \Write, \opCAS, and local instructions,
  and each process takes infinitely many steps,
  and yet no \Read, \Write, \opCAS, or local instruction returns,
  contradicting the lock-freedom of these operations.

  \paragraph{Memory Usage.}
  By \cref{thm:register-main,thm:cas-register-main}, to simulate the registers and CAS registers, the data structure $\mathcal{D}'$ uses $O(M)$ words of shared memory with word size $w \ge \max\BK{\l + \eps \log P, \; 2 \log P}$ bits, where $\eps > 0$ is a fixed constant. Each operation to the registers/CAS-registers uses at most $O(1)$ words of local memory, which will be released after the register/CAS-register operation completes (so we do not need to multiply by $M$).
  The data structure additionally uses $O(M_{\textup{local}})$ words of local memory for each operation, which is the same as in $\mathcal{D}$. Combining all this memory usage gives the desired bound.
\end{proof}

\paragraph{Applications.}
\cref{thm:transformation} is a general transformation that allows us to transform known results from the asynchronous shared-memory model (without contention) to results in the stochastic CRQW model. We can apply this transformation to known constructions to obtain, for example,
max registers and counters with high-probability latency $O( \min(\log v, \, P) \cdot \log P)$ and $O( \min(\log v \log P, \, P) \cdot \log P)$ respectively, using the construction of~\cite{AAC09};
load-linked/store-conditional (LL/SC) objects with high-probability latency $O(\log P)$,
using the construction of~\cite{BW20};
and fetch-and-increment objects with high-probability latency $O(\log^3 P)$,
using the construction of~\cite{ERW12}.%
\footnote{The fetch-and-increment construction of~\cite{ERW12} uses LL/SC objects, which can be further implemented from CAS registers in worst-case constant time~\cite{BW20}.}

\section{Experiments}
\label{sec:experiments}

To better understand how well our model predicts the performance of modern multicores, we implemented a set of benchmarks to measure the performance of loads, stores, and CAS operations under high contention.
We ran our experiments on a Dell server equipped with four Intel Xeon E7-8867 v4 processors, totaling 72 cores (144 SMT threads), each with four-channel DDR4-2400 memory (1TB total) and 45MB of L3 cache per socket. In the experiments, each thread runs a loop in which it accesses one of four locations, each on a different cache line. Each iteration picks one of them at random. We use four locations because on this machine, writes to a single location are combined within a thread. Each loop performs a fixed number ($10^8$) of operations, and we measure the running time. We consider four operation types: (1) only loads, (2) only stores, (3) load/modify/store, and (4) load/modify/CAS. All operations use C++ \texttt{std::atomic} operations with sequential consistency. We then vary the number of threads from 1 to 144. The results are shown in \cref{fig:fig-scaling}. Each experiment is run 10 times, and we report the mean. The bars on each point represent the standard deviation, which is too small to be visible on most points.

The experiments demonstrate that loads scale almost perfectly, indicating that read contention does not significantly affect performance. However, all the operation mixes that include an update (store or CAS) degrade linearly with the number of threads (the dashed line has slope 1). This indicates that updates suffer from contention and that the contention cost grows linearly with the number of threads. This is consistent with the CRQW model.

\section{Related Work}
\label{sec:related-work}

\subsection{Restricted Schedulers}
\label{sec:non_adversarial}

Many researchers have acknowledged the gap between fully synchronous models and fully asynchronous, adversarial models, and have developed models that try to bridge this gap. Some of these are deterministic, and some are randomized.

The semi-synchronous model places constant upper and lower bounds on
the time taken by a process between any two of its steps~\cite{AAT94,EHN12,DLS88,ADLS94,AM94,HK06,ELMS05,T07}. The bounds might be known~\cite{AAT94,EHN12} or unknown~\cite{AAT94,DLS88}. In the model, one can show, for example, that obstruction-free algorithms can be made wait-free~\cite{ELMS05}.
Hendler and Kutten~\cite{HK06} suggest a similar model placing a restriction on asynchronous models that prevents processes from stalling forever. In particular, they define the notion of $k$ synchrony, where a process cannot take $k+1$ steps without all other processes taking at least one step. Although these models allow some processes to run faster than others, they all assume a strict constant upper bound on the time between steps. In the known-delay setting, this allows one to wait some fixed time and then know that an operation that has started has now completed. Only Hendler and Kutten's work considers memory contention.

Alistarh, Censor-Hillel, and Shavit~\cite{ACS16} suggest a stochastic scheduling model in which on
each step every process has at least some $\theta > 0$ probability of being scheduled. This means that over time, every process is scheduled. A natural setting is with $\theta = 1/P$, which means that all processes are equally likely to be scheduled. They refer to this as the uniform random scheduler, and analyze algorithms under this scheduler. They show that with such a scheduler lock-free algorithms
are wait-free with probability 1 (every operation has probability 1 of eventually completing). They also
show tight bounds for certain types of lock-free algorithms. Refinements of the model and further
results were shown in later papers~\cite{ARP15,ART16}. A problem with the uniform random scheduler (or even deviations on it with different probabilities) is that it is inherently not a good parallel scheduler.
In particular in any history on $P$ processes, in expectation there will be $\Theta(\sqrt{P})$ steps between appearances of the same process---i.e., for some constant $c$ it is unlikely that a region of history larger than $c \sqrt{P}$ will have no repeats. Since there must be a timestep between the same processor running two instructions, this implies that the scheduler is only running $O(\sqrt{P})$ instructions on each timestep. This is not a very effective scheduler, since in practice (and also in our stochastic scheduler model) we can run up to $\Theta(P)$ instructions on each timestep.

\subsection{Accounting for Contention}

Several parallel shared-memory machine models have been proposed that account for memory contention. All of them separate local instructions from shared-memory instructions and charge work for contended memory instructions that is proportional to the number of instructions that contend on a memory location. They, however, differ along four dimensions.
First, some of them only consider updating memory instructions as contending, while other models also count reads. Second, they make different assumptions about the scheduler. Third, some require responses from a contended memory location to return in the same order as the requests (i.e., in FIFO order), others do not. Fourth, some allow pipelining of memory instructions (multiple instructions can be outstanding and waiting for a reply on a single process). The models also allow different sets of updating instructions (e.g., CAS, test-and-set, or possibly just write), but any of these instructions could be used
with any of the models so we do not consider this a notable distinction.

Gibbons, Matias, and Ramachandran suggested the Queue-Read Queue-Write PRAM machine model \cite{GibbonsMR96,GibbonsMR98a}. In the model, there are $P$ processors, and the processors proceed in bulk synchronous steps. In each step, each processor $i$ does the following:
(1) reads $r_i$ locations from shared memory, (2) performs $c_i$ local operations, and (3) writes $w_i$ locations to shared memory.
The \emph{maximum contention} $k$ of a step is the maximum over all shared memory
locations of the number of reads and writes from that location on that step. The reads only see values from the previous step. The cost of a step is $\max(k, \max_i (r_i+c_i+w_i))$. The \emph{time} of a QRQW PRAM algorithm is the sum of the cost of its steps, and
the work is the processor-time product. They show a variety of results in the model, and in particular upper bounds on several basic algorithms, and some lower bound simulation results. Along the four dimensions, the model assumes reads as well as updates contend, it uses a bulk synchronous scheduler, the order of responses is irrelevant since they all respond at the end of the round, and the model allows
all memory instructions within a round to be pipelined. They also mention the possibility of a CRQW variant that does not account for contention on reads, but their results are on the QRQW variant. Later, others consider some algorithms based on the CRQW model~\cite{SGBFG15}.

In later work, Gibbons, Matias, and Ramachandran describe an asynchronous version of the QRQW model~\cite{GibbonsMR98b}. In this version, there are no
bulk synchronous steps, and instead the processes proceed asynchronously.
For analyzing the correctness of algorithms,
they assume an adversarial scheduler can arbitrarily delay any processor.
However, to analyze cost, they assume discrete timesteps and that all processors are greedy, i.e., if they have an instruction to run, they will run it. Local instructions take a single timestep. The model assumes that there is a queue associated with each memory location corresponding to outstanding memory requests on that location. In each step, all memory requests made in that step are added to the back of the queues for the corresponding locations, ties broken adversarially. Then, for each nonempty queue, the model removes the request from the front of the queue and executes it. They show that many of their results from the synchronous QRQW carry over to this model. Along the four dimensions, they again assume reads as well as updates contend, the scheduler is as described (adversarial for correctness, greedy for time), the order of responses is FIFO due to the queues, and memory operations can be pipelined.

Dwork, Herlihy, and Waarts (DHW) independently describe a similar asynchronous model that accounts for memory
contention \cite{DHW97}, but with very different assumptions on the scheduler. In the model, instructions on shared memory are split into an invocation and a later response. As in the standard asynchronous model, they assume that all steps appear in a linear order in the history and that an adversarial scheduler picks the next instruction (or response). They define \emph{contention cost} of a memory operation (consisting of its invocation and response) as the number of responses at the same location that appear between the invocation and the response in the history. The memory operation is said to be \emph{stalled} (or \emph{pending}) during this time.
The \emph{overall cost} of a memory operation is one more than its contention cost, and the cost of all other operations is one. The total cost of a process is the sum of the overall costs of its instructions, and the execution cost (history) is the sum of the overall costs of the processors.
The \emph{contention of an algorithm} (protocol) on the $P$ processors is the worst-case contention cost for
any execution of that algorithm, divided by $P$.
In the model, they consider both randomized and deterministic algorithms (protocols) for some basic primitives,
such as consensus and mutual exclusion. Along the four dimensions, they assume reads as well as updates contend, the scheduler is fully adversarial, the order of responses is arbitrary (an invocation might never respond), and it does not permit pipelining. They use the model to show various lower bounds, but the adversary is so powerful in the model that it is not clear if these bounds are representative.

Ellen, Hendler, and Shavit~\cite{EHN12} use a variant of DHW's model that does not count contention on reads and is otherwise the same.
This model is similar to ours in its CRQW and no-pipelining aspects, while the main difference is that it assumes a fully-adversarial scheduler.

Hendler and Kutten~\cite{HK06} describe a model similar to DHW, but with a significantly more
constrained scheduler. In particular, they define the notion of $k$-synchrony, where a process can take $k+1$ steps without all other processes taking a step.
They show a lower bound on the number of stalls incurred
by a single instruction in any implementation of a large class of
objects.
Atalar et al.~\cite{ARP15} also address conflicts in shared memory and their
effect on algorithm performance. They consider a class of lock-free
algorithms whose operations do some local work and then repeat
a ``Read-CAS loop''. Their proposed model divides conflicts into hardware and logical conflicts.

Ben-David and Blelloch~\cite{BDB17} study randomized back-off and back-on protocols for read-modify-write operations in the context of a model with contention. In their model, the invocation of instructions is fully adversarial, but once a memory instruction is invoked, it is enqueued and responds in FIFO order (one response per timestep). Along the four dimensions, they assume that only updates contend, the scheduler is fully adversarial (oblivious), the order of responses is FIFO (and one response per timestep for each location), and that no pipelining is permitted. Their analysis is for a fixed batch of operations invoked simultaneously, with no new operations invoked afterward. In contrast, most work on concurrent data structures, including this paper, assumes that operations arrive asynchronously over time.

Contention has also been studied in many settings other than shared
memory. Bar-Yehuda, Goldreich, and Itai~\cite{BGI87} applied exponential back-off to the problem of broadcasting a message over
multi-hop radio networks with contention. Bender et al.~\cite{BFHKL05} consider a simple contended channel with
adversarial disruptions. They show that using a back-off/back-on
approach, in which probability of transmission is allowed to rise
in some cases, achieves better throughput. This is similar to the
approach taken by our adaptive probability protocol.

For some problems, most notably mutual exclusion, there is a specific notion of contention based on the concept of a RMR (remote-memory access register)~\cite{Cypher,YA95,AK01,HW11,GW12,GiakkoupisWo14}. In these models the contention is measured as the maximum number of memory operations on a (remote) shared memory location over all time. Such models are overly pessimistic since most of the time the contention could be low.

Several papers have shown experimentally that real parallel machines suffer considerably from memory contention and have developed approaches to alleviate such contention~\cite{acar2017contention,ShunBFG13,BlellochGMZ97}.

\subsection{Analysis of Concurrent Data Structures}

Several works bound the step complexity of various concurrent data structures.
Ellen et al.~\cite{EFHR14} describe a lock-free binary search tree in which the amortized step complexity per operation
is $O(h(\text{op}) + c(\text{op}))$ per operation, where $h(\text{op})$ is the height of the tree at the beginning of $\text{op}$ and $c(\text{op})$ is the maximum number of operations accessing the tree at any one time during $\text{op}$. Unfortunately, this bound is quite weak, since if $P$ processes are accessing the tree at the same time, the bound is $O(P)$ for each operation. Ellen, Ramachandran, and Woelfel~\cite{ERW12} present a wait-free implementation of a fetch-and-increment object that requires just $O(\log^2 P)$ step complexity per operation. This seems to be the first bound for a non-trivial operation on an object that is sublinear in $P$. The bound was later improved to $O(\log P)$~\cite{EW13}, which is optimal for the problem.
Jayanti and Tarjan~\cite{JT21} describe a wait-free union-find object with step complexity bounds that are logarithmic in size and processors. Asbell and Rupert~\cite{AR24} describe a wait-free deque with amortized $O(\log^2 P + \log n)$ step complexity. Since a stack is a special case of a deque, the result also applies to stacks.
We note that these are step-complexity bounds that do not account for delays caused by memory contention.

\section{Conclusion and Future Directions}
\label{sec:conclusion}

The performance of concurrent data structures is governed by how \emph{scalable} they are---to get any benefit out of using more computing units in parallel, each operation's runtime must scale sublinearly with the number of processes.
Achieving scalability has turned out to be difficult, especially in the presence of write contention, where~\cite{EHN12} showed it to be impossible under a fully-adversarial scheduler.
Fully-adversarial schedulers are widely acknowledged to be too pessimistic for many scenarios.
In this paper, we construct scalable shared-memory primitives to address write contention under the \emph{stochastic} scheduler which captures most sources of delays in real-world multicore systems.
Designing and analyzing data structures under non-fully-adversarial schedulers (e.g., stochastic schedulers) offers a path towards rigorous scalability guarantees.

We conclude with a few open directions:
\begin{itemize}
  \item Designing scalable algorithms for richer primitives, starting with
    fetch-and-add, and for more complex objects such as search trees and
    dictionaries. This remains wide open: fetch-and-add has no known
    scalable solution even ignoring contention.
  \item Extending the stochastic scheduler to capture effects it currently
    ignores, such as context switches and cache coherence protocols that could delay the readers when lots of writes contend.
  \item Improving the expected latency of the read/write and CAS registers.
    In particular, we conjecture that there exist solutions with $O(\log \log P)$ expected latency at the cost of $P^{\epsilon}$ high-probability latency.
\end{itemize}

\section*{Acknowledgements}
Michael A.\ Bender was supported in part by NSF grant CCF-2247577, Sandia
National Laboratories, and Broadcom.
Guy E.\ Blelloch was supported in part by NSF grant CCF-2119352.
Rotem Oshman was supported in part by the Binational Science Foundation (BSF)
under Grant No.~2018356.
Renfei Zhou was supported in part by a Jane Street Graduate Research Fellowship
and a MongoDB PhD Fellowship.
This material is based upon work performed while attending the AlgoPARC Workshop
on Parallel Algorithms and Data Structures at the University of Hawaii at Manoa,
supported in part by the National Science Foundation under Grant No.~2452276.
This work has also benefited from discussions at Dagstuhl Seminar 25191,
Adaptive and Scalable Data Structures (\url{https://www.dagstuhl.de/25191}).

\bibliographystyle{alpha}
\bibliography{ref}

\end{document}